\documentclass[11pt,letterpaper]{article}

\usepackage[dvipsnames]{xcolor}

\usepackage{amsfonts}
\usepackage{amsmath}
\usepackage{amssymb}
\usepackage{amsthm}
\usepackage[appendix=append]{apxproof}
\usepackage{array}
\usepackage[blocks]{authblk}

\usepackage{bm}

\usepackage[disableredefinitions]{complexity}
\usepackage{csquotes}
\usepackage{enumitem}
\usepackage[T1]{fontenc}
\usepackage{geometry}
\usepackage{graphicx}

\usepackage{hyperref}
\usepackage[capitalise]{cleveref}

\usepackage{import}
\usepackage[utf8]{inputenc}

\usepackage{lineno}

\usepackage{mathtools}
\usepackage{mathrsfs}

	\usepackage{needspace}
	\usepackage{relsize}
\usepackage{setspace}
\usepackage{soul}
\usepackage{subcaption}

\usepackage{thm-restate}
\usepackage{thmtools}
\usepackage{tikz}
\usetikzlibrary{decorations.pathreplacing}
\usepackage{tikz-cd}
\usepackage{titling}
\usepackage{todonotes}

\usepackage{xspace}

\newcommand{\bbN}{\mathbb{N}}

\newcommand{\bbR}{\mathbb{R}}

\newcommand{\bbZ}{\mathbb{Z}}

\newcommand{\bfA}{\mathbf{A}}

\newcommand{\bfE}{\mathbf{E}}

\newcommand{\bfI}{\mathbf{I}}
\newcommand{\bfJ}{\mathbf{J}}

\newcommand{\bfL}{\mathbf{L}}

\newcommand{\bfN}{\mathbf{N}}

\newcommand{\bfQ}{\mathbf{Q}}

\newcommand{\bfV}{\mathbf{V}}

\newcommand{\calA}{\mathcal{A}}
\newcommand{\calB}{\mathcal{B}}
\newcommand{\calC}{\mathcal{C}}

\newcommand{\calF}{\mathcal{F}}
\newcommand{\calG}{\mathcal{G}}

\newcommand{\calJ}{\mathcal{J}}

\newcommand{\calV}{\mathcal{V}}
\newcommand{\calW}{\mathcal{W}}

\newcommand{\rmC}{\mathrm{C}}

\newcommand\sg[1]{}
\newcommand\sgchanged[1]{}
\newcommand\am[1]{}
\newcommand\amchanged[1]{}

\renewcommand{\flat}{\mathit{flat}}
\newcommand{\reduce}{\leq}
\newcommand{\depth}{\mathit{depth}}
\newcommand{\size}{\mathit{size}}

\renewcommand{\SIZE}{\text{\normalfont\scshape Size}}
\newcommand{\SIZEDEPTH}{\text{\normalfont\scshape SizeDepth}}
\newcommand{\WIRESIZEDEPTH}{\text{\normalfont\scshape WSizeDepth}}

\renewcommand{\REG}{\text{\normalfont\scshape{Reg}}}

\newcommand{\lm}{\mathit{lm}}

\newcommand{\scrC}{\mathscr{C}}

\newcommand{\N}{\mathbb{N}}

\newcommand{\bfone}{\mathbf{1}}

\newcommand{\bfMod}{\mathbf{Mod}}
\newcommand{\ESL}{\bfE\bfJ_\bfone}
\newcommand{\QESL}{\bfQ\bfE\bfJ_{\bfone}}
\newcommand{\ACom}{\mathbf{ACom}}
\newcommand{\EACom}{\bfE\ACom}
\newcommand{\QEACom}{\bfQ\EACom}

\newcommand{\bfJone}{\mathbf{J_1}}

\newcommand{\rmFO}{\mathrm{FO}}

\newcommand{\resp}{\ensuremath{\textit{resp.}}}

\DeclareMathOperator{\cont}{cont} 

\newcommand{\problemx}[3]{
\par\noindent\underline{\sc#1}\par\nobreak\vskip.2\baselineskip
\begingroup\clubpenalty10000\widowpenalty10000
\setbox0\hbox{\bf INPUT:\ }\setbox1\hbox{\bf QUESTION:\ }
\dimen0=\wd0\ifnum\wd1>\dimen0\dimen0=\wd1\fi
\vskip-\parskip\noindent
\hbox to\dimen0{\box0\hfil}\hangindent\dimen0\hangafter1\ignorespaces#2\par
\vskip-\parskip\noindent
\hbox to\dimen0{\box1\hfil}\hangindent\dimen0\hangafter1\ignorespaces#3\par
\endgroup}

\theoremstyle{plain}
\newtheorem{theorem}{Theorem}[section]

\newtheorem{lemma}[theorem]{Lemma}
\newtheoremrep{lemma}[theorem]{Lemma}

\newtheorem{fact}[theorem]{Fact}

\newtheorem{corollary}[theorem]{Corollary}
\newtheoremrep{corollary}[theorem]{Corollary}

\newtheorem{remark}[theorem]{Remark}

\newtheorem{definition}[theorem]{Definition}
\newtheorem{proposition}[theorem]{Proposition}
\newtheoremrep{proposition}[theorem]{Proposition}

\theoremstyle{definition}
\newtheorem{example}[theorem]{Example}
\Crefname{lemma}{Lemma}{Lemmata}
\Crefname{equation}{Line}{Lines}\crefrangelabelformat{equation}{(#3#1#4)--(#5#2#6)}
\Crefname{inequality}{Inequality}{Inequalities}\creflabelformat{inequality}{(#2#1#3)}
\Crefname{equivalence}{Equivalence}{Equivalences}\creflabelformat{equivalence}{(#2#1#3)}
\Crefname{property}{Property}{Properties}\creflabelformat{property}{(#2#1#3)}

\newcommand{\defeq}{\stackrel{\text{def}}{=}}

	\newcommand{\stab}{\text{stab}}

\newif\ifdraft
\drafttrue
\ifdraft
\renewcommand\sg[1]{\todo[inline,size=\scriptsize,backgroundcolor=Yellow]{#1 - \textbf{Stefan}}}
\renewcommand\sgchanged[1]{{\color{red}{#1}}}
\renewcommand\am[1]{\todo[inline,size=\scriptsize,backgroundcolor=Magenta]{#1 - \textbf{Amal}}}
\renewcommand\amchanged[1]{{\color{blue}{#1}}}

\newcommand{\bool}{\textrm{Bool}}

\newcommand{\block}{\mathit{block}}
\newcommand{\rmmod}{\mathrm{mod}}

\newcommand{\ACO}{\mathbf{AC^0}}

\newcommand{\congC}{\mathbin{\mathscr{C}}}

\DeclareMathOperator{\opc}{c}

\setlist[enumerate]{label=(\arabic*),font=\normalfont}

\newcommand{\dotminus}{%
  \mathbin{\ooalign{%
    \hfil$-$\hfil\cr
    \hfil\raisebox{0.35ex}{$\scriptscriptstyle\bullet$}\hfil\cr
  }}%
}

\newcommand{\Cont}{\mathrm{Cont}}

\usetikzlibrary{positioning}

\title{Rational Reductions and\\ Regular Languages of
Constant Circuit Complexity}

\newif\ifanonym
\anonymtrue
\anonymfalse
\ifanonym
\author{}
\else
\author{
    Stefan G\"oller\\
\small{\url{stefan.goeller@uni-kassel.de}}}
\affil{School of Electrical Engineering and Computer Science, Universit\"at Kassel, Germany}
\author{Amaldev Manuel\\\small{\url{amal@iitgoa.ac.in}}}
\affil{School of Mathematics and Computer Science, Indian Institute of Technology Goa, India}
\fi

\begin{document}

\maketitle
\begin{abstract}
We study the circuit complexity of regular languages
in terms of unbounded fan-in Boolean circuit families.
We characterize the regular languages of constant
circuit complexity in terms of the one-variable
	fragment of first-order logic with
	regular predicates,
	in terms of the pseudovariety of stamps
$\mathbf{QEJ}_\mathbf{1}$, suitable word
congruences and regular expressions.
We analogously characterize the neutral letter regular 
	languages of constant circuit complexity.
Our lower bound result implies that the class of 
regular languages of sublogarithmic
circuit complexity coincides with
the one of constant circuit complexity.
In addition we show that deciding whether a regular
	language, given as a nondeterministic
	finite automaton, has constant
	circuit complexity is 
	$\mathbf{PSPACE}$-complete.

We introduce a strong notion of
	reduction, 
called rational truth-table reduction,
that is tailored towards
	algebraically defined classes of languages.
	We show that, for a class of functions we 
	call mild, rational truth-table reductions 
	preserve both upper and 
	lower bounds on circuit complexity.
We show that the class of regular languages,
whose circuit complexity is bounded by a mild
function, is in fact a length-multiplying variety of languages.

Slightly extending the class of regular languages
	of constant circuit complexity, we analogously characterize
	the class of regular languages that
	are in the pseudovariety $\mathbf{QEACom}$.
	For these we derive logarithmic circuit complexity upper bounds.
\end{abstract}

\newpage
\tableofcontents
\newpage

\newif\ifshowproofs
\showproofsfalse
\showproofstrue

\ifshowproofs
\else
  \RenewEnviron{proof}[1][]{}
\fi

\newif\ificalp

 \icalpfalse  

\setcounter{tocdepth}{4}
\nolinenumbers

\section{Introduction}

We study the minimum asymptotic sizes of Boolean circuits of unbounded 
fan-in computing regular languages. 
This measure is often termed the circuit complexity of a language. 
While circuit complexity theory typically considers classes defined by both size and depth, circuit size is the more elementary measure. Depth measures parallel execution time, but the most fundamental question about a language is simply its total computational cost, i.e., the minimal size of a circuit that computes it. Our primary concerns are questions such as: 
what are the various circuit complexity classes of regular languages? Can we separate them? Is it possible to characterize them in terms of algebra, logic and regular expressions? Is it possible to decide the circuit complexity of a given regular language?

Circuit complexity of regular languages is a well established area of research as evidenced by the chapter with the same title by Kouck\'y \cite{Kouckyhandbook}  in the \emph{Handbook of Automata Theory} \cite{hanbookauto}. 
Paraphrasing~\cite{Kouckyhandbook} Boolean circuits are a foundational 
model of parallel computation, much like finite state automata are 
for sequential computation. 
Relating these two models sheds light on the possible 
parallelization of regular language recognition, and more 
generally the relationship between sequential and parallel computation. 
This is particularly relevant as modern computing systems increasingly 
employ parallelism for low-level operations.

There are technical motivations from circuit complexity theory as well: regular languages play a key role in separating complexity classes within $\mathbf{NC}^1$, i.e., the class of languages computed by polynomial-sized Boolean circuits with bounded fan-in and logarithmic depth. 
In particular, they separate $\mathbf{AC}^0$ and $\mathbf{NC}^1$~\cite{FurstSaxeSipser1984}. In fact, circuit classes defined by programs over pseudovarieties of monoids 
\cite{Barrington1989,BarringtonTherien1988} are always separated by regular languages if they are distinct \cite{McKenzieNC1} . 
Furthermore, any separation between the circuit classes $\mathbf{ACC}^0$ and $\mathbf{NC}^1$
and between $\mathbf{TC}^0$ and $\mathbf{NC}^1$ can be witnessed by a regular 
language~\cite{Koucky2009} respectively.
These results reveal the importance of regular languages within 
circuit complexity theory.

Much is known about regular languages and the standard circuit complexity classes (see \S\,\ref{subsec:related}). However, the field remains poorly understood and in particular the questions mentioned in the beginning are yet to be answered. 

In this paper we characterize the regular languages of constant circuit complexity.
It turns out that it contains several well-known classes from the literature. 
Some are listed below. Their containment is easy to verify. 
\begin{itemize}
   
    \item \emph{Finite and co-finite languages}: They correspond to the pseudovariety $\bfN$ of nilpotent semigroups ($x^\omega y = y x^\omega = x^\omega$).
    \item \emph {Idempotent and commutative languages}: These are finite Boolean 
	    combinations of languages of the form $\Gamma^*$, or equivalently 
		$\Sigma^*a\Sigma^*$,  for $\Gamma\subseteq \Sigma$ and $a\in \Sigma$. 
		They correspond to the pseudovariety $\bfJone$ of idempotent and commutative semigroups, also known as semi-lattices   $(xy=yx,x^2=x)$.
\item \emph{Thin regular languages}: Regular languages that have at most one word of
	each length,  e.g.~$(ab)^+$. Recall that every unary regular language is thin.

\item \emph{Slender regular languages}: A regular language $L$ 
	is \emph{slender} if 
	there is a $k\in \bbN_{>0}$ such that for each $n\in\bbN_{>0}$ 
		there are at most $k$ words of length $n$ in $L$. 
		A classical characterization theorem states that a regular language 
		is slender, if and only if, it is a finite union of languages of 
		the form $xu^*y$ for $x,u,y\in \Sigma^*$, we refer 
		to~\cite{Pin2025Mathematical}.
    
\item \emph{Generalized definite languages}: They are finite Boolean combinations of 
	languages of the form $u\Sigma^*$ and $\Sigma^*u$ for $u\in \Sigma^*$. 
		They correspond to the pseudovariety $\bfL\bfI$ of locally trivial semigroups ($x^\omega y x^\omega = x^\omega$).
\end{itemize}

However there are languages of constant circuit complexity that do not belong to any
of the above classes or finite Boolean combinations thereof, 
for example the regular language~$a(b+c)^*a$.

 We use the framework of $\calC$-varieties of Straubing~\cite{Straubing02,StraubingPin10}, i.e., length-multiplying pseudo-varieties of surjective morphisms onto finite semigroups (lm-variety of stamps) in our particular context, to 
 characterize the complexity classes of regular languages.

\subsection{Our contribution}

\paragraph{Circuit Model and Input Encodings}
The {\em circuit complexity of a language $L\subseteq\Sigma^+$} is defined
to be the function $\opc(L):\bbN_{>0}\to\bbN_{>0}$
that assigns to each input length $n\in\bbN_{>0}$ 
the size (the number of non-input gates)
of the smallest unbounded fan-in circuit computing $L\cap\Sigma^n$.

Our circuit model employs a subset encoding for the inputs following
Kouck\'y~\cite{Kouckyhandbook}. 
Since most circuit complexity classes allow for polynomial or larger sized
circuits the input encodings are usually assumed to be binary.
However for sublinear size circuits the input encoding becomes crucial.
We discuss this in detail and show the following, where for an alphabet 
$\Sigma=\{a_0,\dots,a_{k-1}\}$ 
of size $k$ the {\em one-hot
encoding} is the morphism $\varphi:\Sigma^+\to\{0,1\}^+$, where
$\varphi(a_i)=0^i10^{k-i-1}$ for each $i\in[0,k-1]$.
In contrast, a \emph{one-hot circuit family with promise computing
a language $L\subseteq\Sigma^+$}
is a circuit family $\calC$ over the input alphabet $\{0,1\}$ such that 
$\varphi^{-1}(L(\calC))=L$.

\begin{samepage}
	\begin{itemize}
\fi

\item (Proposition~\ref{P with promise})
Our subset encoding model of Definition~\ref{D Circuit} is equivalent to the
	one-hot encoding with promise in the following sense.
For all languages $L\subseteq\Sigma^+$ and all functions $f:\bbN_{>0}\to\bbN_{>0}$
the following two statements are equivalent:
	\begin{enumerate}
		\item There is a one-hot circuit family with promise $\calC$ computing $L$,
			where $|\calC|(kn)=f(n)$ for all $n\in\bbN_{>0}$.
		\item There is a circuit family $\calC$ computing $L$,
			where $|\calC|(n)=f(n)$ for all $n\in\bbN_{>0}$.
	\end{enumerate}
		\item (Proposition~\ref{P encoding}) 
	Let $\Sigma=\{a,b\}$ and let $L=\Sigma^+$. Then the following holds:
	\begin{enumerate}
		\item $\opc(\varphi(L))\in\Omega(\log n)$, whereas 
			$\opc(\varphi(\Sigma^+\setminus L))\in O(1)$.
\item Let $\alpha:\Sigma^+\to\{b\}^+$ denote the length-multiplying
	morphism, where $\alpha(a)=\alpha(b)=b$.
			Then $\opc(\varphi(\{b\}^+))\in O(1)$,
			whereas 
			$\opc(\varphi(\alpha^{-1}(\{b\}^+)))\in \Omega(\log n)$.
	\end{enumerate}
	Thus, the class of regular languages of constant circuit
	complexity under the one-hot encoding (without promise) is neither closed 
	under complement
	nor under inverse length-multiplying morphisms, and 
	hence is in particular
	not an lm-variety of languages (see Section~\ref{S Circuit LM}
	for more details).

	\end{itemize}
\end{samepage}

\paragraph{Rational Truth-Table Reductions}
We introduce the following notion of reduction (Definition~\ref{D Rational}), a particularly 
strong variant of constant depth reductions~\cite{ConstantDepth1984}.
Let $L_1\subseteq \Sigma_1^+$ and $L_2\subseteq \Sigma_2^+$ be languages.
	The language {\em $L_1$ rationally truth-table reduces to $L_2$}, 
$L_1\reduce L_2$ in notation, if 
there exists a length-multiplying morphism 
$\varphi:\Sigma_1^+\rightarrow \Sigma_2^+$ such that
$L_1$ is a Boolean combination over
a finite set of languages all of the form $\varphi^{-1}(u^{-1}L_2v^{-1})$,
where $u$ and $v$ range over $\Sigma_2^*$.

We show that rational truth-table reductions are well-behaved.
We introduce the class of mild functions (Definition~\ref{D mild}).
	A function $f:\bbN_{>0}\to\bbR_{\geq 0}$ is {\em mild} if 
	$f$ is non-decreasing, eventually positive, and satisfies 
	$f(n)=\Theta(f(\lceil\varepsilon n\rceil))$ for every $\varepsilon\in\bbR_{>0}$.
For instance, every non-decreasing and eventually positive polynomial is mild, as well as
all the functions $n\mapsto \log^k n$ for all
$k\in\N_{>0}$. Note that $n\mapsto 2^n$ is not mild.

	\begin{itemize}
		\item (Proposition~\ref{P reduce transitivity})
If $L_1\reduce L_2$ and $L_2\reduce L_3$, then $L_1\reduce L_3$.
\item 
	(Theorem~\ref{T asymptotics}) Let $f:\bbN_{>0}\to\bbR_{\geq 0}$ be a mild function
	and $L_1$ and $L_2$ be languages such that $L_1\leq L_2$.
			We mean Hardy-Littlewood $\Omega$ notation (see
			also Definition~\ref{D Hardy}).
	\begin{enumerate}

		\item If $\opc(L_2) \in O(f)$, then $\opc(L_1) \in O(f)$.
		\item If $\opc(L_1) \in \Omega(f)$, then $\opc(L_2) \in \Omega(f)$.
	\end{enumerate}
\item (Theorem~\ref{T mild variety})
	If $f,g:\bbN_{>0}\to\bbR_{\geq 0}$ are mild functions, then 
	\begin{enumerate}
		\item $\SIZE(O(f))\cap\REG$ and 
		\item $\SIZEDEPTH(O(f),O(g))\cap\REG$
	\end{enumerate}
are $\lm$-varieties of languages
			(see \S\,\ref{S Circuit Complexity} and~\S\,\ref{S Circuit LM} 
			for more details). 
	\end{itemize}
\begin{samepage}
\paragraph{$O(1)$-Characterization Theorem}
We refer to Theorem~\ref{T const main} for a more detailed statement.
		Let $L\subseteq\Sigma^+$ be a regular language. Then
	the following statements are equivalent:
	\begin{enumerate}
		\item $\opc(L)\in O(1)$.
		\item $\opc(L)\in o(\log n)$.
		\item The syntactic morphism $\eta_L$ belongs to $\QESL$, i.e.
			the stable semigroup $\stab(\eta_L)$ of $\eta_L$ satisfies the identities
			$epxqf=epxxqf$ and $epxyqf=epyxqf$, 
			where $e,f$ range over idempotents 
			and where $p,q,x,y$ range over all elements of $\stab(\eta_L)$ (see \S\,\ref{sec:qesl} for more details).
		\item $L$ is a union of $\approx_k$-classes for some
			$k\in\bbN_{>0}$, where
			$\approx_k$ denotes the following congruence on $\Sigma^+$ (see \S\,\ref{sec:approx_k} for more details):
			$u\approx_k v$ if, and only if, $u=v$ or
			\begin{itemize}
				\item $|u|,|v|>2k$,
				\item $u=u'xu''$ and $v=v'yv''$ have the same 
					$k$-length prefix $u'=v'$ and 
			the same $k$-length suffix $u''=v''$, 
		\item $|x|\equiv|y|\text{ mod }k$, and
		\item $x$ and $y$ agree on the set of letters appearing
			in positions that are equivalent mod $k$.
			\end{itemize}
		\item $L$ is a finite Boolean combination of finite languages and languages
			 $u(\Gamma_0\dots \Gamma_{k-1})^+v$, where
			$k\in\bbN_{>0}$, $\emptyset\not=\Gamma_0,\dots,\Gamma_{k-1}\subseteq\Sigma$
			and $u,v\in\Sigma^+$.
		\item $L$ is definable in 
			$\rmFO^1[\Sigma, \mathrm{arb}]$
			(see \S\,\ref{S FO} for more details),
		more precisely in the one-variable fragment of first-order
			logic with arbitrary numerical predicates.
	\end{enumerate}
\end{samepage}

The heart of the proof is our lower bound proof, i.e. the direction 
(2) $\Longrightarrow$ (3), which is proven in~\S\,\ref{S Lower Bound const}.
Making use of the above-mentioned rational truth-table reductions
	we prove that any regular language $L$ for which 
	the stable semigroup of the syntactic morphism falsifies
	at least one of the identities $epxqf=epxxqf$ or $epxyqf=epyxqf$
	must satisfy $\opc(L)\in\Omega(\log n)$.
When the latter identity is falsified, logarithmic circuit complexity lower bounds 
	can actually be derived from logarithmic communication complexity lower 
	bounds under deterministic arbitrary-position partitions.~\cite{Hromkovic85a} 
	Yet there are languages, like $0^*10^*1(0+1)^*$, of constant
	communication complexity that require logarithmic circuit complexity lower bounds.
Our lower bound proof consists of one short argument that deals with the
	falsification of any of the two identities in
one shot: we show that any circuit family of sufficiently small size 
for a language for which the stable
semigroup of the syntactic morphism violates at least one of the two identities
must treat a sufficient number of positions equally, thus enabling us to swap
letters at those positions without changing the behavior of the circuit, 
yielding a contradiction when applied to languages where such a swap
toggles membership.
The lower bound quite immediately implies the above-mentioned complexity
jump between the one-hot encoding (without promise) 
of the languages like $\{a,b\}^+$ and its complement $\emptyset$, respectively.
	\paragraph{$O(1)$-Characterization Theorem For Neutral Letter Regular
	Languages}
The analogous statement for neutral letter languages 
	reads (Corollary~\ref{C Neutral Letter}): 
	Let $L\subseteq \Sigma^+$ be a neutral letter language 
	(see~\S\,\ref{neutral letter} for more details). Then
	the following statements are equivalent:
    \begin{enumerate}
	    \item $\opc(L)\in O(1)$.
        \item  $L$ is idempotent and commutative.
	\item $L$ is alphabetic, i.e.~membership of a word $w$ in $L$
		  only depends upon the letters in $w$.
        \item $L$ is a finite Boolean combination of languages of the form $\Gamma^*$ with $\Gamma\subseteq \Sigma$.
    \item $L$ is definable in $\rmFO^1[\Sigma]$ (see
    \S\,\ref{S FO} for more details). 
    \end{enumerate}

The characterization of the neutral letter regular languages allows for a
short $\mathbf{PSPACE}$-hardness proof of the following result.

	\paragraph{Decidability of $O(1)$-Membership}
(Theorem~\ref{T constant PSPACE})
	Given a regular language $L$ as a nondeterministic finite automaton, 
	the problem of deciding if
	$\opc(L)\in O(1)$ is $\mathbf{PSPACE}$-complete.

	The $\mathbf{PSPACE}$ upper bound is obtained by an on the fly algorithm
	that verifies the identities of $\QESL$ (Point (3) of the above-mentioned $O(1)$-Characterization Theorem) via representatives of the stable semigroup of the syntactic morphisms
	via suitably computed transition matrices of the nondeterministic finite automaton.

	\paragraph{Characterization and Complexity of $\QEACom$}
	Let $k,t\in \bbN_{>0}$. 
	For each vector $\nu = (\nu_{i,a})_{i \in [0,k-1], a \in \Sigma}$, 
	where $\nu_{i,a} \in [0,t]$ we define the language
	\begin{align}
    \mathrm{Count}_k^t(\nu)&=\{ w \in \Sigma^{k\N_{>0}} : 
		\left|\pi_{i,k}(w)\right|_a \equiv^t \nu_{i,a}\text{ for all
		$i\in[0,k-1]$ and all $a\in\Sigma$}\}.
	\end{align}
	For each $i\in [0,k-1]$ and $a \in \Sigma$ 
	the language $\mathrm{Thr}_{i,k}(a,t)$ is defined to be
    \begin{align}
        \mathrm{Thr}_{i,k}(a, t) &= 
	    \{ w \in \Sigma^{k\bbN_{>0}} : \left|\pi_{i,k}(w)\right|_a \geq t \}.
    \end{align}

    We refer to Theorem~\ref{T QEAComm} for the following
    characterization.
	Let $L\subseteq\Sigma^+$ be a regular language. 
		Then the following statements are equivalent:
	\begin{enumerate}
		\item $\eta_L$ belongs to $\QEACom$,
			i.e. the stable semigroup $\stab(\eta_L)$ of $\eta_{L}$
			satisfies  the identities
$epx^{\omega+1} qf=epx^{\omega}qf$ and 
$epxyqf=epyxqf$, where $e,f$ range over idempotents and where
			$p,q,x,y$ range over all elements of $\stab(\eta_L)$.
		\item $L$ is a union of $\approx_k^t$-classes
			for some $k,t\in\bbN_{>0}$,
			where
$\approx_k^t$ denotes the following congruence on $\Sigma^+$ (see \S\,\ref{sec:approx_k} for more details):
			$u\approx_k^t v$ if, and only if, $u=v$ or
			\begin{itemize}
				\item $|u|,|v|>2k$,
				\item $u=u'xu''$ and $v=v'yv''$ have the same 
					$k$-length prefix $u'=v'$ and 
			the same $k$-length suffix $u''=v''$, 
		\item $|x|\equiv|y|\text{ mod }k$, and
		\item for all residue classes $i\in[0,k-1]$, up to threshold $t$,
			the words $x$ and $y$ have the same number 
			of occurrences of letters appearing
			in positions equivalent to $i \bmod k$.
			\end{itemize}
        \item $L$ is a disjoint finite union of finite languages and languages of the 
		form $u\mathrm{Count}_k^t(\nu)v$, where $k,t\in\bbN_{>0}$, 
			$\nu \in [0,t]^{[0,k-1]\times \Sigma}$ and $u ,v\in \Sigma^+$.
		\item $L$ is a finite Boolean combination of finite languages and languages
			of the form
			$u\mathrm{Thr}_{i,k}(a,t)v$,
			where $k,t\in\bbN_{>0}$, $i \in [0,k-1]$,
            $a \in\Sigma$, and $u,v\in\Sigma^+$.
  \item $L$ is definable in 
			$\rmC^1[\Sigma, \mathrm{arb}]$,
			where $\rmC^1$ denotes the extension of $\rmFO^1$ with counting
			quantifiers of the form $\exists^{\geq t}$.
    \end{enumerate}
The following circuit complexity upper bound is immediate from~\cite{HASTAD1994200,Friedman86}.
    \begin{itemize}
	    \item (Proposition~\ref{P QEACom log}) If $\eta_L \in \QEACom$, then $\opc(L) \in O(\log n)$.
    \end{itemize}

\subsection{Related work}
\label{subsec:related}

\paragraph{Bounds on Circuit Size}
All languages including noncomputable ones have exponential sized circuit
families and almost all languages require exponential sized 
circuits~\cite{AroraBarak}. 
No explicit language is known that requires super-linear circuit complexity. 
Hromkovič proved in~\cite{Hromkovic85a} that certain explicit Boolean 
functions require 
a linear number of gates in the unrestricted unbounded fan-in Boolean circuit 
model, thus establishing one of the earliest linear lower bounds for general
unbounded fan-in circuits.

Coming to regular languages, every regular language is computed by a family of linear size and logarithmic depth, i.e., all regular languages are in $\mathbf{NC}^1$. 
Bounds are known for circuits computing Boolean operations and concatenations of  regular languages in terms of the circuit complexity of the constituent languages. Likewise, there are theorems relating circuits computing membership in a regular language, computing the image under its syntactic morphism, and the word problem for the syntactic monoid (\cite{Koucky2009} Proposition 3.1). If the product over the syntactic monoid of $L$ is computable by a polynomial-size Boolean circuit over arbitrary gates, then for every $\varepsilon>0$, there is a family of circuits  of size
$O(n^{1+\varepsilon})$ for $L$ over the same set of gates.

\paragraph{Length-Multiplying Pseudovarieties}

Schützenberger proved that a regular language is star-free if, and only if,
its syntactic monoid is aperiodic (i.e., in the pseudovariety $\bfA$), providing the fundamental algebraic characterization of star-free languages~\cite{Schutz65}.
McNaughton and Papert~\cite{McNPap71}  proved that the star-free languages are exactly those definable in first-order logic \(\rmFO[\Sigma, <]\), thereby connecting automata, algebra, and logic.
In their seminal work~\cite{BCST92}, 
Barrington, Compton, Straubing, and Th{\'e}rien proved the analogous result 
that regular languages in $\ACO$ are precisely those whose stable 
semigroups are aperiodic, and equivalently definable in $\rmFO[\Sigma,\mathrm{arb}]$.

The notion of $\calC$-variety, of which lm-variety is a particular instance, was introduced 
by Straubing in~\cite{Straubing02}. 
~\cite{ActionCvariety, StraubingPin10} extended Eilenberg's variety theory from monoids to general 
$\mathcal{C}$-varieties by proving a Reiterman-style identity theorem and by developing generalized notions of Mal'cev product, positive varieties, 
polynomial closure, and concatenations and wreath products,
thereby unifying several previous variants of algebraic language 
theory. An analogous equational framework based on 
generalized implicit operations was developed by Kunc~\cite{Kunc10} yielding an equational characterization of these pseudovarieties and, in particular, of the one corresponding to $\mathbf{AC}^0$.

Several key results already use or could be stated in terms of lm-varieties. The class of regular languages in $\mathbf{NC}^1$ corresponds to the lm-variety $\mathbf{Fin}$ of all stamps.
The characterization of regular languages in $\ACO$ by
Barrington, Compton, Th{é}rien and Straubing~\cite{BCST92} can be stated in the following way: the class of regular languages in $\ACO$ corresponds to the lm-variety $\bfQ\bfA$ of quasi-aperiodic stamps. Likewise, the regular languages in 
 $\mathbf{WLAC}^0$, i.e. those computed by constant-depth circuit families 
  with a linear number of  wires,
 have recently been characterized
 by Paperman and Cadilhac~\cite{CadilhacP22} as those languages whose
 syntactic morphism lies in the lm-variety $\mathbf{QLDA}$, and equivalently those definable in $\rmFO^2[\Sigma, <,+1, \mathrm{mod}]$. Apart from this, a number of classes of regular languages defined in terms of logical fragments, most notably low levels of alternation hierarchy, namely  $\mathcal{B}\Sigma_1$ and $\Sigma_2$ (\resp~$\Pi_2$) fragment of $\rmFO[<,\bmod]$ \cite{Chaubard}, and the two-variable fragment $\rmFO^2[\Sigma, <,\mathrm{mod}]$ \cite{Papermanfo2,KW15} are also characterized in terms of lm-varieties.

\paragraph{Programs over Monoids}
The surveys~\cite{BridgesTessonTherien, LogicMeetsAlgebra} and the monograph \cite{Straubingbook} 
discuss in detail the semigroup theory related to circuits. 
Following Barrington's theorem~\cite{Barrington1989}, there has been considerable interest in programs over monoids~\cite{BarringtonTherien1988,McKenzieNC1,PeladeauTCS} and in particular obtaining lower bounds using them. However, the framework turned out to be complex, and except for partial results no 
breakthroughs were achieved using it~\cite{GrosshansMS22}. There have been several works studying the power of programs over specific varieties~\cite{NathanJ,GrosshansMS22}.
Languages defined by polynomial-sized programs over monoid varieties 
constitute lm-varieties~\cite{McKenzieNC1}.
In particular the lm-variety $\QESL$ that we show to characterize 
$O(1)$ circuit complexity has been studied in the context of programs over monoids~\cite{NathanJ}.

\subsection{Organization of this paper}

In~\S\,\ref{S Prelim} we introduce general mathematical notation,
in particular words and languages, recognition by finite semigroups,
circuit complexity, and first-order logic over words.
Rational truth-table reductions and
length-multiplying varieties are introduced in~\S\,\ref{S Reductions}.
In~\S\,\ref{S Characterization} we concern ourselves with
the characterization of the regular languages of constant
circuit complexity (proof of Theorem~\ref{T const main}).
In~\S\,\ref{S PSPACE} we show that $O(1)$ membership
is $\mathbf{PSPACE}$-complete for nondeterministic finite automata.
In~\S\,\ref{S Discussion} we discuss the pseudovariety
$\QEACom$ and some subtleties on encodings of input gates.
We conclude in~\S\,\ref{S Conclusion}.
Some of the proofs can be found in the appendix.


\section{Preliminaries}\label{S Prelim}
By $\N=\{0,1,\dots\}$ we denote the non-negative integers and by 
$\N_{>0}=\{1,2,\dots\}$ the positive integers.
By $[i,j]$ we denote the interval $\{i,i+1,\dots,j\}$.
For each set $X$ we denote by $2^X$ the powerset of $X$.
For each $k\in\bbN_{>0}$ 
let $\equiv_k$ denote the congruence modulo $k$ on $\bbN$
and by $k\N=\{n\in\bbN\mid n\equiv_k0\}$
the set of all multiples of $k$. For each $t\in\bbN_{>0}$ 
let $\equiv^t$ denote the threshold congruence on $\bbN$ defined as 
    $m \equiv^t n$ if $m = n$  or  $m,n\geq t$.

\subsection{Words and Languages}
For a finite alphabet~$\Sigma$, let $\Sigma^*$ (\resp~$\Sigma^+)$ denote the set of all 
finite words (\resp~non-empty finite words) over~$\Sigma$.
Let $\varepsilon$ denote the {\em empty word} and let
$\cdot$ denote concatenation of words (that we mostly drop). Too,
we use $\prod$ in the context of concatenation of words.
The length of a word~$w$ is denoted by~$|w|$.
By $\Sigma^{\leq n}$ we denote the set of words of length at most $n$
and by $\Sigma^{s\N}$ (\resp~$\Sigma^{s\N_{>0}}$) the set 
of (\resp~non-empty) words whose length is a multiple of $s$.
For each letter $a\in\Sigma$ we denote by $|w|_a$ the number of
occurrences of the letter $a$ in $w$.

Throughout the paper we write words $w\in\Sigma^n$ as $w=w_0\dots w_{n-1}$,
where 
$0,\dots,n-1$ are {\em the positions of $w$} and
$w_0,\dots, w_{n-1}\in\Sigma$ are letters.
We also use the notation $(w)_i=w_i$ to denote the letter at position $i$ of $w$.
For a position $i\in[0,n-1]$ and a letter $a\in\Sigma$ we denote
by $w[a/w_i]=w_0\dots w_{i-1}aw_{i+1}\dots w_{n-1}$ the word
that is obtained from $w$ by replacing the letter at position $i$ by $a$.
The \emph{alphabetic content} of a word $w\in\Sigma^*$ is defined
as $\cont(w)=\{(w)_i\in\Sigma\mid i\in[0,|w|-1]\}$.

We say $u\in\Sigma^*$ is a \emph{prefix} (\resp~\emph{suffix}) of $v\in\Sigma^*$, 
if there exists $w\in \Sigma^*$ such that 
$v=uw$ (\resp~$v=wu$). More generally, $u$ is a \emph{factor} of $v$ if there 
exist $w,w'\in\Sigma^*$ such 
that $v=wuw'$.

If $u$ is a prefix of $v$, then the \emph{left quotient} of $v$ by $u$, denoted by~$u^{-1}v$, is the 
word~$w\in\Sigma^*$ such that $v=uw$.
Analogously the \emph{right quotient} of $v=wu$ by $u$ is the 
word~$vu^{-1}=w$. 
If $v=uxw$ then the \emph{two-sided quotient}~$u^{-1}vw^{-1} = x$ is well defined since $u^{-1}(vw^{-1}) = (u^{-1}v)w^{-1}$.

A {\em language} is a subset of $\Sigma^*$ and it is {\em positive} if
it is a subset of~$\Sigma^+$. 
The {\em positive left-quotient} of a language $L\subseteq\Sigma^+$ by a word $x\in\Sigma^*$
is the positive language $x^{-1}L = \{w \in \Sigma^+ \mid xw \in L\}$.
The {\em positive right-quotient} is defined analogously.
Unless stated otherwise we will be working with {\em positive} languages in the rest of this
paper.

\subsection{Recognition by Finite Semigroups}\label{S Semigroups}

A semigroup $(S,\cdot)$ is a set $S$ with a binary operation 
$\cdot : S \times S \rightarrow S$ that is associative.  
We often denote the semigroup operation by juxtaposition
and often write $S$ for $(S,\cdot)$.
The semigroup $(S,\cdot)$ is a {\em monoid} if it has a neutral element $1_S\in S$, i.e.
$s1_S=1_Ss=s$ for all $s\in S$. A monoid is a {\em group} if each element $s\in S$ 
has an inverse $s^{-1}\in S$ such that $ss^{-1}=s^{-1}s=1_S$.
An {\em idempotent} of a semigroup $S$ is an element $e\in S$ satisfying $e\cdot e=e$;
by $E(S)$ we denote the set of all idempotents of $S$.
In a finite semigroup $S$ every element $s$ has a unique idempotent
power $s^\omega\in E(S)$.
Moreover, there exists $n\geq 1$ such that 
$s^n$ is an idempotent for all $s\in S$, for instance
$n=|S|!$.
Any such $n$ is called an {\em exponent} of $S$.

Let $(S,\cdot)$ and $(T,+)$ be semigroups.  
A {\em (semigroup) morphism} from $S$ to $T$ is a map $\varphi:S\rightarrow T$ such that 
\begin{align}
\varphi(s_1\cdot s_2)=\varphi(s_1)+\varphi(s_2)
\quad \text{for all } s_1,s_2\in S.
\end{align}

Let~$\Sigma$ and~$\Gamma$ be finite alphabets. 
Any morphism 
$\varphi : \Sigma^+ \to \Gamma^+$ is completely specified by 
its restriction~$\varphi\restriction \Sigma : \Sigma \to \Gamma^+$. 
It is \emph{length-multiplying} if there exists $\ell\in\bbN_{>0}$ 
such that $\varphi(a) \in \Gamma^\ell$ for each~$a \in \Sigma$. 
We remark that if $\varphi:\Sigma^+\to \Gamma^+$ and $\psi:\Gamma^+\to \Upsilon^+$ are
both length-multiplying, then $\psi\circ\varphi$ is also a length-multiplying
morphism from $\Sigma^+$ to $\Upsilon^+$.
Given $K \subseteq \Gamma^+$, the \emph{inverse image} of~$K$ under~$\varphi$ is $\varphi^{-1}(K) = \{u \in \Sigma^+ \mid \varphi(u) \in K\}$.

The {\em evaluation morphism} of a semigroup $S$ is the morphism
$\gamma:S^+\to S$ satisfying $\gamma(x)=x$ for all $x\in S$.
A {\em semigroup stamp} is a {\em surjective} semigroup morphism of
the form $\varphi:\Sigma^+\to S$, where
$\Sigma$ is a finite alphabet and $S$ is finite.
Since we are only concerned with semigroup
stamps we henceforth call them {\em stamps}.

A subset $T\subseteq S$ is a \emph{subsemigroup} of $S$ 
if $T$ is closed under the semigroup operation. 
An equivalence relation $\sim$ on $S$ is a \emph{congruence} 
if $x \sim y$ implies $xs \sim ys$ and $sx \sim sy$ for all $x,y,s\in S$.  
The {\em quotient} of $S$ by $\sim$, denoted by $S/{\sim}$, 
is the semigroup $\{[s] \mid s \in S\}$ formed by equivalence classes of $S$
with the multiplication $[s_1][s_2]=[s_1s_2]$.
A semigroup $T$ {\em divides} a semigroup $S$, denoted as $T \preceq S$, 
if $T$ is a quotient of a subsemigroup of $S$.  

A language $L \subseteq \Sigma^+$ is \emph{recognized} by a finite semigroup $S$
if there exists a morphism 
\[
\varphi : \Sigma^+ \to S 
\]
and a subset $P \subseteq S$ such that 
\[
L = \varphi^{-1}(P).
\]
A language $L \subseteq \Sigma^+$ is said to be \emph{recognizable} 
if it is recognized by some finite semigroup. The class of recognizable languages over the alphabet $\Sigma$ is denoted by $\REG(\Sigma^+)$.

Recognizable languages are closed under Boolean operations, quotients by words and images 
under morphisms and inverse morphisms. 
Given a recognizable language $L\subseteq\Sigma^+$ the relation $\sim_L$ on $\Sigma^+$ given by
\begin{align}
u\sim_L v ~~\text{ if }~~
xuy\in L\Leftrightarrow xvy\in L \text{ for each $x,y\in \Sigma^*$}    
\end{align}
is a congruence, known as the {\em syntactic congruence}. 
The corresponding semigroup $\Sigma^+/{\sim_L}$ is called the {\em syntactic semigroup}, 
denoted by $S(L)$ and the canonical morphism $\eta_L: \Sigma^+ \to S(L)$ is called the {\em syntactic morphism of $L$}. 
Note that $\eta_L$ is a stamp, we therefore refer to $\eta_L$ also as
the {\em syntactic stamp of $L$}.
Syntactic semigroups possess the following canonical property,
for a proof we refer to \cite[Lemma 4.11]{Pin2025Mathematical}: for each 
$x\in S(L)$ there exists a finite set of pairs of words
$\{(u_i,v_i)\mid i\in I\}$ and $\{(u_j,v_j)\mid j\in J\}$
such that 
\begin{align}\label{E syn dist}
	\eta_L^{-1}(x)=\bigcap_{i\in I} u_{i}^{-1}Lv_{i}^{-1}\setminus
	\bigcup_{j\in J} u_{j}^{-1}Lv_{j}^{-1}.
\end{align}

It is well-known that the class of recognizable languages coincides
with the class of languages accepted by finite-state automata or equivalently
by regular expressions, namely the class of {\em regular languages}.

Let $\varphi:\Sigma^+\to S$ be a morphism onto a finite semigroup $S$.
The smallest number $s\geq 1$ satisfying $\varphi(\Sigma^s)=\varphi(\Sigma^{2s})$ is called the {\em stability index of $\varphi$};
the {\em stable semigroup of $\varphi$} is 
$\stab(\varphi)=\varphi(\Sigma^s)$.
We remark that $s$ is (at most) exponential in $|S|$.

\subsection{Circuit Complexity}\label{S Circuit Complexity}

Consider two functions $f,g:\N_{>0}\to\bbR_{\geq 0}$.
As usual, we write $f\in O(g)$ if there exists
$c\in\bbR_{>0}$ such that 
 $f(n)\leq c\cdot g(n)$ for all but finitely many $n\in\bbN_{>0}$.
 We write $f\in\Theta(g)$ if $f\in O(g)$ and $g\in O(f)$.
Recall that $f\in o(g)$ if for all
$c\in\bbR_{>0}$ 
we have $f(n) < c\cdot g(n)$ for all but finitely many $n\in\bbN_{>0}$.

For lower bounds the following weaker notion is more
suitable for circuit complexity lower bounds.
We refer to~\cite{VitanyiM84} for more details.

\begin{definition}[Hardy-Littlewood $\Omega$ notation]\label{D Hardy}
	For two functions $f,g:\N_{>0}\to\bbR_{\geq 0}$ we write 
	$f\in\Omega(g)$ if 
	there exists $c\in\bbR_{>0}$ such that $f(n) \geq c\cdot g(n)$
	for infinitely many $n\in\N_{>0}$.
\end{definition}

Note that $f\in\Omega(g)$ if, and only if, $f\not\in o(g)$.

Next we introduce the particular circuit model used in this paper.
The presentation of the inputs employs a subset encoding following
Kouck\'y~\cite{Kouckyhandbook}.
With respect to succinctness it is asymptotically equivalent to the
one-hot encoding model
under the promise that the input is well-formed, as precisely formulated
in Proposition~\ref{P with promise}.

\begin{samepage}
	\begin{definition}\label{D Circuit}
	For each $n\in\bbN_{>0}$ an {\em $n$-circuit} 
over a finite non-empty alphabet $\Sigma$ 
is a tuple $C=(G,\prec,\lambda,g_{out})$, where the following hold:
\begin{itemize}
	\item  $(G,\prec)$ is a 
		directed acyclic graph,
		i.e. $\prec^+$ is a strict partial order.
		The set $G$ is a non-empty set of {\em gates} and 
		$\prec\ \subseteq G\times G$ is a set of {\em wires}. 
		The set of gates $G_0\subseteq G$ of 
		in-degree $0$ are the {\em input gates}. 
		The set of gates $G_1=G\setminus G_0$
		of in-degree $>0$ are the {\em non-input gates}.
\item $\lambda:G\rightarrow\{\wedge,\vee\}\cup \left([0,n-1]\times 2^\Sigma\right)$
assigns each gate a type, 
		where moreover we have $\lambda(G_0)= [0,n-1]\times 2^\Sigma$
		and $\lambda(G_1)\subseteq\{\wedge,\vee\}$.
	\item There is a unique gate $g_{out}\in G_1$ 
that is maximal with respect to $\prec^+$; it is called
the {\em output gate}.
\end{itemize}
\end{definition}
\end{samepage}
An {\em input} to $C$ is a word from $\Sigma^n$.
Given such an input $w=w_0\dots w_{n-1}\in \Sigma^n$, we define the relation
$w\models g$ by induction on $\prec$.
For the induction base, i.e., when $g$ is of in-degree $0$, we define $w\models g$
if $\lambda(g)=(i,\Gamma)$ and $w_i\in\Gamma$.
For the induction step, in case $g$ is an $\wedge$-gate 
(\resp~$\vee$-gate)
we have $w\models g$ if $w\models h$ for
all (\resp~for some) $h\in G$ with $h\prec g$.
We define the {\em language} computed by the circuit $C$ as
$L(C)=\{w\in \Sigma^n\mid w\models g_{out}\}$.

The {\em size} of $C$ is $|C|=|G_1|$, i.e. 
the number of non-input gates. Note that since we require that $g_{out}\in G_1$ 
we always have $|C|\geq 1$.
The {\em wire size} of $C$ is $|\!\!\prec\!\!|$.
The {\em depth} of $C$, denoted by $\depth(C)$, is the length 
(i.e.,~the number of wires) of a longest path from 
some input gate to the output gate in $(G,\prec)$.

A {\em circuit family over the alphabet $\Sigma$} is an infinite tuple 
$\calC=(C_n)_{n\geq 1}$, 
where each $C_n$ is an $n$-circuit over $\Sigma$.
The circuit family $\calC=(C_n)_{n\geq 1}$ over
$\Sigma$ {\em computes}
the language $L(\calC)=\bigcup\{L(C_n)\mid n\geq 1\}$.
The {\em size} of the family $\calC$, denoted 
by $|\calC|:\N_{>0}\to\bbN_{>0}$, is defined to be the function $n\mapsto |C_n|$.
The {\em depth} of the family $\calC$, denoted 
by $\depth(\calC):\N_{>0}\to\N_{>0}$,
is defined to be the function $n\mapsto\depth(C_n)$.
The {\em wire size} is defined analogously.
We remark that our circuit families have unbounded fan-in and
that they are nonuniform:
we impose no condition on our circuit families, other than that 
they satisfy the required size bounds.
Note that given any language $L\subseteq \Sigma^+$ and any input length 
$n\in\bbN_{>0}$
there exists an $n$-circuit $C_n$ such that $L\cap \Sigma^n=L(C_n)$.
Note that for every language $L\subseteq\Sigma^+$ there exists a circuit
family $\calC$ computing $L$, where moreover $|\calC|$ is exponential.
The {\em circuit complexity of a language $L\subseteq\Sigma^+$} is defined
to be the function $\opc(L):\bbN_{>0}\to\bbN_{>0}$, where 
\begin{align}
	\opc(L): n\mapsto
	\min\{|C|:\text{$C$ is an $n$-circuit computing $L\cap\Sigma^n$}\}.
\end{align}
Next, we define circuit complexity classes. Since we are interested in the complexity of regular languages, we formulate them as parameterized families following the definition of varieties of languages. 
For a class of functions $\calF$ from $\N_{>0}$ to $\N_{>0}$ define the size complexity class
$\SIZE(\calF)$ to be the parameterized family of languages that maps each finitely-generated free semigroup $\Sigma^+$ to a set of subsets of $\Sigma^+$ given by
\begin{align}
	\SIZE(\calF):\Sigma^+\mapsto\{L \subseteq \Sigma^+ \mid \exists 
	\text{circ. fam.} \calC: L(\calC)=L, |\calC|\in \calF\}.
\end{align}
Analogously, for classes $\calF$ and $\calG$ of functions, define
\begin{align}
\SIZEDEPTH(\calF,\calG): \Sigma^+ \mapsto \{ L \subseteq \Sigma^+ \mid \exists
\text{circ. fam.}\ \calC: 
	L(\calC)=L, |\calC|\in \calF,\depth(\calC)\in\calG\}.\label{D SIZEDEPTH}
\end{align}
The definition of the class $\WIRESIZEDEPTH(\calF,\calG)$ is similar. We say a language $L\subseteq \Sigma^+$ is in $\SIZE(\calF)$ if $L \in \SIZE(\calF)(\Sigma^+)$. 
For a single function $f:\N_{>0}\to\N_{>0}$ we just write $\SIZE(f)$ to denote $\SIZE(\{f\})$.
The parameterized family $\REG$ of recognizable languages is 
\begin{align}
    \REG: \Sigma^+ \mapsto \REG(\Sigma^+).
    \label{REG}
\end{align}
Finally, for two parameterized families  $\calV, \calW$ we define $\calV\cap\calW$ to be the family
$\calV\cap\calW:\Sigma^+\mapsto \calV(\Sigma^+) \cap \calW(\Sigma^+)$.

Assume that $L$ is a regular language. For an element $s\in S(L)$, it is 
straightforward to construct a circuit family of linear size and logarithmic depth 
that computes the set of words whose image under the syntactic morphism 
$\eta_L$ is $s$.  As a consequence, we obtain

\begin{samepage}
\begin{fact}
\label{prop:reg-lin} For each regular language $L$ we have
	$L\in\SIZEDEPTH(O(n),O(\log n))$.
\end{fact}
\end{samepage}

\subsection{First-Order Logic Over Words}\label{S FO}

As usual, a word $w=a_0\dots a_{n-1}\in\Sigma^n$ is represented as a first-order
structure over the domain $\{0,\dots,n-1\}$ along 
with the unary letter predicates $P_a$, for each $a\in \Sigma$, each of which
is true at a position $i$ if, and only if, $a_i=a$. 
The set of unary letter predicates $(P_a)_{a\in \Sigma}$ is represented in 
the vocabulary by $\Sigma$. 

In addition to the letter predicates, we have \emph{numerical} predicates and constants whose interpretation only depends on the domain. 
The constant symbols $\min$ and $\max$ denote the \emph{first} and the \emph{last} positions of the word, respectively.
The binary predicate $<$ denotes the order on the positions. For $k\geq 0$, the binary relation $+k$, written as $x+k=y$, holds if and only if $y$ is the $k$-th \emph{successor} of $x$.
The collection of all such successor relations is denoted in the vocabulary by $+\omega$. For $r,k\in \bbN$ with $k\geq 1$,
the unary modular predicate $x \equiv r \bmod k$  is true at a position $j$ if $j\equiv r \bmod k$. The collection of modular predicates is denoted by $\mathrm{mod}$ in the vocabulary. 

By $\mathrm{arb}$ we denote the class of all \emph{arbitrary} numerical predicates and constants. A numerical predicate or constant is \emph{regular} if it is definable by an MSO formula using the order relation.  It is easy to verify that all the constants and predicates in the previous paragraph are regular. We let $\mathrm{reg}$ denote  the class of all regular numerical predicates and constants.

For $m\in \bbN$, the counting quantifier $\exists^{\geq m} x$
is true of a formula $\varphi(x)$ whenever at least $m$ elements satisfy $\varphi$, that is,
\begin{align}
\exists^{\geq m} x\, \varphi(x) \Longleftrightarrow |\{ i  \mid w,i \models \varphi(x)\}|\geq m.
\end{align}
The quantifier $\exists^{\leq m} x$ is defined as the negation of $\exists^{\geq m+1}$, that is
\begin{align}
    \exists^{\leq m} x\, \varphi(x) \defeq \neg \exists^{\geq m+1} x\, \varphi(x).
\end{align}

First-order logic with the vocabulary $\tau$ is denoted by $\rmFO[\tau]$. We write $C[\tau]$  for the extension of $\rmFO[\tau]$ obtained by adding counting quantifiers.
For each $k\geq 0$, the logic $\rmFO^k[\tau]$ is the fragment of $\rmFO[\tau]$ consisting of formulas that use at most $k$-variables.
Likewise $C^k[\tau]$ is the $k$-variable fragment of $C[\tau]$. In particular, $\rmFO^1[\tau]$ is the fragment using only the variable $x$ and $\rmFO^2[\tau]$ is the fragment using only the variables $x$ and $y$.

The language $L(\varphi)$ of a formula $\varphi$ is the set of all words satisfying $\varphi$. A language $L\subseteq \Sigma^+$ is defined by the formula $\varphi$ if $L(\varphi)=L$.

\begin{example}
\label{example:logic-const}
Let $w=a_0\ldots a_{n-1}\in \Sigma^+$. The singleton set $\{w\}$ is defined by the following formula in $\rmFO^1[\,\Sigma, \min, \max, +\omega\,]$.
  \begin{align}
    \varphi_w &\defeq \min+(n-1)=\max \,\wedge \bigwedge_{0\leq i \leq n-1}\exists x\, (\min+i = x \wedge P_{a_{i}}(x)).
    \end{align}
The language $(ab)^+$ is defined by the formula $\varphi_{(ab)^+} \in \rmFO^1[\,\Sigma, \min, \max, \mathrm{mod}\,]$.
\begin{align}
 \varphi_{(ab)^+} &\defeq \max \equiv 1 \bmod 2\, \wedge\forall x\, (( x \equiv 0 \bmod 2 \rightarrow P_a(x) ) \wedge (x \equiv 1 \bmod 2 \rightarrow P_b(x))).
\end{align}
The language $\mathit{b^*ab^*ab^*} \subseteq \{a,b\}^*$ is defined by the formula $\varphi_{\mathit{b^*ab^*ab^*}} \in  C^1[\,\Sigma\,]$.
\begin{align}
\varphi_{\mathit{b^*ab^*ab^*}} &\defeq \exists^{\geq 2}x\, P_a(x) \wedge \neg \exists^{\geq 3}x\, P_a(x).
\end{align}

\end{example}

\section{Reductions \& Varieties of Recognizable Languages}\label{S Reductions}

This section introduces a particular notion of rational
truth-table reduction that we use
throughout the paper, mainly as a tool for proving lower bounds.
We introduce these in~\S\,\ref{S Reductions Recognizable},
show their transitivity and also introduce
the class of mild functions for which circuit complexity
lower and upper bounds are proved to be carried over.
\S\,\ref{S Reductions Stamps} provides standard tools
for proving circuit complexity for regular languages yet for our notion
of rational truth-table reduction.
In~\S\,\ref{S Circuit LM} we prove that the classes of regular languages
$\SIZE(O(f(n)))$ and 
$\SIZEDEPTH(O(f(n)),O(g(n)))$ for mild functions $f,g$ are in fact lm-varieties of
languages, a notion we also recall there.
~\S\,\ref{S Pseudo} is concerned with length-multiplying
pseudovarieties of stamps and~\S\,\ref{sec:identity} with
identities on semigroups and stamps.

\subsection{Reductions for Recognizable Languages}\label{S Reductions Recognizable}
Throughout the paper we will be working with the following 
notion of reduction, a rigid form of constant-depth reductions~\cite{ConstantDepth1984}.

\begin{samepage}
	\begin{definition}[Rational Truth-Table Reduction]\label{D Rational}
Let $L_1\subseteq \Sigma_1^+$ and $L_2\subseteq \Sigma_2^+$ be languages.
	The language {\em $L_1$ rationally truth-table reduces to $L_2$}, 
$L_1\reduce L_2$ in notation, if 
there exists a length-multiplying morphism 
$\varphi:\Sigma_1^+\rightarrow \Sigma_2^+$ such that
$L_1$ is a finite Boolean combination over
a finite set of languages all of the form $\varphi^{-1}(u^{-1}L_2v^{-1})$,
where $u$ and $v$ range over $\Sigma_2^*$.
\end{definition}
\end{samepage}

We remark that there are two canonical instances of a rational 
truth table reduction. 
The first occurs when the Boolean function is the identity and 
$u=v=\varepsilon$; such reductions correspond to inverse images of 
length-multiplying morphisms. The second occurs when 
$\varphi$ is the identity; these correspond to finite Boolean combinations of two-sided quotients.
The following proposition states that various operations commute. It
is handy for showing transitivity of rational truth-table reductions.

\begin{minipage}{\textwidth}
\begin{samepage}
\begin{propositionrep}[Commuting operations]\label{P commute}
We observe the following commutativities of operations.
	\begin{enumerate}
		\item Inverse (length-multiplying) morphisms and Boolean operations 
			commute, i.e.,
			\begin{align}
				\varphi^{-1}(L_1\cap L_2)&=\varphi^{-1}(L_1)\cap \varphi^{-1}(L_2)
				\label{E mucap}\\
				\varphi^{-1}(\overline{L})&=\overline{\varphi^{-1}(L)}\label{E muneg} \end{align}
			for all (length-multiplying) morphisms $\varphi:\Sigma^+\to\Gamma^+$
			and languages $L_1,L_2,L\subseteq\Gamma^+$.
		\item Two-sided quotients and Boolean operations commute, i.e. 
			\begin{align}
				u^{-1}(L_1\cap L_2)v^{-1}&=
				u^{-1}L_1v^{-1}\cap u^{-1}L_2v^{-1}
				\label{E ucap}\\
				u^{-1}\overline{L}v^{-1}&=
				\overline{u^{-1}Lv^{-1}}\label{E uneg}
			\end{align}
			for all $u,v\in\Sigma^*$ and languages $L\subseteq\Sigma^+$.
		\item Inverse (length-multiplying) morphisms and two-sided
			quotients commute in the following sense:
			\begin{align}
				u^{-1}\varphi^{-1}(L)v^{-1}&=
				\varphi^{-1}(\varphi(u)^{-1} L\varphi(v)^{-1})
				\label{E umu}
			\end{align} for all (length-multiplying) morphisms 
			$\varphi:\Sigma^+\to\Gamma^+$ and languages $L\subseteq\Gamma^+$.
	\end{enumerate}
	\end{propositionrep}
\end{samepage}
\end{minipage}
\begin{appendixproof}
	\noindent\quad (1)
	\begin{align}
		\varphi^{-1}(L_1\cap L_2)&=\{w\in \Sigma^+\mid
		\varphi(w)\in L_1\cap L_2\}\\
		&= \{w\in\Sigma^+\mid\varphi(w)\in L_1\}\cap 
\{w\in\Sigma^+\mid\varphi(w)\in L_2\}\\
		&=\varphi^{-1}(L_1)\cap\varphi^{-1}(L_2)
	\end{align}
	\begin{align}
		\varphi^{-1}(\overline{L})&=
		\{w\in\Sigma^+\mid\varphi(w)\in\overline{L}\}\\
		&=\{w\in\Sigma^+\mid\varphi(w)\in\Gamma^+\setminus L\}\\
		&=\{w\in\Sigma^+\mid\varphi(w)\in\Gamma^+\}
		\setminus \{w\in\Sigma^+\mid\varphi(w)\in L\}\\
		&=\Sigma^+\setminus\{w\in\Sigma^+\mid\varphi(w)\in L\}\\
		&=\overline{\varphi^{-1}(L)}
	\end{align}
	\noindent\quad (2)
	\begin{align}
		u^{-1}(L_1\cap L_2)v^{-1}&=\{w\in\Sigma^+\mid uwv\in
		L_1\cap L_2\}\\
		&=\{w\in\Sigma^+\mid uwv\in L_1\}\cap
		\{w\in\Sigma^+\mid uwv\in L_2\}\\
		&=u^{-1}L_1v^{-1}\cap u^{-1}L_2v^{-1}
	\end{align}
	\begin{align}
		u^{-1}\overline{L}v^{-1} &=
		\{w\in\Sigma^+\mid uwv\in\overline{L}\}\\
		&=\{w\in\Sigma^+\mid uwv\in\Sigma^+\setminus L\}\\
		&=\{w\in\Sigma^+\mid uwv\in\Sigma^+\}\setminus
\{w\in\Sigma^+\mid uwv\in L\}\\
		&=\Sigma^+\setminus\{w\in\Sigma^+\mid uwv\in L\}\\
		&=\overline{u^{-1}Lv^{-1}}
	\end{align}
	\noindent\quad (3)
	\begin{align}
		u^{-1}\varphi^{-1}(L)v^{-1}&=
		\{w\in\Sigma^+\mid uwv\in\varphi^{-1}(L)\}\\
		&=\{w\in\Sigma^+\mid \varphi(uwv)\in L\}\\
		&=\{w\in\Sigma^+\mid \varphi(u)\varphi(w)\varphi(v)\in L\}\\
		&=\{w\in\Sigma^+\mid \varphi(w)\in\varphi(u)^{-1}L\varphi(v)^{-1}\}\\
		&=\varphi^{-1}(\varphi(u)^{-1}L\varphi(v)^{-1})
	\end{align}
\end{appendixproof}

\begin{proposition}\label{P reduce transitivity}
	If $L_1\reduce L_2\reduce L_3$, then $L_1\reduce L_3$.
\end{proposition}
\begin{proof}
For a class $\calC$ of languages over $\Sigma$
let $\bool(\calC)$ denote the class of {\em finite Boolean combinations of languages
in $\calC$}.

Assume $L_1\reduce L_2\reduce L_3$
for languages $L_1\subseteq\Sigma_1^+$,
$L_2\subseteq\Sigma_2^+$, and
$L_3\subseteq\Sigma_3^+$.
By definition there exist length-multiplying
morphisms $\varphi:\Sigma_1^+\rightarrow\Sigma_2^+$
and $\nu:\Sigma_2^+\rightarrow\Sigma_3^+$,
finite sets of pairs of words
	$I\subseteq\Sigma_2^*\times\Sigma_2^*$
	and $J\subseteq\Sigma_3^*\times\Sigma_3^*$
such that 
\begin{align}
	L_1\in\bool\left(\biggl.\left\{\bigl.\varphi^{-1}(u^{-1}L_2v^{-1})\mid (u,v)\in I\right\}\right)\label{E R2}
\end{align}
	and 
	\begin{align}
		L_2\in\bool\left(\biggl.\left\{
			\nu^{-1}(x^{-1}L_3y^{-1})\mid (x,y)\in J\right\}\right).\label{E R3}
	\end{align}
	We substitute the Boolean expression for $L_2$ inside the Boolean expression
	for $L_1$ and distribute the Boolean operations over 
	quotients and the inverse morphism via Point (1) and (2) 
	of Proposition~\ref{P commute}.
	Thus, we obtain that $L_1$ is a finite Boolean combination of languages of the form
	\begin{align}
		\varphi^{-1}\left(u^{-1}\left[ 
		\nu^{-1}\left(x^{-1}L_3y^{-1}\right)\right]
		v^{-1}\right),\label{E mu u}
	\end{align}
	where $(u,v)\in I$ and $(x,y)\in J$.
	By Point (3) of Proposition~\ref{P commute}
	we obtain that (\ref{E mu u}) is equivalent
	to 
	\begin{align}
		\varphi^{-1}\left(\nu^{-1}\left(
		\nu(u)^{-1}\left(x^{-1}L_3y^{-1}\right)
		\nu(v)^{-1}\right)\right),\label{E mu u2}
	\end{align}
	which in turn, by the identity
	$s^{-1}(r^{-1}Lq^{-1})t^{-1}=(rs)^{-1}L(tq)^{-1}$,
	is equivalent to 
	\begin{align}
		\varphi^{-1}\left(\nu^{-1}
		\left((x\nu(u))^{-1}L_3(\nu(v)y)^{-1}
		\right)\right).\label{E mu u3}
	\end{align}
 Since $\varphi^{-1}(\nu^{-1}(L))=(\nu\circ\varphi)^{-1}(L)$
	we have that (\ref{E mu u3})
is equivalent to
	\begin{align}
(\nu\circ\varphi)^{-1}\left(
		(x\nu(u))^{-1} L_3
		(\nu(v)y)^{-1}\right).\label{E compose morph}
	\end{align}
	Thus, $L_1$ is a finite Boolean combination of languages of the
	form (\ref{E compose morph}). Since $(\nu\circ\varphi)$ is again
	a length-multiplying morphism we obtain
$L_1\reduce L_3$.
\end{proof}

\begin{lemma}\label{L sum}
Assume $L_1\reduce L_2$. 
	Then there exist constants $k\in\bbN$ and
	$d,\ell,m\in\bbN_{>0}$ such that for all $n\in\N_{>0}$ we have
	\begin{align}
		\opc(L_1)(n)\leq d+m\cdot\sum_{i=0}^{k}\opc(L_2)(\ell n+i).
		\label{E CC}
	\end{align}
\end{lemma}
\begin{proof}
Assume $L_1\reduce L_2$ for languages $L_1\subseteq \Sigma_1^+$ and 
	$L_2\subseteq \Sigma_2^+$.
Recall that by definition there exists a length-multiplying
	morphism  $\varphi:\Sigma_1^+\to\Sigma_2^+$ such that $\varphi(
	\Sigma_1) \subseteq \Sigma_2^\ell$ for some $\ell\in\bbN_{>0}$
and a finite set $J\subseteq\Sigma_2^*\times\Sigma_2^*$ such that 
	$L_1$ is a finite Boolean combination $\beta$ of languages from
	\begin{align}
		\{\varphi^{-1}(u^{-1}L_2v^{-1})\mid (u,v) \in J\} \label{E CC1}.
	\end{align}

	Now
	$\beta$ may be assumed to be an intersection of unions of
	languages of the form~(\ref{E CC1}) or complements thereof.
	To prove the lemma it suffices to show that
	for every circuit family $\calC'=(C_n')_{n\in\bbN_{>0}}$ computing
	$L_2$ one can construct a circuit family 
	$\calC=(C_n)_{n\in\bbN_{>0}}$ computing
	$L_1$ such that (\ref{E CC}) is satisfied. 
	By virtue of the previous 
	observation (\ref{E CC1}),
	for each $n\in\bbN_{>0}$, it suffices for the circuit $C_n$ to 
	verify a Boolean condition on the membership of the 
	words $\{u\varphi(w)v \mid (u,v)\in J\}$ in $L_2$.

    Let $u=u_0\ldots u_{s-1}$ and $v=v_0\ldots v_{t-1}$, for $s,t\in\bbN$.
	A circuit $C_n$ that takes an input $w\in \Sigma_1^n$  and verifies the membership of the word $u\varphi(w)v$ in $L_2$ is constructed as follows: We take the circuit $C'_{\ell n+s+t}$ and fix its first $s$ and last $t$ inputs to be the words $u$ and $v$ respectively. I.e., 
	each wire from an input gate 
	$(p,\Gamma_p)$, where $p\in[0,s-1]$
	to some AND/OR-gate $g$ is redirected
	 from an input that is always true (say from $(0,\Sigma_1)$) to $g$ if 
	 $u_p\in\Gamma_p$ 
	and from an input that is always false (say $(0,\emptyset)$) to $g$ 
	if $u_p\not\in\Gamma_p$. 
	Likewise, 
each wire from an input gate 
	$(p,\Gamma_p)$ where $p\in[\ell n+s,\ell n+s+t-1]$
	to a gate $g$ is redirected
	from $(0,\Sigma_1)$ to $g$ if $v_{p-(\ell n+s)}\in\Gamma_p$ 
	and from $(0,\emptyset)$ to $g$ if $v_{p-(\ell n+s)}\not\in\Gamma_p$.

    Finally for all the other inputs, we apply the morphism $\varphi$. I.e., each wire from an input gate $(p,\Gamma_p)$
	where $p\in[s,\ell n+s-1]$ to some AND/OR-gate $g$ is
	redirected from 
	$(p',\Gamma_p') \in [0,n-1]\times 2^{\Sigma_1}$ to $g$ where $p'$ and $\Gamma_p'$ are defined as follows:
    \begin{align}
	    p'&= \left\lfloor \frac{p-s}{\ell}\right\rfloor\\
        \Gamma_p' &=\{ a \in \Sigma_1 \mid (\varphi(a))_{(p-s) \bmod \ell} \in \Gamma_p\}
	\end{align}
    We tacitly assume $J\not=\emptyset$.
	The conjunctive normal form underlying
	$\beta$ has size $d\in O(1)$. For 
	$k=\max\{|u|+|v|:(u,v)\in J\}$
	and $m=2\cdot|J|$ (considering both the circuit as well as the negated circuit)
	we obtain the estimation
	$|C_n|\leq d+m\cdot\sum_{i=0}^k|C_{\ell n+i}'|$, 
	thus showing (\ref{E CC}).
\end{proof}

Let $f:\bbN_{>0}\to\bbR_{\geq 0}$ be a function. 
We call $f$ {\em non-decreasing} if 
$f(n)\leq f(n+1)$ for all $n\in\bbN_{>0}$.
We call $f$ {\em eventually positive} if $f(n)>0$ for all but finitely 
many $n\in\bbN_{>0}$.

\begin{definition}\label{D mild}
	A function $f:\bbN_{>0}\to\bbR_{\geq 0}$ is {\em mild} if 
	$f$ is non-decreasing, eventually positive, and satisfies 
	$f(n)\in\Theta(f(\lceil\varepsilon n\rceil))$ for every $\varepsilon\in\bbR_{>0}$.
\end{definition}

\begin{example}
	The eventually positive constant functions are mild
	and every non-decreasing and eventually positive polynomial is mild.
	Moreover, $n\mapsto\log^k n$ is mild for all $k\in\N_{>0}$.
 The function $n\mapsto 2^n$ is not mild. 
\end{example}

\begin{samepage}
	\begin{theorem}\label{T asymptotics}
	Let $f:\bbN_{>0}\to\bbR_{\geq 0}$ be a mild function.
	Let $L_1$ and $L_2$ be languages such that $L_1\leq L_2$.
	\begin{enumerate}
		\item If $\opc(L_2)\in O(f)$, then $\opc(L_1)\in O(f)$.
		\item If $\opc(L_1)\in\Omega(f)$, then $\opc(L_2)\in \Omega(f)$.
	\end{enumerate}
\end{theorem}
\end{samepage}
\begin{proof}
	Let $f:\bbN_{>0}\to\bbR_{\geq 0}$ be a mild function and let $L_1\reduce L_2$.
	By Lemma~\ref{L sum} there exist constants $k\in\bbN$ and $d,\ell,m\in\N_{>0}$
	such that
	\begin{align}
		\opc(L_1)(n)\leq d+m\cdot\sum_{i=0}^{k}
		\opc(L_2)(\ell n+i)\label{E CCCC}
	\end{align}
	for all $n\in\bbN_{>0}$.

	\noindent (1)\quad
	Assume $\opc(L_2)\in O(f)$.
	For $\varepsilon=\ell+1$  
	for all but finitely many $n\in\bbN_{>0}$
	we have:
	\begin{align}
		\forall i\in[0,k]:\quad \ell n+i \leq \lceil\varepsilon n\rceil\label{E eps upper}
	\end{align}
Hence there exist constants $c_1,c_2,c_3\in\bbR_{>0}$ such that for all
but finitely many $n\in\bbN_{>0}$ we have:
\begin{align}
	\opc(L_1)(n)&\leq d+m\cdot\sum_{i=0}^k\opc(L_2)(\ell n+i) & \text{(By (\ref{E CCCC}))}\\
	&\leq d+m\cdot\sum_{i=0}^k c_1\cdot f(\ell n+i)&\text{(Since $\opc(L_2)\in O(f)$)}\\
	&\leq d+m\cdot\sum_{i=0}^k c_1\cdot f(\lceil\varepsilon n\rceil)&\text{(By (\ref{E eps upper})
	\& since $f$ is non-decreasing)}\\
	&=d+c_1\cdot (k+1)\cdot m\cdot f(\lceil\varepsilon n\rceil)\\
	&\leq d+c_2\cdot f(n)&\text{(Since $f\in\Theta(f(\lceil\varepsilon n\rceil))$)}\\
	&\leq c_3\cdot f(n)&\text{(Since $f$ is eventually positive)}
\end{align}
	Thus, $\opc(L_1)\in O(f)$.\\

	\noindent (2)\quad
	Assume $\opc(L_1)\in \Omega(f)$.

	In case $f$ is bounded, i.e., there exists $b\in\bbN_{>0}$ such that
	$f(n)\leq b$ for all $n\in\bbN_{>0}$, we have that 
	$\opc(L_2)\in \Omega(f)$ since every $n$-circuit has size at least
	one.
	It remains to consider the case when 
	$f$ is not bounded.

	Since $\opc(L_1)\in \Omega(f)$ there exists a constant
	$c_1\in\bbR_{>0}$ and an infinite subset $N_1\subseteq\bbN_{>0}$ such that
	\begin{align}
		c_1\cdot f(n)\leq\opc(L_1)(n)\label{E c1}
	\end{align}
	for all $n\in N_1$.
	By applying the pigeonhole principle to (\ref{E CCCC}) 
	for each $n\in\bbN_{>0}$ there exists
	$\overline{n}\in\{\ell n+i\mid i\in[0,k]\}$ such
	that
	\begin{align}
		\opc(L_2)(\overline{n})\geq \frac{1}{(k+1)m}\cdot 
		(\opc(L_1)(n)-d).\label{E sizeL2}
	\end{align}
	Note that $\overline{N_1}=\{\overline{n}\mid n\in N_1\}$ is an infinite set.
	Again by the pigeonhole principle, there exists 
	some $i\in[0,k]$ and an infinite subset $N_2\subseteq \overline{N_1}$
	such that all elements $\overline{n}\in N_2$ have the same difference 
	$\overline{n}-\ell n=i$. 
	Hence, $\frac{\overline{n}-i}{\ell}$ is an element of $N_1$
	for all $\overline{n}\in N_2$.
	For $\varepsilon=\frac{1}{2\ell}$ we have 
	\begin{align}
		\lceil\varepsilon \overline{n}\rceil \leq \frac{\overline{n}-i}{\ell}
		\label{E varepsilon}
	\end{align}
	for all but finitely many elements $\overline{n}$ of the infinite set $N_2$.
	Finally, there exist constants $c_2,c_3\in\bbR_{>0}$ such that for
	all but finitely many $\overline{n}\in N_2$ the following holds:
	\begin{align}
		\opc(L_2)(\overline{n})&\geq \frac{1}{(k+1)m}\left(\opc(L_1)
		\left(\frac{\overline{n}-i}{\ell}\right)-d\right)
		&\text{(By (\ref{E sizeL2}))}\\
		&\geq\frac{1}{(k+1)m}\left(c_1\cdot f\left(\frac{\overline{n}-i}{\ell}\right)-d\right)
		&\text{(By (\ref{E c1}))}\\
		&\geq c_2\cdot f\left(\frac{\overline{n}-i}{\ell}\right)
		&\text{(As $f$ is not bounded \& non-decreasing)}\\
&\geq c_2\cdot f\left(\lceil\varepsilon \overline{n}\rceil\right)
		&\text{(By (\ref{E varepsilon}))}\\
		&\geq c_3\cdot f(\overline{n}) 
		&\text{(Since $f(n)\in\Theta(f(\lceil\varepsilon n\rceil))$)}
	\end{align}
	Thus, $\opc(L_2)\in\Omega(f)$, as required.
\end{proof}

\subsection{Reductions and Stamps}\label{S Reductions Stamps}

\begin{lemma}\label{L reduce}
	Let $L$ be a regular language.
	Then for each $x\in S(L)$ we have $\eta_L^{-1}(x)\reduce L$.
\end{lemma}
\begin{proof}
Follows immediately from~(\ref{E syn dist}).
\end{proof}

\begin{corollary}\label{C U reduce}
Let $L$ be a regular language.
	Then for all $P\subseteq S(L)$ we have
	$\eta_L^{-1}(P)\reduce L$.
\end{corollary}

\begin{samepage}
\begin{definition}
	For each $s\in\bbN_{>0}$ we define the mutually-inverse 
	morphisms $\block_s: \Sigma^{s\bbN_{>0}}\to (\Sigma^s)^+$ and 
	$\flat_s: (\Sigma^s)^+ \to \Sigma^{s\bbN_{>0}}$ as
follows: 
	for each $u=u_0\ldots u_{ns-1}\in\Sigma^{ns}$ we define
	\begin{align}
		\block_s(u)=\overline{v_0}\ldots \overline{v_{n-1}}
	\end{align}
	where $\overline{v_i}=u_{is}\ldots u_{(i+1)s-1}$ for all $i\in[0,n-1]$.
	Conversely, for each $v=\overline{v_0}\ldots \overline{v_{n-1}}\in(\Sigma^s)^{n}$
	where $u_{is+j}=(\overline{v_i})_j$ for each $i\in[0,n-1]$ and each $j\in[0,s-1]$
	we define
	\begin{align}
	\flat_s(v)= u_{0}\ldots u_{ns-1}.
	\end{align}
\end{definition}
\end{samepage}

\begin{proposition}\label{P stable reduce}
Let $L\subseteq\Sigma^+$ be a regular language with the
syntactic morphism $\eta_L:\Sigma^+\rightarrow S(L)$
and let $s$ be the stability of $\eta_L$.
Then $K_x\reduce L$ for each $x\in \stab(\eta_L)$,
where $K_x\subseteq(\Sigma^s)^+$
	is the language $K_x=(\eta_L\circ\flat_s)^{-1}(x)$.
\end{proposition}
\begin{proof}
	Let $x\in\stab(\eta_L)$.
	Since $K_x=(\eta_L\circ\flat_s)^{-1}(x)=\flat_s^{-1}(\eta_L^{-1}(x))$
	and $\flat_s$ is length-multiplying
	we obviously have $K_x\leq \eta_L^{-1}(x)$.
	By Lemma~\ref{L reduce} we have $\eta_L^{-1}(x)\leq L$.
	From 
	Proposition~\ref{P reduce transitivity} we
	deduce $K_x\reduce L$.
\end{proof}

\subsection{Circuits \& Length-Multiplying Varieties of Languages}\label{S Circuit LM}

Let $\calV$ be a map that associates each finite alphabet $\Sigma$ with a family of regular languages $\calV(\Sigma^+)$ over the alphabet $\Sigma$.
The map $\calV$ is a \emph{variety of languages} if the following conditions are met.
\begin{enumerate}
\item For each alphabet $\Sigma$, the family $\calV(\Sigma^+)$ is closed under Boolean operations.
\item For each alphabet $\Sigma$, the family $\calV(\Sigma^+)$ is closed under 
	positive left and right quotients by letters, i.e., 
		if $L\in\calV(\Sigma^+)$ and $a\in\Sigma$,
		then $a^{-1}L,La^{-1}\in\calV(\Sigma^+)$.
\item For each morphism $\varphi: \Gamma^+ \rightarrow 
\Sigma^+$ and $L\in \calV(\Sigma^+)$, then $\varphi^{-1}(L) \in \calV(\Gamma^+)$.
\end{enumerate}
A \emph{length-multiplying variety of languages}, \emph{lm-variety} for short, $\calV$ 
is the generalization, where $\varphi$ is restricted to length-multiplying morphisms. Notice that every variety is also an lm-variety, but the converse is not necessarily true.
It is readily observed that $\REG$ (\ref{REG}) is the largest variety and lm-variety of languages. 

The below proposition is folklore, its proof uses ideas from~\cite{Straubing02}.
In fact, it follows immediately from Point (1) of Theorem~\ref{T asymptotics}.

\begin{samepage}
\begin{theorem} \label{T mild variety}
	If $f,g:\bbN_{>0}\to\bbR_{\geq 0}$ are mild functions, then 
	\begin{enumerate}
		\item $\SIZE(O(f))\cap\REG$ and 
		\item $\SIZEDEPTH(O(f),O(g))\cap\REG$
	\end{enumerate}
are $\lm$-varieties of languages. 
\end{theorem}
\end{samepage}
\begin{proof}
We only give the proof for Point (1), the proof for Point (2) is a slight adaption.
Let $\calV$ denote the family of regular languages $\SIZE(O(f))\cap\REG$. 
	We need to show that $\calV$ is closed under Boolean operations, positive quotients by letters, and inverse lm-morphisms.  
	The closure under positive quotients by letters, inverse lm-morphisms
	and negation follows from Point (1) of Theorem~\ref{T asymptotics}.
	It remains to prove closure under union.
Let $\Sigma$ be a finite alphabet. 
Assume that $L, L' \in \calV(\Sigma^+)$. Then there exist
	circuit families $\calC=\left(C_n\right)_{n \in \bbN_{>0}}$ and $\calC'=\left(C'_n\right)_{n \in \bbN_{>0}}$ computing $L$ and $L'$, respectively,
such that $|\calC|,|\calC'| \in O(f)$.
It is clear that the language $L \cup L'$ can be computed by a circuit family 
	$\calC''$ satisfying
	$|\calC''|(n)=|\calC|(n)+|\calC'|(n)+1$ for all $n\in\bbN_{>0}$.
	Since the function $f$ is mild, we have $|\calC''|\in O(f)$ as well 
	and therefore $L \cup L' \in \calV(\Sigma^+)$, as required.
\end{proof}

\begin{corollary} 
    The intersections of the classes $\SIZE(O(1)), \SIZE(O(\log^k n))$, and
	$ \SIZE(O(n^\varepsilon))$ with $\REG$ are lm-varieties for all 
	$k\in\bbN_{>0}$ and $\varepsilon\in \bbR_{> 0}$.
\end{corollary}

\subsection{Length-Multiplying Pseudovarieties of Stamps}\label{S Pseudo}

A \emph{pseudovariety of finite semigroups} is a class $\bfV$ 
of finite semigroups closed under finite direct 
products and 
division, i.e., if $S\in\bfV$ and $T$ divides $S$, then $T$ is in $\bfV$.  
Eilenberg's variety theorem states that pseudovarieties of  finite semigroups and varieties of languages are in one-to-one correspondence \cite{Eilenberg1976}. Pseudovarieties are denoted by boldface letters and the  corresponding varieties of languages are denoted by calligraphic letters. For example $\bfJ_{\bfone}$ denotes the pseudovariety of idempotent and commutative semigroups 
and $\calJ_{1}$ denotes the variety of idempotent and commutative languages.

\begin{figure}[ht]
\centering
\begin{subfigure}{0.48\textwidth}
\centering
\[\begin{tikzcd}[
    column sep=2em,
    row sep=1.5em
]
&&&S_1\\
&&&\\
&&&\\
\Sigma^+ \arrow[uuurrr,"\varphi_1"] 
\arrow[rr,"\varphi"]
\arrow[dddrrr,"\varphi_2"] 
&&\operatorname{Im}(
\varphi)\subseteq S_1\times S_2
\arrow[uuur,"\pi_1" '] 
\arrow[dddr,"\pi_2"] 
&\\
&&&\\
&&&\\
&&&S_2
\end{tikzcd}
\]
\caption{Diagonal direct product $\varphi$ of the stamps $\varphi_1$ and $\varphi_2$, where $\pi_1$ and $\pi_2$ denote the canonical projection morphisms.}
\label{fig:first}
\end{subfigure}
\hfill
\begin{subfigure}{0.48\textwidth}
    \centering
    \raisebox{1cm}{
    \begin{tikzpicture}[>=stealth]
  \node (A) at (0,3) {$\Sigma^+$};
  \node (B) at (4.1,3) {$\Gamma^+$};
  \node (C) at (0,0) {$S_1$};
  \node (D) at (3,0) {$\operatorname{Im}(\varphi_2 \circ \alpha)\subseteq
     S_2$};
     \node (E) at (4.1,0) {\phantom{S}};
  \draw[->] (A) -- node[midway, above] {$\alpha$} (B);
  \draw[->] (A) -- node[midway, left] {$\varphi_1$} (C);
  \draw[->] (B) -- node[midway, right] {$\varphi_2$} (E);
  \draw[->] (D) -- node[midway, below ] {$\beta$} (C);
\end{tikzpicture}}

    \caption{Stamp $\varphi_1$ lm-dividing the stamp $\varphi_2$, the morphism $\alpha$ is length-multiplying and $\beta$ is surjective.} 
    \label{fig:second}
\end{subfigure}
\caption{Product \& Division of Stamps}
\label{fig:both}
\end{figure}
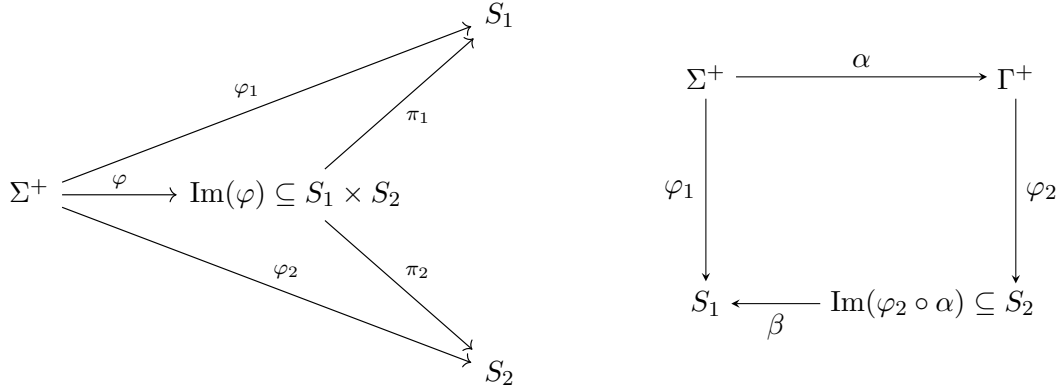

A similar correspondence exists for lm-varieties of languages.   
The \emph{diagonal direct product} of the stamps $\varphi_1: \Sigma^+ \to S_1$ and $\varphi_2: \Sigma^+ \to S_2$ is the stamp 
$\varphi:w \mapsto (\varphi_1(w),\varphi_2(w))$ from $\Sigma^+$ onto the 
subsemigroup $\varphi(\Sigma^+) \subseteq S_1\times S_2$. 
The stamp $\varphi_1 : \Sigma^+ \rightarrow S_1$ {\em lm-divides} 
the stamp  $\varphi_2 : \Gamma^+ \rightarrow S_2$ 
if there exists a length-multiplying morphism 
$\alpha:\Sigma^+ \rightarrow \Gamma^+$ and a surjective morphism $\beta: \operatorname{Im}\,(\varphi_2 \circ \alpha) \rightarrow S_1$ 
such that $\varphi_1 = \beta \circ \varphi_2 \circ \alpha$. 
The operations are described in the \Cref{fig:both}.
A class of stamps $\bfV$ is an
{\em lm-pseudovariety of stamps} if it is closed under finite diagonal direct products and lm-divisions.

\begin{theorem}[Straubing \cite{Straubing02}]
Lm-pseudovarieties of stamps and lm-varieties of languages are in one-to-one correspondence.
\end{theorem}

If $\bfV$ is a variety of semigroups, then the collection of all stamps whose stable semigroup is in $\bfV$ forms an  lm-pseudovariety of stamps, denoted by $\bfQ\bfV$ \cite{Straubing02}.

\subsection{Identities on Semigroups}
\label{sec:identity}

An \emph{identity} is an equation $t=t'$, where
$t$ and $t'$ are finite words built from a set of variables
using multiplication and $\omega$-power. 
For example $xy=yx$, $x = x^\omega$, and $(xy^\omega z)^\omega=xy^\omega$ are all identities. 
A semigroup $S$ {\em satisfies} the identity $t=t'$ if the two sides evaluate to the same element under every assignment of the variables to elements of $S$.
For instance, the first two identities are satisfied by $S$ if, and only if,
it is commutative and every element of $S$ is idempotent. By convention, we use variables $e,f,\dots$ as shorthand for idempotent elements.


\section{The Regular Languages of Constant Circuit Complexity}\label{S Characterization}

This section contains various characterizations of regular languages with constant circuit complexity. 
It is organized as follows.
We begin in \S\,\ref{sec:c_k_t} by introducing the modular threshold-counting congruence $\congC_k^t$ followed by its refinement $\approx_k^t$  (\S\,\ref{sec:approx_k}) that is used in the $O(1)$ characterization theorem. The pseudovariety $\QESL$ that characterizes constant circuit complexity is presented in \S\,\ref{sec:qesl}. 
The $O(1)$ characterization theorem is proved in \S\,\ref{sec:O(1)}. 
The two difficult directions of the proof are presented as separate subsections. 
In \S\,\ref{S Lower Bound const}, it is shown that if the syntactic stamp 
of a language is not in $\QESL$, then it has $\Omega(\log n)$ circuit complexity. 
Finally, in 
\S\,\ref{S Congruence const}, it is shown that if the syntactic stamp of a language 
is in $\QESL$, then it is a finite union of equivalence classes of the congruence $\approx_k$ for 
some $k\in \bbN_{>0}$.

\subsection{The Modular Threshold-Counting Congruence $\congC_k^t$}
\label{sec:c_k_t}

\begin{samepage}
\begin{definition}[Modular Projection $\pi_{i,k}$]
	Let $k\in\bbN_{>0}$. For a word $w=a_0\dots a_{n-1}\in \Sigma^*$ 
and $i\in\bbZ$ we define
\begin{align}
	\pi_{i,k}(w) = \prod_{j\in [0,n-1],\atop
	j\equiv_k i}a_j.
\end{align}
 \end{definition}
\end{samepage}

\begin{samepage}
\begin{definition}[$\cont_{i,k}$ and $\cont_{i,k}^{\bowtie\,t}$]
    Given a threshold $t\in\bbN_{>0}$, a comparison operation 
${\bowtie}\in\{=,<,>,\leq,\geq\}$,
and a word $w\in \Sigma^*$ let 
$\cont^{\bowtie\, t}(w)=\{a\in \Sigma:|w|_a\bowtie t\}$.
For each $k\in\bbN_{>0}$ and each $i\in\bbN$,  
	let  $\cont_{i,k}=\cont \mathbin{\circ} \pi_{i,k}$ and
more generally, $\cont^{\bowtie\, t}_{i,k} = \cont^{\bowtie\, t} \mathbin{\circ} \pi_{i,k}$.
\end{definition}
\end{samepage}

\begin{definition}[The Modular Threshold-Counting Congruence $\scrC_k^t$]
Let $k,t\in \bbN_{>0}$.  The relation $\congC_k^t$ on $\Sigma^+$ is defined as follows: $u \congC_k^t v$ if
\begin{align}
  |u|\equiv_k |v| \text{ and }   \left |\pi_{i,k}(u)\right |_a \equiv^t \left |\pi_{i,k}(v)\right |_a 
\end{align}
	for all $a\in \Sigma$ and $i\in[0,k-1]$.
\end{definition}

Clearly, $\congC_k^t$ is an equivalence relation. 
Note that every word of length $<k$ forms a singleton equivalence class.
Also, $u \congC_k^1 v$ if $|u|\equiv_k|v|$ and $\cont_{i,k}(u)=\cont_{i,k}(v)$ 
for all $i\in\bbN$.

\begin{lemmarep}
Let $k,t \in \bbN_{>0}$. 
\begin{enumerate}
	\item  $\congC_{k}^t$ is a congruence.
    \item  $\congC_{k}^t$ refines $\congC_{s}^t$ whenever $s|k$.
\end{enumerate}
\label{lemma:prop-cont}
\label{L cont}
\end{lemmarep}
\begin{appendixproof}
	\noindent (1)\quad Assume that 
         $u \congC_k^t v$. Therefore by definition $|u|\equiv_k |v|$ and 
$\left|\pi_{i,k}(u)\right |_a \equiv^t \left |\pi_{i,k}(v)\right |_a$ for all $a\in \Sigma$ and $i\in [0,k-1]$. Let $x,y\in \Sigma^*$. Clearly, $|xuy|\equiv_k |xvy|$.  Let $i\in [0,k-1]$ and $a\in \Sigma$. Then,
        \begin{align}\left|\pi_{i,k}(xuy)\right|_a&=\left |\pi_{i,k}(x)\right|_a+\left |\pi_{i-|x|,k}(u)\right|_a+\left|\pi_{i-|xu|,k}(y)\right|_a\\
        &\equiv^t\left |\pi_{i,k}(x)\right|_a+\left |\pi_{i-|x|,k}(v)\right|_a+\left|\pi_{i-|xv|,k}(y)\right|_a\\
        &=\left|\pi_{i,k}(xvy)\right|_a.
        \end{align}
  Since $a$ and $i$ were arbitrary, we conclude that 
	$\left|\pi_{i,k}(xuy)\right |_a \equiv^t \left |\pi_{i,k}(xvy)\right |_a$ for all $a\in \Sigma$ and $i\in [0,k-1]$. Hence, $xuy \congC_k^t xvy$.      
        \newline
\noindent (2)\quad 
	Assume that $s|k$ and $u \congC_k^t v$.  
	By assumption it follows $|u| \equiv_s |v|$.  
	Let $a\in \Sigma$ and $i\in [0,s-1]$. Then,
        \begin{align}
        \left |\pi_{i,s}(u)\right|_a&=\sum_{\substack{j\in [0,k-1],\\ 
		j\equiv_s i }}\left|\pi_{j,k}(u)\right|_a&\text{(Since $s|k$)}\\
        &\equiv^t \sum_{\substack{j\in [0,k-1],\\ 
		j\equiv_s i }}\left|\pi_{j,k}(v)\right|_a&\text{(Since $u \congC_k^t v$)}\\
		&=\left |\pi_{i,s}(v)\right|_a.&\text{(Since $s|k$)}
        \end{align}
      Therefore  $u \congC_s^t v$.
\end{appendixproof}

	\subsection{The $O(1)$ Congruence}

\label{sec:approx_k}
	In this section we introduce a refinement of $\congC_k^t$ and use it to 
	define the $O(1)$ congruence.

\begin{definition}[$\alpha_k,\omega_k,\kappa_k$]
	Let $k\in\bbN$ and let $u\in\Sigma^*$ be a word with $|u|\geq k$.
	By $\alpha_{k}(u)$ (\resp~$\omega_{k}(u)$) we denote
	the unique prefix (\resp~suffix) of $u$ of length $k$. 
In case $|u|>2k$ let $\kappa_{k}(u)$ denote the 
unique factor of $u$ such that $u = \alpha_{k}(u)\cdot \kappa_{k}(u) \cdot \omega_{k}(u)$.
\end{definition}

\ificalp
\else
We are now ready to define the relation $\approx_k^t$.

\begin{definition}[$\approx_k^t$]\label{FIL equiv}
	For each $k,t\in\bbN_{>0}$ define the relation 
	$\approx_k^t$ on $\Sigma^+$:
	$u\approx_k^t v$ if 
\begin{itemize}
    \item $u=v$, or 
     \item $|u|,|v| > 2k$ and
	     \begin{itemize}
		     \item $\alpha_k(u)=\alpha_k(v)$, 
		     \item $\omega_k(u)=\omega_k(v)$, 
			     and
		     \item $\kappa_k(u) \congC_{k}^t  \kappa_k(v)$.
	     \end{itemize}
     \end{itemize}
\end{definition}
 \fi   
Note that if $u\approx_k^t v$ and $|u|,|v|>2k$, then $|u|\equiv_k |v|$. 
Moreover, observe that each word of length $<3k$ constitutes a singleton
equivalence class of $\approx_k^t$.

\begin{lemmarep} Let $k,t\in \bbN_{>0}$.
\begin{enumerate}
   \item $\approx_k^t$ is a congruence on $\Sigma^+$.
   \item $\approx_k^t$ refines
	$\approx_s^t$ whenever
	$s|k$.\label{lemma:equiv_k_congruence}
    \end{enumerate}
\end{lemmarep}
\begin{appendixproof}
	\noindent(1)\quad To prove $\approx_k^t$ is an equivalence relation, 
	it suffices to prove that $\approx_k^t$ is an equivalence relation
	on $\Sigma^{> 2k}$. 
	On $\Sigma^{> 2k}$ we have that equality under $\alpha_k$,
	$\omega_k$ and $\congC_k^t$-equivalence under $\kappa_k$ are equivalence relations, respectively.
	Hence their conjunction defines an equivalence relation.

Let us argue that for all $u,v\in\Sigma^+$ and all
	$a\in\Sigma$ we have that $u\approx_k^t v$ implies $ua\approx_k^t va$.
	The case when $u=v$ is trivial.

	Let us assume $|u|,|v|>2k$. 
	Since $\alpha_k(u)=\alpha_k(v)$, trivially we have
	$\alpha_k(ua)=\alpha_k(u)=\alpha_k(v)=\alpha_k(va)$.
	Since $\omega_k(u)=\omega_k(v)$ we have 
	$\omega_k(ua)=\omega_{k-1}(u)a=\omega_{k-1}(v)a=\omega_{k}(va)$.
Since $|u|\equiv_k |v|$, 
	it follows $|ua|\equiv_k |va|$.
	To finish the proof 
	observe that $\kappa_k(ua) =  \kappa_k(u) \cdot a'$ and 
	$\kappa_k(va) =  \kappa_k(v) \cdot a'$,
	where $a'=(\omega_{k}(u))_{0}=(\omega_{k}(v))_{0}$. 
	Since $\kappa_k(u) \congC_k^t  \kappa_k(v)$, 
	it follows by Point (1) of Lemma~\ref{L cont} 
	that $\kappa_k(ua) \congC_k^t  \kappa_k(va)$. 
	Thus $ua \approx_k^t va$. By a symmetric argument $au \approx_k^t av$. Thus $\approx_k^t$ is a congruence.\\
    \newline
	\noindent (2)\quad Assume that $s|k$ and 
	$u\approx_k^t v$.
	Without loss of generality let us assume $|u|,|v|>2k$.
	Since $s|k$, equality under $\alpha_k$ implies
	equality under $\alpha_s$, equality under $\omega_k$
	implies equality under $\omega_s$. Finally, we have
    $\kappa_s(u) = 
		x\kappa_k(u)y$ and 
$\kappa_s(v) = 
		x\kappa_k(v)y$ where  $x=
\omega_{k-s}(\alpha_k(u))$ and $y=\alpha_{k-s}(\omega_k(u))$. 
Since $\kappa_k(u)\congC_k^t \kappa_k(v)$, we have $\kappa_k(u)\congC_s^t \kappa_k(v)$ 
	by Point (2) of Lemma~\ref{L cont}. 
	Also, $x\kappa_k(u)y\congC_s^t 
x\kappa_k(v)y$ since 
	$\congC_s^t$ is a congruence (by Point (1) of Lemma~\ref{L cont}). 
	Thus $\kappa_s(u) \congC_s^t \kappa_s(v)$. Hence $u \approx_s^t v$. 
\end{appendixproof}

We define the following congruence that is used to characterize the
regular languages with constant circuit complexity.
\begin{definition}[$O(1)$ Congruence] $\approx_k$ is $\approx_k^1$.
	\label{D O(1)}
\end{definition}

\subsection{The Pseudovariety $\ESL$}
\begin{samepage}
\label{sec:qesl}

We recall the pseudovariety $\ESL$ of finite semigroups, and also introduce corresponding lm-pseudovariety of stamps $\QESL$.  
	\begin{definition}\label{D ESL}
	By $\ESL$ we denote the pseudovariety of all finite
	semigroups $S$ that are characterized by the identities 
	(see \S\,\ref{sec:identity})
\begin{align}
epxqf&=epxxqf \label{eq:Xeq1}\\
epxyqf&=epyxqf\label{eq:Xeq2}
\end{align}
	for all $e,f\in E(S)$ and $p,q,x,y\in S$.
By $\QESL$ we denote the lm-pseudovariety of
	stamps $\varphi:\Sigma^+\to S$ whose stable semigroup 
	is in $\ESL$.
\end{definition}

\begin{example}
\label{example:const}
    The stable semigroups of the syntactic stamps of the 
	languages $L_1=(ab)^+, L_2=b^*a(a+b)^*$, given in \Cref{fig:ab} and \Cref{fig:1a} respectively, are both idempotent and commutative, i.e., 
	they are in the pseudovariety $\bfJone$. 
	Hence they are in $\ESL$ also.
    
    The stable semigroups of the syntactic stamps of  $L_3=a^+b^+$ and $L_4=b^*ab^*a(a+b)^*$ are given in \Cref{fig:aabb} and \Cref{fig:2a} respectively. The former does not satisfy the identity $epxyqf=epyxqf$ for the assignment $e=p=x=a,y=q=f=b$. The latter does not satisfy the identity 
$epxqf=epxxqf$ for the assignment $e=p=q=f=b, x=a$. Thus both are not in $\ESL$.
\begin{figure}
\centering
\captionbox{Syntactic and stable semigroups of $L_1=(ab)^+$.
	The stability index of $\eta_{L_1}$ is two.
	$\stab(\eta_{L_1})$ is in $\bfJone$, hence
	$\eta_{L_1}$ is in $\QESL$.\label{fig:ab}}
[.4\textwidth]{\begin{tikzpicture}
    \draw (0,0) rectangle (2,2);
    \draw (0,1) -- (2,1);
    \draw (1,0) -- (1,2);
    \draw (0.5,-1.5) rectangle (1.5,-.5);

\node at (0.5,0.5) {$ba$};
 \node at (.25,0.75) {*};
 \node at (0.5,1.5) {$a$};
 \node at (1.5,0.5) {$b$};
 \node at (1.5,1.5) {$ab$};
  \node at (1.25,1.75) {*};
\node at (1,-1) {$aa$};
 \node at (.75,-.75) {*};

  \draw (3.5,0) rectangle (4.5,1);
    \draw (5,0) rectangle (6,1);
    \draw (4.25,-1.5) rectangle (5.25,-.5);
     \node at (4,0.5) {$ba$};
     \node at (3.75,0.75) {*};

    \node at (5.5,0.5) {$ab$};
    \node at (5.25,0.75) {*};

    \node at (4.75,-1) {$aa$};
    \node at (4.5,0.-.75) {*};
\end{tikzpicture}}\qquad
\captionbox{Syntactic and stable semigroup of $L_2=b^*a(a+b)^*$. 
	The stability index of $\eta_{L_2}$ is one.
	$\stab(\eta_{L_2})$ is in $\bfJone$, hence $\eta_{L_2}$ is
	in $\QESL$.\label{fig:1a}}
[.4\textwidth]{\begin{tikzpicture}
    \draw (0,0) rectangle (1,1);
    \draw (0,-1.5) rectangle (1,-.5);
     \node at (0.5,0.5) {$b$};
     \node at (0.25,0.75) {*};

    \node at (0.5,-1) {$a$};
    \node at (0.25,-.75) {*};
\end{tikzpicture}}
\end{figure}

\begin{figure}
\centering\captionbox{Syntactic and stable semigroup of $L_3=a^+b^+$. 
	The stability index of $\eta_{L_3}$ is two.
	$\stab(\eta_{L_3})$ is not in $\ESL$ since $aaabbb=ab\neq ba =aababb$.
	Hence $\eta_{L_3}$ is not in $\QESL$.\label{fig:aabb}}
[.4\textwidth]{\begin{tikzpicture}
    \draw (0,0) rectangle (1,1);
    \draw (1.5,0) rectangle (2.5,1);
    \draw (.75,-1.5) rectangle (1.75,-.5);
    \draw (.75,-3) rectangle (1.75,-2);
     \node at (0.5,0.5) {$a$};
     \node at (0.25,0.75) {*};

    \node at (2,0.5) {$b$};
    \node at (1.75,0.75) {*};

    \node at (1.25,-1) {$ab$};
    \node at (1.25,-2.5) {$ba$};
    \node at (1,0.-2.25) {*};
\end{tikzpicture}}
\qquad\centering\captionbox{Syntactic and stable semigroup of $L_4=b^*ab^*a(a+b)^*$. 
	The stability index of $\eta_{L_4}$ is two.
	$\stab(\eta_{L_4})$ is not in $\ESL$ since 
	$bbabb=a\neq aa =bbaabb$.
Hence $\eta_{L_4}$ is not in $\QESL$.
	\label{fig:2a}}
[.4\textwidth]{\begin{tikzpicture}
    \draw (0,0) rectangle (1,1);
    \draw (0,-1.5) rectangle (1,-.5);
     \draw (0,-3) rectangle (1,-2);

      \node at (0.5,0.5) {$b$};
     \node at (0.25,0.75) {*};

    \node at (0.5,-1) {$a$};
    \node at (0.5,-2.5) {$aa$};
    \node at (0.25,-2.25) {*};
\end{tikzpicture}}
\end{figure}

\end{example}

The following lemma gives an alternative characterization of $\ESL$.
\begin{samepage}
	\begin{lemmarep}\label{L cont EJ1}
A finite semigroup $S$ belongs to $\ESL$ if, and only if, it satisfies the following property,
where $\pi:S^+\to S$ denotes the evaluation morphism:
		\begin{align}
\cont(u)\setminus \{e,f\}=\cont(v)\setminus \{e,f\}\quad\Longrightarrow\quad \pi(euf)=\pi(evf)
			\label{E pi}
		\end{align}
		for all $e,f\in E(S)$ and all words $u,v\in S^+$.
\end{lemmarep}
\end{samepage}
\begin{appendixproof}
	\noindent
($\Longleftarrow$)\quad Any semigroup $S$ satisfying the property stated in the lemma also 
	satisfies the identities (\ref{eq:Xeq1}) and (\ref{eq:Xeq2}), and thus belongs to $\ESL$.

\noindent
($\Longrightarrow$)\quad 
	Assume that $S$ is in $\ESL$. 
	In the following let $e,f\in E(S)$.
	For each $n\geq 2$, each $u=u_0\ldots u_{n-1} \in S^+$
	and each $i\in [0,n-2]$ let $u_{(i,i+1)}$ be the transposed word
\begin{align}
	u_{(i,i+1)}&=u_0\ldots u_{i-1}u_{i+1}u_iu_{i+2}\ldots u_{n-1}.
\end{align}
	We claim that \begin{align}
	\pi(euf)&=\pi(eu_{(i,i+1)}f). \label{E swap}
\end{align} 
Indeed we observe,
 \begin{align}
\pi(euf)&=
	 \pi(\underbracket[.5pt]{e}\,\underbracket[.5pt]{eu_0\ldots u_{i-1}}\,\underbracket[.5pt]{u_i}\,\underbracket[.5pt]{u_{i+1}}\,
\underbracket[.5pt]{u_{i+2}\ldots u_{n-1}f}\,\underbracket[.5pt]{f})
	 & &\text{(Since $e,f\in E(S)$)}\\
&=\pi(eeu_0\ldots u_{i-1}u_{i+1}u_{i}u_{i+2}\ldots u_{n-1}ff)&&(\text{Applying  $epxyqf=epyxqf$)}\\        
	 &=\pi(eu_{(i,i+1)}f).
\end{align}

Since adjacent transpositions generate all permutations of a word, we conclude that $\pi(euf)=\pi(eu'f)$ for any permutation $u'$ of $u$. 

Further if $u_i=u_{i+1}$, for $i\in [0,n-2]$, then
we claim that 
\begin{align}
    \pi(euf)&=\pi(eu_0\ldots u_iu_{i+2}\ldots u_{n-1}f). \label{E reduce}
\end{align} 
Indeed we note,
\begin{align}
    \pi(euf)&=\pi(\underbracket[.5pt]{e}\underbracket[.5pt]{eu_0\ldots u_{i-1}} \underbracket[.5pt]{u_{i}}\underbracket[.5pt]{u_{i+1}}\underbracket[.5pt]{u_{i+2}\ldots u_{n-1}f}
    \underbracket[.5pt]{f})\\
    &=\pi(eeu_0\ldots u_{i-1}u_{i}u_{i+2}\ldots u_{n-1}ff)&&
	\text{(Applying  $epxxqf=epxqf$)}\\
    &=\pi(eu_0\ldots u_{i-1}u_iu_{i+2}\ldots u_{n-1}f).
\end{align}

Assume that $\cont(u)\setminus\{e,f\}=\cont(v)\setminus\{e,f\}=
	\{s_1,\ldots, s_k\}\subseteq S$ for some $k\in\bbN$.
	Assume that $e\not=f$ (when $e=f$ the analysis is similar and hence omitted).
	Since $e^{|u|_e}s_1^{|u|_{s_1}}\ldots s_k^{|u|_{s_k}}f^{|u|_f}$
	is a permutation of $u$
	we deduce
\begin{align}
	\pi(euf)&=\pi(eeuff)\\
	&=\pi(ee^{|u|_e+1}s_1^{|u|_{s_1}}\ldots s_k^{|u|_{s_k}}f^{|u|_f+1}f)\\
&=\pi(ees_1\ldots s_kff)\\
    &=\pi(ee^{|v|_e+1}s_1^{|v|_{s_1}}\ldots s_k^{|v|_{s_k}}f^{|v|_f+1}f)\\
    &=\pi(eevff)\\
    &=\pi(evf).
\end{align}
\end{appendixproof}

\begin{corollary}
The following are identities of $\ESL$, where $e$ and $f$ range over idempotents.
\begin{enumerate}
\item $exf=exef=efxf$
\item$exyf=exeyf=exfyf$
\end{enumerate}
\label{E exyf}
\end{corollary}

The operation $\bfE\bfV$, called \emph{essentially} $\bfV$, for a pseudovariety 
$\bfV$ of semigroups is studied in~\cite{Grosshans21}, where the pseudovarieties 
$\bfV$ satisfying the equation $\bfE\bfV=\bfV \vee \bfL\bfI$ are characterized.
More details on the notation $\bfE\bfV$ can be found in Footnote~1 of the same reference.
Incidentally, $\bfJone$ is a variety that does not satisfy this identity (\cite{Grosshans21}, Proposition 13). Pseudovarieties of the form $\bfQ\bfE\bfV$ are further studied in the context of programs over monoids in \cite{GrosshansMS22}, where it is conjectured that  $\bfQ\bfE\bfV = \bfE\bfV * \bfMod$ if $\bfV$ is local. 
The equality $\bfQ\ESL= \ESL * \bfMod$ does hold, but we do not prove it in this paper.

\subsection{$O(1)$ Characterization Theorem}\label{S Main Result const}
\label{sec:O(1)}

Before stating the main result of this section we require 
a further definition.

\begin{definition}\label{D Cont}
	For each $k$-tuple $(\Gamma_0,\dots,\Gamma_{k-1})$ of 
	non-empty subsets of a 
	finite alphabet $\Sigma$
	we define
	\begin{align}\Cont_k(\Gamma_0,\dots,\Gamma_{k-1})=\Sigma^{k\N_{>0}}\cap
	\bigcap_{i\in[0,k-1]}\cont_{i,k}^{-1}(\Gamma_i).
	\end{align}
\end{definition}

\begin{samepage}
	\begin{theorem}[$O(1)$ Characterization]\label{T const main}
Let $L\subseteq\Sigma^+$ be a regular language.
Then the following statements are equivalent:
	\begin{enumerate}
		\item $\opc(L)\in O(1)$.
		\item $\opc(L)\in o(\log n)$.
		\item $\eta_L\in\QESL$.
		\item $L$ is a union of $\approx_k$-classes
			for some $k\in\bbN_{>0}$.
        \item $L$ is a disjoint finite union of finite languages and languages of the 
		form $u\Cont_k(\Gamma_0,\ldots, \Gamma_{k-1})v$, where $k\in\bbN_{>0}$,
			$\emptyset\not=\Gamma_0,\dots,\Gamma_{k-1}\subseteq\Sigma$, and $u ,v\in \Sigma^+$.
		\item $L$ is a finite Boolean combination of finite languages and languages
			of the form
			$u(\Gamma_{0}\dots \Gamma_{k-1})^+v$,
			where $k\in\bbN_{>0}$, $\emptyset\not=\Gamma_0,\dots,\Gamma_{k-1}\subseteq\Sigma$ 
			and $u,v\in\Sigma^+$.
		\item $L$ is definable in 
			$\rmFO^1[\Sigma,  +\omega, \min, \max, \rmmod]$.
        \item $L$ is definable in 
			$\rmFO^1[\Sigma, \mathrm{reg}]$.
        \item $L$ is definable in 
			$\rmFO^1[\Sigma, \mathrm{arb}]$.
            
	\end{enumerate}
\end{theorem}
\end{samepage}
\begin{example}
	By the above theorem the languages $L_1=(ab)^+$ and $L_2=b^*a(a+b)^*$ from 
	\Cref{example:const} have constant circuit complexity since
	$\eta_{L_1},\eta_{L_2}\in\QESL$.
	The languages $L_3=a^+b^+$ and $L_4=b^*ab^*a(a+b)^*$ have circuit complexity 
	$\Omega(\log n)$ --- recall that $f\in o(g)$ if, and only if, $f\not\in\Omega(g)$.
	In fact, one can prove that $\opc(a^+b^+)\in\Theta(n)$ 
	and that $\opc(b^*ab^*a(a+b)^*)\in\Theta(\log n)$. 
\end{example}

\begin{proof}
We show 
$(1) \Longrightarrow (2) \Longrightarrow (3) \Longrightarrow 
(4) \Longrightarrow (5) \Longrightarrow (6)
\Longrightarrow (7)\Longrightarrow (8)\Longrightarrow (9)\Longrightarrow (1).$
\medskip

	\begin{samepage}
\noindent $(1) \Longrightarrow  (2)$\quad
Trivial.
	\end{samepage}
\medskip

	\begin{samepage}
\noindent 
	$(2) \Longrightarrow  (3)$\quad
Follows from Theorem~\ref{T const lower} proved in \S~\ref{S Lower Bound const}.
	\end{samepage}
\medskip

	\begin{samepage}
		\noindent
	$(3) \Longrightarrow  (4)$\quad
		By Corollary~\ref{C const main} of Proposition~\ref{P const tech} proved in 
		\S\,\ref{S Congruence const}.
	\end{samepage}
\medskip

	\begin{samepage}
		\noindent
\noindent$(4) \Longrightarrow  (5)$\quad Each equivalence class of $\approx_k$ is either a singleton 
		set, or a set of the form  $uKv$ where $u,v \in \Sigma^{k}$ and 
\begin{align}
	K&= \Sigma^{r+k\bbN_{>0}} \cap\bigcap_{i\in[0,k-1]} 
	\cont_{i,k}^{-1}(A_i)\\
	&=\{ w\in \Sigma^{+} : |w| \equiv r \bmod k,\cont_{i,k}(w)=A_i 
	\text{ for all } i\in[0,k-1]\}
\end{align}
		for some $r\in[0,k-1]$ and 
		$\emptyset\not=A_i\subseteq \Sigma$ for each $i\in[0,k-1]$. 
For each $A\subseteq\Sigma$ and each $a\in A$ we define
		$A\dotminus a$ as $A\setminus\{a\}$ if $|A|>1$ and $\{a\}$ otherwise.
We write $K$ as a disjoint union $K'$ of languages by performing a case analysis
on the last $r$ letters, and whether or not they appear in their respective residue classes
		previously:
\begin{align}
    K'\quad =\quad
    \biguplus_{\substack{
    x=a_0\ldots a_{r-1}\\ ~\in A_0\ldots A_{r-1}
    }}\ {\mathlarger{\biguplus}_{(X_0,\dots,X_{r-1})\in\prod_{i=0}^{r-1} 
	\{A_i, A_i\dotminus a_i\}}}
	 \Cont_k(X_0,\ldots, X_{r-1},A_r,\ldots, A_{k-1}) \cdot x.
\end{align}
		We claim that $K=K'$.
To show the right-to-left inclusion assume some $w\in K'$.
		Then $|w|\equiv_k r$  and $\cont_{i,k}(w)=A_i$ for
		each $i\in[0,k-1]$. 
		Hence $w\in K$.
Conversely assume that $w\in K$. 
		Then $w=vx$
		for words $v\in\Sigma^{k\N_{>0}}$ and $x\in\Sigma^r$.
		Let us write $x=a_0\dots a_{r-1}$, where $a_i\in A_i$ for each $i\in[0,r-1]$.
		Let $X_i=\cont_{i,k}(v)\subseteq A_i$ for each $i\in[0,r-1]$.
		Since $\cont_{i,k}(w)\supseteq\cont_{i,k}(v)\supseteq\cont_{i,k}(w) 
		\setminus \{a_i\}=A_i \setminus \{a_i\}$, in fact we have 
		$X_i\in\{A_i,A_i\dotminus a_i\}$ for each $i\in[0,r-1]$.
		Thus, $v\in\Cont_k(X_0,\dots,X_{r-1},A_r,\dots,A_{k-1})$
		and we obtain that $w=vx\in K'$, as required.

Therefore we can write $uKv$ as a disjoint union of languages
		of the form $uK''xv$, where $K''$ is a language
		of the form $\Cont_k(\Gamma_0,\dots,\Gamma_{k-1})$,
		where $\Gamma_i\subseteq\Sigma$ for each $i\in[0,k-1]$.
	\end{samepage}

	\begin{samepage}
		\medskip

		\noindent
		\noindent$(5) \Longrightarrow  (6)$\quad Since $\Gamma_0,\dots,\Gamma_{k-1}$
		are non-empty we have
\[\Cont_k(\Gamma_0,\ldots, \Gamma_{k-1})= 
		\bigcap_{i\in[0,k-1]} \left[(\Sigma^i \Gamma_{i}\Sigma^{k-i-1})^+ 
		\mathbin{\big\backslash}
		\bigcup_{\emptyset\not=\Gamma_i'\subsetneq \Gamma_i}(\Sigma^i \Gamma_{i}'
		\Sigma^{k-i-1})^+\right].
\] 
	\end{samepage}

	\begin{samepage}
\noindent$(6) \Longrightarrow  (7)$\quad
Since $\rmFO^1[\Sigma, +\omega, \min, \max, \rmmod]$ is closed under Boolean operations, it suffices to show that there is a formula defining a given word $w$ as well as the sets of the form 
$L=a_0\ldots a_{m-1} (\Gamma_{0}\Gamma_{1}\dots \Gamma_{k-1})^+b_0\ldots b_{n-1}$. \Cref{example:logic-const} gives a formula $\varphi$ defining a word $w$, 
		and $\varphi_L$ below defines $L$.
		Let us first define the auxiliary formula
\begin{align}
	&x \in [\min + m, \max - n] \defeq \bigwedge_{ i\in[0,m-1]} \min+i \neq x~\wedge \bigwedge_{i\in[0,n-1]} x+i \neq \max.
\end{align}
	For each $i\in[0,k-1]$ let $i'=(m+i)\bmod k\in[0,k-1]$.
		We define
\begin{align}
	\varphi_L&\defeq \bigwedge_{i\in[0,m-1]} \exists x\, (\min+i=x \wedge P_{a_i}(x)) \wedge \bigwedge_{i\in[0,n-1]} \exists x\, (x+i = \max \wedge P_{b_{n-1-i}}(x))\label{E formula 1}\\
	&~~\wedge\bigwedge_{i\in[0,k-1]} \forall x\,(x \in [\min + m, \max - n] \wedge x \equiv i' \bmod k \to \bigvee_{a\in \Gamma_i} P_{a}(x))\label{E formula 2}\\
	&~~\wedge\max\equiv m+n-1\bmod k\ \wedge
	\exists x\, (x\in[\min + m, \max - n])\label{E formula 3}
\end{align}
	\end{samepage}
The part (\ref{E formula 1}) expresses that our word begins
with $a_0\dots a_{m-1}$ and ends with $b_0\dots b_{n-1}$.
The part (\ref{E formula 2}) expresses that 
a position $p\equiv m+i \bmod{k}$ outside the prefix and suffix must carry a letter from 
$\Gamma_i$.
The part (\ref{E formula 3}) expresses that the length of the word is 
in residue class $m+n\bmod k$ and that the word has length at least $m+n+1$.
\medskip

\begin{samepage}

\noindent $(7) \Longrightarrow  (8)$\quad
Trivial.
\medskip
\end{samepage}

\begin{samepage}
\noindent $(8) \Longrightarrow  (9)$\quad
Trivial.
\medskip
\end{samepage}

\begin{samepage}
    
\noindent$(9) \Longrightarrow  (1)$\quad It suffices to show that every language, not necessarily regular,  defined by an  $\rmFO^1[\,\Sigma,\mathrm{arb}]$ sentence can be computed by a circuit family of size $O(1)$.  Below we freely use the fact that languages 
	recognized by circuit families of constant size are closed under Boolean operations.

Let $\varphi$ be an $\rmFO^1[\, \Sigma, \mathrm{arb}]$ sentence defining $L$ and 
let $x$ be the unique variable used by $\varphi$. Without loss of generality, we can assume that $\varphi$ is a Boolean combination of formulas of the form $Q x\, \psi$, where $Q\in \{\exists, \forall\}$ and $\psi$ is quantifier-free. 
	It suffices to construct a circuit family of 
	constant size for each formula of the form $Q x\, \psi$, where $Q\in \{\exists, \forall\}$. Since $\forall x\, \psi \equiv \neg \exists x\, \neg \psi$ we can limit our attention to formulas of the form $\exists x\, \psi$. 

Consider a formula $\exists x\, \psi \in \rmFO^1[\,\Sigma, \mathrm{arb}]$,
	where $\psi$ is quantifier-free. We write $\psi$ in disjunctive normal form and distribute the quantifier over the disjunction. 
	Again, it suffices to construct a circuit family 
	$\calC=(C_n)_{n\in\bbN_{>0}}$
	of constant size computing the set of words
	satisfying a formula of the form $\exists x\, \xi$ where $\xi$ is a conjunction of literals. 
	Let $c_0,\ldots, c_{\ell-1}$ be the constants used by $\xi$. Note that the negative literal 
	$\neg P_a(x)$ is equivalent to the disjunction $\vee_{b \in \Sigma \setminus \{a\}} P_b(x)$. 
	By grouping the numerical predicates together, and by replacing the negative letter predicates, it is easy to see that the formula $\exists x\, \xi$ is equivalent to a 
	formula of the form
\begin{equation}
	\exists x\, \xi'(x,c_0,\ldots, c_{\ell-1}) \wedge \bigvee_{a\in \Gamma} P_a(x) \wedge \bigwedge_{i \in [0,\ell-1]} \bigvee_{a\in \Lambda_i} P_a(c_i),
\end{equation}
where $\Gamma,\Lambda_i \subseteq \Sigma$ for $i\in [0,\ell-1]$ and $\xi'$ is a conjunction of numerical predicates or their negations. 
Let $n\in \bbN_{>0}$, and let 
	$m_0,\ldots, m_{\ell-1}\in [0,n-1]$ be the interpretation of the constants 
	$c_0,\ldots, c_{\ell-1}$ for words of length $n$. 
	Let $P\subseteq [0,n-1]$ be the set of positions that satisfy the 
	formula $\xi'(x,m_0,\ldots, m_{\ell-1})$ in a word of length $n$. Then the 
	$n$-circuit $C_n$ realizing 
\begin{equation}
\bigvee_{p\in P} (p, \Gamma) \wedge \bigwedge_{i\in [0,\ell-1]} (m_i, \Lambda_i) 
\end{equation}
computes
the words of length $n$ that satisfy the formula $\exists x\, \xi$. 
The family $(C_{n})_{n\in \bbN_{>0}}$ accepts the set of all words satisfying the formula $\exists x\, \xi$ and it is of constant size. Hence the direction follows.
\end{samepage}

\end{proof}

\subsubsection{The Neutral Letter Languages of Constant Circuit Complexity}

\label{neutral letter}
A language $L\subseteq \Sigma^+$ has a \emph{neutral letter} if there exists a letter $c\in \Sigma$ such that $uv\in L$ if, and only if, $ucv\in L$ for all 
$u,v\in \Sigma^*$ such that $uv\in\Sigma^+$. 
A language having a neutral letter is called a \emph{neutral letter language}.
The languages $b^*a(a+b)^*
$ and $b^*ab^*a(a+b)^*
$ from \Cref{example:const} have the neutral letter $b$, whereas $(ab)^+$ and $a^+b^+$  from the same example do not.

The syntactic semigroup of a neutral letter language $L$ with the neutral letter $c$ is a monoid with the identity $\eta_L(c)$. Moreover, the stable semigroup of its syntactic morphism is the entire syntactic semigroup.
A language $L\subseteq \Sigma^+$ is \emph{commutative} if  
$uxyv\in L$ if, and only if, $uyxv \in L$ for all $u,v,x,y\in \Sigma^*$. 
It is \emph{idempotent} if $uxv\in L$ if, and only if, $uxxv\in L$ for 
all $u,v\in \Sigma^*$ and $x\in \Sigma^+$.
Note that if $S(L)$ is both idempotent and commutative, then so is $L$.

\begin{corollary}\label{C Neutral Letter}
    Let $L\subseteq \Sigma^+$ be a regular neutral letter language. Then, 
	the following statements are equivalent:
    \begin{enumerate}
	    \item $\opc(L)\in O(1)$.
        
        \item  $L$ is idempotent and commutative.
          \item $L$ is a disjoint union of $\congC_1^1$-classes.
     
	  \item $L$ is a Boolean combination of languages of the form $\Gamma^+$ 
		with $\Gamma\subseteq \Sigma$.
\item $L$ is definable in $\rmFO^1[\Sigma]$.   \end{enumerate}
    
\end{corollary}
\begin{proof}
Equivalence of $(2),(3),(4),(5)$ is given by  \cite{SurveySmall}, Theorem 1.

\noindent
	$(5)\Longrightarrow(1)$\quad Follows immediately from $(9)\Longrightarrow(1)$ of Theorem~\ref{T const main}.

\noindent$(1) \Longrightarrow  (2)$ Assume that $L$ is a regular neutral letter language.
	By Theorem~\ref{T const main} it suffices to prove that if 
$\eta_L\in \QESL$, then $L$ is idempotent and commutative.
Since $L$ has a neutral-letter, $S(L)$ is a monoid and moreover $\stab(\eta_L)=S(L)$. 
	Therefore $S(L)$ is in the semigroup variety $\bfE\bfJone$. 
	Since every monoid in $\bfE\bfJone$ is idempotent and commutative, 
	the claim follows.
\end{proof}

\subsection{From Algebra to Circuit Complexity Lower Bounds}
\label{S Lower Bound const}

This section is devoted to proving the following theorem.
\begin{theorem}\label{T const lower}
Let $L$ be a regular language whose syntactic morphism
	is not in $\QESL$. Then $\opc(L)\in\Omega(\log n)$.
\end{theorem}

\begin{definition}\label{D simC}
	Let $n\in\bbN_{>0}$ and let 
	$C=(G,\prec,\lambda,g_{out})$ be an $n$-circuit over the 
alphabet $\Sigma$. We define the following equivalence relation $\sim_C$ 
	on $[0,n-1]$:
	\begin{align}
		i\sim_C j\quad\defeq\quad
		\forall g\in G_1\forall \Gamma\subseteq\Sigma:
		(i,\Gamma)\prec g\Longleftrightarrow (j,\Gamma)\prec g
	\end{align}
\end{definition}

For convenience we identify an input gate $g\in G_0$ with its unique label
$\lambda(g)\in[0,n-1]\times 2^{\Sigma}$.
\begin{remark}\label{R simC}
Let $C$ be an $n$-circuit and let $P\subseteq[0,n-1]$ be an equivalence class of 
	$\sim_C$.
	Then for each $g\in G_1$ and each $\Gamma\subseteq\Sigma$, either 
	$(i,\Gamma)\prec g$ for all $i\in P$, or 
	$(i,\Gamma)\not \prec g$ for all $i\in P$. 
\end{remark}
The following proposition states that the circuit is invariant under content-preserving
relabeling of $\sim_C$-equivalent positions.

\begin{propositionrep}\label{P swap}
Let $C$ be an $n$-circuit over $\Sigma$ and let 
	$P\subseteq[0,n-1]$ be a set of $\sim_C$-equivalent positions.
	Let $w=a_0\dots a_{n-1}$ and $w'=a_0'\dots a_{n-1}'\in\Sigma^n$ be words 
	such that
	$a_i=a_i'$ for all $i\in[0,n-1]\setminus P$ and
	$\{a_i\mid i\in P\}=\{a_i'\mid i\in P\}$.
	Then $w\in L(C)$ if, and only if, $w'\in L(C)$.
\end{propositionrep}
\begin{appendixproof}
	Let $C=(G,\prec,\lambda,g_{out})$ be an $n$-circuit
and let $P\subseteq[0,n-1]$ be an equivalence class of 
	$\sim_C$. Let $w$ and $w'$ be as in the statement of the proposition. 
		For each equivalence
	class $J$ of $\sim_C$ and $\Gamma \subseteq \Sigma$  we have:
	\begin{align}
		w\models \bigvee_{j \in J} (j,\Gamma) 
		&\quad\Longleftrightarrow\quad 
\exists s \in J : a_s \in \Gamma\\[-10pt]
        &\quad\Longleftrightarrow\quad
\exists t \in J : a_t' \in \Gamma && \text{ (By assumption on $w$ and $w'$)}\\
&\quad\Longleftrightarrow\quad
		w'\models \bigvee_{j \in J} (j,\Gamma) 
        \label{eq:JGammavee}
	\end{align}
Likewise, 
\begin{align}
		w\models \bigwedge_{j \in J} (j,\Gamma) 
		&\quad\Longleftrightarrow\quad 
\forall s \in J : a_s \in \Gamma\\[-10pt]
        &\quad\Longleftrightarrow\quad
\forall t \in J : a_t' \in \Gamma && \text{ (By assumption on $w$ and $w'$)}\\
&\quad\Longleftrightarrow\quad
		w'\models \bigwedge_{j \in J} (j,\Gamma) 
        \label{eq:JGamma}
	\end{align}
    
By induction on
	$\prec^+$ we show that for all $g\in G_1$ we have
	$w\models g$ if, and only if, $w'\models g$.
    For the induction base, without loss of generality assume that $g\in G_1$ is an $\vee$-gate only connected to the input gates. Clearly, if $w\models g$ then there 
	exists $(j,\Gamma)\prec g$ such that $w \models (j,\Gamma)$. 
	Thus $w\models \bigvee_{i \sim_C j} (i,\Gamma)$ and therefore
$w'\models \bigvee_{i \sim_C j} (i,\Gamma)$ by (\ref{eq:JGammavee}).
    By  Remark~\ref{R simC} we have $(i,\Gamma)\prec g$ for all $i \sim_C j$. 
	Thus, $w' \models g$. 
	By a symmetric argument, if 
	$w' \models g$ then $w \models g$. When $g$ is an $\wedge$-gate the argument 
	is analogous. This completes the base case. 
    
For the induction step assume  an $\vee$-gate $g\in G_1$.
If $w\models g$ then $w\models g'$ for some $g' \prec g$. If $g' \in G_1$, then it follows that $w' \models g'$ by IH and hence $w' \models g$.  Otherwise if $g'=(j,\Gamma) \in G_0$, then $w \models \bigvee_{i \sim_C j} (i,\Gamma)  \Longleftrightarrow w' \models \bigvee_{i \sim_C j} (i,\Gamma)$ and  $w' \models g$ by repeating the argument in the base case. Thus $w' \models g$. By a symmetric argument, if $w' \models g$ then $w\models g$.
The case when $g$ is an $\wedge$-gate is analogous. 
\end{appendixproof}

\begin{definition}[Separating Circuit] \label{D separator}
Let $n\in\bbN_{>0}$ and let $L,L'\subseteq\Sigma^n$ be languages.
	An $n$-circuit $C$ over $\Sigma$
	\emph{separates} $(L,L')$ if $L\subseteq L(C)$ and
	$L'\cap L(C)=\emptyset$.
\end{definition}

The heart of our lower bound is the following proposition.
The two identities in the statement are obeyed by
all semigroups in $\ESL$ by Lemma~\ref{L cont EJ1}.

\begin{samepage}
	\begin{proposition}\label{P eqconst}
Let $L\subseteq\Sigma^+$ be a regular language. 
If $\stab(\eta_L)$
falsifies one or both of the identities below,
		then $\opc(L)\in\Omega(\log n)$.
		Let $e,f\in E(\stab(\eta_L))$ denote idempotents.
		\begin{enumerate}
			\item $epxeyqf=epyexqf$
			\item $epexeqf=epexexeqf$
		\end{enumerate} 
\end{proposition}
\end{samepage}
\begin{proof}
Let us fix a regular language $L\subseteq\Sigma^+$,
let 
$S=\stab(\eta_L)$ be the stable semigroup of $\eta_L$,
	and let $s$ denote its stability index.
In the following let $\Gamma=\Sigma^s$ and let
	$\varphi=\eta_L\circ\flat_s:\Gamma^+\rightarrow S$.

\medskip

	\noindent (1)\quad 
	Assume that $S$ does {\em not} satisfy the identity $epxeyqf=epyexqf$
	for some $e,f\in E(S)$ and $p,q,x,y\in S$.
	Since $\varphi^{-1}(epxeyqf)\reduce L$ by Proposition~\ref{P stable reduce}
		it suffices to prove 
	$\opc(\varphi^{-1}(epxeyqf))\in\Omega(\log n)$
	by Point (2) of Theorem~\ref{T asymptotics}.

	Let us fix letters $a_e,a_f,a_p,a_q,a_x,a_y\in\Gamma$ such that
	$\varphi(a_b)=b$ for all $b\in\{e,f,p,q,x,y\}$.
	In the following assume $n\geq 6$. We define 
	$P_n=\{(i,j)\in[1,n-1]^2\mid i+j<n-4\}$
	and for each $(i,j)\in P_n$ we define the words
	\begin{align}
		u_{n,i,j}&=a_e^ia_pa_xa_e^{j}a_ya_qa_f^{n-(i+j+4)}
		\in\varphi^{-1}(epxeyqf)\\
		v_{n,i,j}&=a_e^ia_pa_ya_e^{j}a_xa_qa_f^{n-(i+j+4)}
		\not\in\varphi^{-1}(epxeyqf)
	\end{align}
	We remark that $|u_{n,i,j}|=|v_{n,i,j}|=n$.
	Let \begin{align}
		U=\{u_{n,i,j}\mid (i,j)\in P_n\}\text{ and }
	V=\{v_{n,i,j}\mid (i,j)\in P_n\}.
	\end{align}
	Clearly, every $n$-circuit computing $\varphi^{-1}(epxeyqf)\cap\Gamma^n$ must
	separate $(U,V)$.

	To prove Point (1) it thus suffices to show that every
	$n$-circuit $C$ that separates $(U,V)$
	satisfies $|C|\geq\log_{2^{2^{|\Gamma|}}}(n/6)$.
	For the sake of contradiction, let us assume that there is
	an $n$-circuit 
	$C=(G,\prec,\lambda,g_{out})$
	that separates $(U,V)$
	and that satisfies $|C|<\log_{2^{2^{|\Gamma|}}}(n/6)$.
	Note that $\sim_C$ from Definition~\ref{D simC} has at most
	\begin{align}
		\left(2^{2^{|\Gamma|}}\right)^{|G_1|}
	\end{align}
	equivalence classes.
	Since 
	\begin{align}
		|C|=|G_1|<\log_{2^{2^{|\Gamma|}}}(n/6)
\end{align}
	by assumption 
	we have 
	\begin{align}
	\left(2^{2^{|\Gamma|}}\right)^{|G_1|}<n/6.
\end{align}
By the pigeonhole principle among the positions $[0,n-1]$ there exist seven positions
that are all $\sim_C$-equivalent.
In particular, discarding the first two and last two positions,
among the positions $[2,n-3]$ there are at least three $\sim_C$-equivalent positions left.
Thus,
 there exist two $\sim_C$-equivalent positions $k,\ell\in[2,n-3]$ such that $k+2\leq\ell$.
Let $i=k-1$ and $j=\ell-k-1$. 
We claim that $(i,j)\in P_n$. Indeed, $i=k-1\geq 2-1=1\leq \ell-k-1=j$
and $i+j=(k-1)+(\ell-k-1)=\ell-2\leq n-5<n-4$.
We observe that
the letters $a_x$ and $a_y$ (\resp~$a_y$ and $a_x$) appear 
at position $k=i+1$ and $\ell=i+j+2$ respectively in the word $u_{n,i,j}$
(\resp~in the word $v_{n,i,j}$).
Moreover, positions $k$ and $\ell$ are the only positions
where $u_{n,i,j}$ and $v_{n,i,j}$ differ.
Since $k\sim_C \ell$ we have $u_{n,i,j}\in L(C)$ if, and only if,
$v_{n,i,j}\in L(C)$ by Proposition~\ref{P swap}, clearly
contradicting our assumption that $C$ separates $(U,V)$.
\medskip

	\noindent (2)\quad 
Let us now assume that the stable semigroup $S$ of the syntactic
morphism $\eta_L$ does not satisfy the identity $epexeqf=epexexeqf$.
As before, we choose letters $a_b\in\Gamma$ for each $b\in\{e,p,x,q,f\}$.
Assume $n>6$ in the following.
For each 
$i\in[1,n-6]$ we define the words
\begin{align}
	u_{n,i}=a_ea_pa_e^ia_xa_e^{n-(i+5)}a_qa_f \in \varphi^{-1}(epexeqf).
\end{align}
Moreover we define $Q_n=\{(i,j)\in[1,n]^2\mid i+j<n-6\}$ and for each 
$(i,j)\in Q_n$ the words
\begin{align}
	w_{n,i,j}=a_ea_pa_e^ia_xa_e^ja_xa_e^{n-(i+j+6)}a_qa_f\not\in \varphi^{-1}(epexeqf).
\end{align}
Fix any $(i,j)\in Q_n$.
Observe that $|u_{n,i}|=|w_{n,i,j}|=n$ and moreover $u_{n,i}$ 
can be obtained from $w_{n,i,j}$ by substituting
the letter $a_x$ at position $i+j+3$ by $a_e$.
Let 
\begin{align}
	X=\{u_{n,i}\mid i\in[1,n-6]\}\text{ and }
Y=\{w_{n,i,j}\mid (i,j)\in Q_n\}.
\end{align}
Thus, it suffices to show that each $n$-circuit $C$ that
separates $(X,Y)$ satisfies 
\begin{align}
	|C|\geq \log_{2^{2^{|\Gamma|}}}(n/8).
\end{align}
Towards a contradiction, let us assume that there is an $n$-circuit $C$ that 
separates $(X,Y)$ and that 
satisfies \begin{align}
	|C|<\log_{2^{2^{|\Gamma|}}}(n/8).  \end{align}
The following analysis is similar to the one in Point (1).
	By the pigeonhole principle there exist nine distinct
	positions that are all $\sim_C$-equivalent.
	Discarding the first three and the last three positions, there exist
	three $\sim_C$-equivalent positions $k<\ell<m$ in $[3,n-4]$, 
	so in particular $2\leq m-k$.
	Let $i=k-2$ and let $j=m-k-1$.
	We claim that $(i,j)\in Q_n$. Firstly, we have $i,j\geq 1$ since
	$i=k-2\geq 3-2=1=2-1\leq m-k-1=j$.
	Secondly, we have
	$i+j=(k-2)+(m-k-1)=m-3\leq n-7<n-6$.
	Recall that $u_{n,i}$ can be obtained from $w_{n,i,j}$ by
	replacing (the rightmost occurrence of) the letter $a_x$ 
	at position $m=i+j+3$ by $a_e$. Hence $u_{n,i}$ and $w_{n,i,j}$ 
	agree on all positions
	except for position $m=i+j+3$. In particular they agree
	on positions $k$ and $\ell$ that are labeled by the letters
	$a_x$ and $a_e$ respectively.
	Wrapping up, the words $u_{n,i}$ and $w_{n,i,j}$ agree on all positions
	in $[0,n-1]\setminus\{k,\ell,m\}$ and constitute the set 
	of letters $\{a_e,a_x\}$ on the set of $\sim_C$-equivalent
	positions $\{k,\ell,m\}$.
	By Proposition~\ref{P swap}
	we have $u_{n,i}\in L(C)$ if, and only if, $w_{n,i,j}\in L(C)$,
	contradicting the assumption that
	$C$ indeed separates $(X,Y)$.
	\end{proof}

The next lemma introduces further semigroup identities that
are helpful in linking the identities of $\ESL$ with the more specific ones that appeared
in Proposition~\ref{P eqconst}.

\begin{samepage}
	\begin{lemmarep}\label{L identities}
Let $e,f$ denote idempotents.
The following implications of identities hold true 
on all finite semigroups:
	\begin{align}
		&&  epxeyqf&=epyexqf\label{E exy}\\
		&\Longrightarrow&  epxyqf&=epyxqf\label{E xy}\\
		&\Longrightarrow& epxyqf&=epexeyeqf\label{E exeye}
	\end{align}
\end{lemmarep}
\end{samepage}
\begin{appendixproof}
	Let us first prove~(\ref{E exy}) implies
	(\ref{E xy}).
	We observe
	\begin{align}
		epxyqf &=\underbracket[.5pt]{e}\,\underbracket[.5pt]{e}\,\underbracket[.5pt]{e}\,
		\underbracket[.5pt]{e}\,\underbracket[.5pt]{p}\,\underbracket[.5pt]{xyq}\,\underbracket[.5pt]{f}&&\text{}\label{E rewrite first}\\
		&=\underbracket[.5pt]{e}\,\underbracket[.5pt]{e}\,\underbracket[.5pt]{p}\,
		\underbracket[.5pt]{e}\,\underbracket[.5pt]{e}\,\underbracket[.5pt]{xyq}\,\underbracket[.5pt]{f}&&\text{(Applying
		the identity~(\ref{E exy}))}\\
		&=\underbracket[.5pt]{ee}\,\underbracket[.5pt]{p}\,\underbracket[.5pt]{e}\,\underbracket[.5pt]{e}\,\underbracket[.5pt]{x}\,
		\underbracket[.5pt]{yq}\,\underbracket[.5pt]{f}\\
		&=\underbracket[.5pt]{ee}\,\underbracket[.5pt]{p}\,\underbracket[.5pt]{x}\,\underbracket[.5pt]{e}\,\underbracket[.5pt]{e}\,
		\underbracket[.5pt]{yq}\,\underbracket[.5pt]{f}&&\text{(Applying
		the identity~(\ref{E exy}))}\\
&=\underbracket[.5pt]{e}\,\underbracket[.5pt]{p}\,\underbracket[.5pt]{x}\,\underbracket[.5pt]{e}\,
		\underbracket[.5pt]{y}\,\underbracket[.5pt]{q}\underbracket[.5pt]{f}&&\text{(Rewriting $ee=e$)}
		\label{E rewrite last}\\
		&=\underbracket[.5pt]{e}\,\underbracket[.5pt]{p}\,\underbracket[.5pt]{y}\,\underbracket[.5pt]{e}\,
		\underbracket[.5pt]{x}\,\underbracket[.5pt]{q}\underbracket[.5pt]{f}&&\text{(Applying
		the identity~(\ref{E exy}))}
		\\
    &=epyxqf.&&\text{(Applying (\ref{E rewrite first}--\ref{E rewrite last})
		in reverse)}
	\end{align}
	This shows that the identity~(\ref{E exy}) implies (\ref{E xy}).
Next we prove~(\ref{E xy}) implies
	(\ref{E exeye}).

\begin{align}
epxyqf&=\underbracket[.5pt]{e}\,\underbracket[.5pt]{e}\,\underbracket[.5pt]{e}\,\underbracket[.5pt]{p}\,
		\underbracket[.5pt]{xyq}\,\underbracket[.5pt]{f}
		\\
        &=\underbracket[.5pt]{e}\,\underbracket[.5pt]{e}\,\underbracket[.5pt]{p}\,\underbracket[.5pt]{e}\,
		\underbracket[.5pt]{xyq}\,\underbracket[.5pt]{f}
		&&\text{(Applying
		the identity~(\ref{E xy}))}
		\\
        &=\underbracket[.5pt]{e}\,\underbracket[.5pt]{pe}\,\underbracket[.5pt]{e}\,\underbracket[.5pt]{x}\,
		\underbracket[.5pt]{yq}\,\underbracket[.5pt]{f}
		\\
        &=\underbracket[.5pt]{e}\,\underbracket[.5pt]{pe}\,\underbracket[.5pt]{x}\,\underbracket[.5pt]{e}\,
		\underbracket[.5pt]{yq}\,\underbracket[.5pt]{f}
		&&\text{(Applying
		the identity~(\ref{E xy}))}
		\\
        &=\underbracket[.5pt]{e}\,\underbracket[.5pt]{pexe}\,\underbracket[.5pt]{e}\,\underbracket[.5pt]{y}\,
		\underbracket[.5pt]{q}\,\underbracket[.5pt]{f}
		\\
    &=\underbracket[.5pt]{e}\,\underbracket[.5pt]{pexe}\,\underbracket[.5pt]{y}\,\underbracket[.5pt]{e}\,
		\underbracket[.5pt]{q}\,\underbracket[.5pt]{f}&&\text{(Applying
		the identity~(\ref{E xy}))}
		\\
    &=epexeyeqf.
\end{align}
\end{appendixproof}

We are now ready to give the proof of Theorem~\ref{T const lower}.

\subsubsection{Proof of Theorem~\ref{T const lower}}
\label{S T const proof}
\begin{proof}
	Let $L\subseteq\Sigma^+$ be regular.
Assume that its syntactic morphism
	$\eta_L:\Sigma^+\rightarrow S(L)$ is not in $\QESL$.
	Hence $\stab(\eta_L)$ does not satisfy
	one or both of the identities~(\ref{eq:Xeq1}) and
	(\ref{eq:Xeq2}). 
	In case $\stab(\eta_L)$ does not satisfy the identity (\ref{eq:Xeq2}),
	then it does not satisfy the identity $epxeyqf=epyexqf$
	by Lemma~\ref{L identities}. 
	By Point (1) of Proposition~\ref{P eqconst} we have $\opc(L)\in\Omega(\log n)$.

	Otherwise assume $\stab(\eta_L)$ satisfies the identity (\ref{eq:Xeq2}),
	that is	$epxyqf=epyxqf$, but not the identity (\ref{eq:Xeq1}), that is 
	$epxqf=epxxqf$.
	Due to the implication
(\ref{E xy}) $\Longrightarrow$(\ref{E exeye})
we may assume the identity (\ref{E exeye}).

Thus
	\begin{align}
		epxqf&=\underbracket[.5pt]{e}\,\underbracket[.5pt]{e}\,\underbracket[.5pt]{p}
		\,\underbracket[.5pt]{x}\,\underbracket[.5pt]{q}\,\underbracket[.5pt]{f}&&\\
		&=\underbracket[.5pt]{e}\,\underbracket[.5pt]{e}\,\underbracket[.5pt]{ep}
		\,\underbracket[.5pt]{ex}\,\underbracket[.5pt]{eq}\,\underbracket[.5pt]{f}&&\text{By (\ref{E exeye})}\\
		&=epexeqf
	\end{align}
	and 
	\begin{align}
		epxxqf&=epexexeqf&&\text{By (\ref{E exeye})}
	\end{align}
	Hence $\stab(\eta_L)$ does not satisfy the identity
	$epexeqf=epexexeqf$.
	By Point (2) of Proposition~\ref{P eqconst} 
	we conclude $\opc(L)\in\Omega(\log n)$.
\end{proof}

\subsection{From Algebra to Congruence}
\label{S Congruence const}

This section is devoted to proving the following proposition. 

\begin{proposition}\label{P const tech}
	Let $\varphi:\Sigma^+\to S$ be a stamp in $\QESL$ with stability index $s\in\bbN_{>0}$. 
	Let $N=|\stab(\varphi)|$. 
	For each $u,v \in \Sigma^{+}$ we have that 
	if $u \approx_{2sN} v$, then $\varphi(u)=\varphi(v)$. 
\end{proposition}

Before we discuss the proof strategy for Proposition~\ref{P const tech}
let us show that implication
$(3)\Longrightarrow (4)$ of Theorem~\ref{T const main} 
can be derived from Proposition~\ref{P const tech}.

\begin{corollary}\label{C const main}
If the syntactic morphism $\eta_L$ of a regular language
$L\subseteq\Sigma^+$ is in $\QESL$, then 
	there exists a $k\geq 1$ such that $L$ is a union of 
$\approx_k$-classes.
\end{corollary}
\begin{proof}
Assume the syntactic morphism 
		$\eta_L:\Sigma^+\rightarrow S(L)$ of $L$ is in $\QESL$.
		Let $s\in\bbN_{>0}$ be the stability index of $\eta_L$ 
		and let $N=|\stab(\eta_L)|$.
	By Proposition~\ref{P const tech} 
	for $k=2sN$ for all $u,v\in\Sigma^+$ we have that
	if $u\approx_k v$, then $\eta_L(u)=\eta_L(v)$. 
		Thus $\eta_L^{-1}(x)$ is a union of $\approx_k$-classes 
		for each $x\in S(L)$. 
		Hence $L=\eta_L^{-1}(\eta_L(L))$ itself is a 
		union of equivalence classes of $\approx_k$.
\end{proof}

For proving Proposition~\ref{P const tech} it
comes in handy to give an alternative formulation of how a stamp
maps words whose length is a multiple of its stability index
to elements of its stable semigroup.
Let $\varphi:\Sigma^+\to S$ be a stamp with stability index $s\in\bbN_{>0}$. Let $\sigma_\varphi: (\Sigma^s)^+\to (\stab(\varphi))^+$ be the unique letter-to-letter morphism defined by the map $\sigma_\varphi(w)=\varphi(w)$ for $w \in \Sigma^s$. Let $\pi$ be the evaluation morphism on $\stab(\varphi)$. Notice that for each word $w\in \Sigma^{s\bbN_{>0}}$,
\begin{align}
   \varphi(w)=\pi\circ \sigma_\varphi(\block_s(w))\label{E pi-sigma}.
\end{align}
\begin{samepage}
\begin{definition}
In the context of a morphism $\varphi:\Sigma^+\to S$ we write
$u=_\varphi v$ to mean $\varphi(u)=\varphi(v)$. 
A word $w\in\Sigma^+$ is {\em idempotent} if $\varphi(w)$ is idempotent, i.e., $w=_\varphi ww$.
\end{definition}
\end{samepage}

The heart of our argument for proving Proposition~\ref{P const tech} will be the next lemma.
It states that if $s$ is the stability index of the morphism and if 
$p\in\Sigma^s$ is an idempotent, then any
conjugate of $p^2$ is an idempotent as well; moreover it states
that given any word $u\in\Sigma^{s\bbN_{>0}}$ we have
$\varphi(pup)=\varphi(pu'p)$, where $u'$ is a word obtained
from $u$ by introducing a suitably shifted conjugate 
of $p^2$ after every letter of $u$.

\begin{samepage}
\begin{lemma}\label{L core algebra}
	Let $\varphi : \Sigma^+ \to S$ be a stamp whose
	stable semigroup satisfies the identity $exyf=exeyf$ with stability index $s\in\bbN_{>0}$. 
Let $p=p_0\ldots p_{s-1}\in \Sigma^s$ be an idempotent and 
for each $i\in\bbN$ let
\begin{align}
e_i = \begin{cases} 
p^2 &\text{if $i\equiv_s -1$,}\\
	p_{(i+1) \bmod s}\ldots p_{s-1} p p_0\ldots p_{i \bmod s} &\text{otherwise.}
\end{cases}
\end{align}
	Then $e_i$ is an idempotent as well for each $i\in\N$.
	Moreover, for all $u=u_0\ldots u_{m-1} \in \Sigma^{s\bbN_{>0}}$
	we have $pup=_\varphi pu_0e_0u_1e_1\ldots u_{m-1}e_{m-1}p$.
\label{lemma:idem-insert}
\end{lemma}
\end{samepage}
\begin{proof}
For each $i\in\bbN$ we define the following:
\begin{align}
\ell_i&=p_0\dots p_{i \bmod s}\in\Sigma^+\\
	r_i&=(\ell_i)^{-1}p=p_{(i\bmod s) +1}\dots p_{s-1}\in\Sigma^*
\end{align}
Clearly, $\ell_ir_i=p$, $|\ell_i|\equiv_s i+1$, and
	$|r_i|\equiv_s -(i+1)$.
In particular, $\ell_i=p$ and $r_i=\varepsilon$
	whenever $i\equiv_s s-1$.
		
	We further remark that $e_i=(r_i\ell_i)^2=r_i p\ell_i$,
	so in particular $|e_i|\equiv_s 0$ and $\varphi(e_i)\in\stab(\varphi)$.
	Moreover, $\ell_ie_ir_i= \ell_ir_ip\ell_ir_i=p^3$.
	Bearing in mind that $p$ is an idempotent we observe
\begin{align}
    e_ie_i&=
    r_{i}p\ell_ir_{i}p\ell_i\\
    &=_\varphi r_{i}p^3\ell_i\\
    &=_\varphi r_{i}p\ell_i\\
    &= e_i,
\end{align}
thus $e_i$ is idempotent for each $i\in\N$.

	We claim that for all words $v=v_0\ldots v_{n-1}\in \Sigma^{s\bbN_{>0}}$ 
	and for all $i \in [0,n-1]$ we have
\begin{align} 
\label{eq:idem-insert}
	p^3vp^3&=_\varphi p^3\alpha_{i+1}(v)e_i \omega_{n-(i+1)}(v)p^3.
\end{align}

	Let $v'=\alpha_{i+1}(v)=v_0\ldots v_i$ and $v''=\omega_{n-(i+1)}(v)=v_{i+1}\ldots v_{n-1}$. 
	Therefore $|r_{i}v'|\equiv_s -(i+1)+(i+1)\equiv_s 0$. 
	Likewise $|v''\ell_i|\equiv_s n-(i+1)+(i+1)\equiv_s 0$
	since $s|n$.
Hence, both $\varphi(r_{i}v')$ and  $\varphi(v''\ell_i)$ 
are in $\stab(\varphi)$.
	Using the identity $exyf=exeyf$ 
	to show the equality between (\ref{E Cor exyf1}) and
	(\ref{E Cor exyf2}), we observe that 
\begin{align}
p^3vp^3
&=\ell_i e_i r_{i}v'v'' \ell_i e_i r_{i}\\
	&=\ell_i\underbracket[.5pt]{e_i} 
	\underbracket[.5pt]{r_{i} v'}
	\underbracket[.5pt]{v''\ell_i} 
	\underbracket[.5pt]{e_i} r_{i}\label{E Cor exyf1}\\
	&=_\varphi \ell_i e_i r_{i} v'e_iv''\ell_i e_i r_{i}\label{E Cor exyf2}\\
	&=_\varphi p v'e_i v'' p\\
&=_\varphi p^3v'e_iv''p^3.
\end{align}
Hence the claim is proved. 

Let $u=u_0\ldots u_{m-1}\in \Sigma^{s\bbN_{>0}}$ 
	be as stated in the lemma.
Since $pup=_\varphi p^3up^3$, to finish the proof, it suffices to show 
	that $p^3up^3=_\varphi p^3u_0e_0u_1e_1\ldots u_{m-1}e_{m-1}p^3$. 

Let $u^{(i)}=(\prod_{j=0}^{i}u_je_j)  \omega_{m-(i+1)}(u)$ for 
each $i\in [0,m-1]$. 
We prove by induction on $i$ that 
$p^3up^3=_\varphi p^3 u^{(i)}p^3$. For the induction base, 
$p^3up^3=_\varphi p^3u^{(0)}p^3$ by 
(\ref{eq:idem-insert}). For the inductive step, for all $i\in[0, m-2]$
we have
\begin{align}
	p^3up^3&=_\varphi p^3u^{(i)}p^3&&
	\text{(By induction hypothesis)}\\
	&=_\varphi p^3\big(\prod_{j=0}^{i}u_je_j\big)\omega_{m-(i+1)}(u)p^3\\
	&=_\varphi p^3\big(\prod_{j=0}^{i+1}u_je_j\big) \omega_{m-(i+2)}(u)p^3 &&\text{(Applying (\ref{eq:idem-insert}) to the word $u^{(i)}$)}\\
	&=_\varphi p^3u^{(i+1)}p^3.
\end{align}
 Therefore, $pup=_\varphi p^3up^3=_\varphi p^3u^{(m-1)}p^3=_\varphi pu_0e_0u_1e_1\ldots u_{m-1}e_{m-1}p$, as required.
\end{proof}

The following is a technical lemma used in the proof of Proposition~\ref{P const tech}.

\begin{lemma}
	Let $\varphi : \Sigma^+ \to S$ be a stamp in $\QESL$ with stability index $s\in\bbN_{>0}$. 
	Let $u,v \in \Sigma^{s\bbN_{>0}}$ be  words such that 
	$u \congC_s^1 v$.
Then $pup=_\varphi pvp$
for all idempotents $p\in \Sigma^s$.
\label{lemma:pup=pvp}
\end{lemma}
\begin{proof}
	Let $u,v \in \Sigma^{s\bbN_{>0}}$ be such that 
	$u \congC_s^1 v$. 
	Assume that $u=u_0\ldots u_{m-1}\in \Sigma^{s\bbN_{>0}}$ 
	and $v=v_0\ldots v_{n-1}\in \Sigma^{s\bbN_{>0}}$, where $m,n\in s\bbN_{>0}$. 
Let $p=p_0\ldots p_{s-1}\in \Sigma^s$ be idempotent and for each $i\in\bbN$ let
\begin{align}
e_i = \begin{cases} 
p^2 &\text{if $i\equiv_s -1$,}\\
	p_{(i+1) \bmod s}\ldots p_{s-1} p p_0\ldots p_{i \bmod s} &\text{otherwise,}
\end{cases}
\end{align}
	as stated in Lemma~\ref{lemma:idem-insert}.
We have $pup=_\varphi pu'p$,
	where $u'=u_0e_0u_1e_1\ldots u_{m-1}e_{m-1}$
by \Cref{lemma:idem-insert}.
Recall that $\ell_i= p_0\ldots p_{i\bmod s}$,
	$r_i=(\ell_i)^{-1}p$ and $e_i=r_ip\ell_i$ for each $i\in\bbN$.

	For $i\in \bbN$ we define $\lambda_i(a) \in \Sigma^s$ and $\lambda'_i(a)$ as follows:
    \begin{align}
        \lambda_i(a) &= \begin{cases}
            ar_0 = a p_1 \ldots p_{s-1} &\text{if $i\equiv_s 0$}\\
            \ell_{i-1} a r_i = p_0 \ldots p_{(i-1)\bmod s} a p_{(i \bmod s)+1} \ldots p_{s-1} &\text{otherwise}
            \end{cases}\\
             \lambda'_i(a) &= \begin{cases}
            p \lambda_i(a) = \ell_{s-1}ar_0 &\text{if $i\equiv_s 0$}\\
           \lambda_i(a) = \ell_{i-1} a r_i &\text{otherwise}
        \end{cases}
    \end{align} 
We observe that
\begin{align}
	pu'&= p\prod_{i=0}^{m-1} u_ie_i  \\
&= p\prod_{i=0}^{m-1} u_ir_{i}p\ell_i \\
&= pu_0r_0p\ell_0 \prod_{i=1}^{m-1}  u_{i}r_{i}p\ell_{i} \\
&= pu_0r_0 p \Big(\prod_{i=1}^{m-1}  \ell_{i-1} u_{i}r_{i}p \Big)\ell_{m-1} \\
&= \lambda'_0(u_0) p \Big( \prod_{i=1}^{m-1} \lambda'_i(u_i) p \Big) p\\
&= \Big( \prod_{i=0}^{m-1} \lambda'_i(u_i) p \Big) p.
\end{align}

Again by Lemma~\ref{lemma:idem-insert} we obtain that 
	$pvp=_\varphi pv'p$,
	where $v'=v_0e_0v_1e_1\ldots v_{n-1}e_{n-1}$. Analogously as above
	we have
\begin{align}
 pv'&= p \prod_{i=0}^{n-1} v_ie_i 
=  \Big( \prod_{i=0}^{n-1} \lambda'_i(v_i) p \Big) p.
\end{align}
	Recall that since $u\congC_s^1 v$ we have $\cont_{i,s}(u)=\cont_{i,s}(v)$
	for all $i\in[0,s-1]$.
	We note that 
\begin{align}
    \cont(\block_s(pu'p)) &=\{\lambda_i(u_i)\mid i \in [0,m-1]\}\cup\{p\}
	& &\text{(Since $\lambda_i(u_i)\in\Sigma^s$)}\\
    &=\{\lambda_i(v_i)\mid i \in [0,n-1]\}\cup\{p\}&& (\text{Since $\cont_{i,s}(u) = \cont_{i,s}(v)$})\\
    &= \cont(\block_s(pv'p)) \label{E cont-pvp}.
\end{align}

Finally we have
	\begin{align}
		\varphi(pup)&=\varphi(pu'p)\\
        &=\pi\circ \sigma_\varphi(\block_s(pu'p)) &&(\text{By (\ref{E pi-sigma})})\\
        &=\pi\circ \sigma_\varphi(\block_s(pv'p)) &&(\text{By (\ref{E cont-pvp}) \& Lemma~\ref{L cont EJ1}})\\
        &=\varphi(pv'p)\\
        &=\varphi(pvp).
	\end{align}
\end{proof}

\begin{samepage}
\begin{lemma}[Proposition~6.34 in \cite{Pin2025Mathematical}]
\label{lem:prefidem}
Let $S$ be a finite semigroup and let $n = |S|$. 
	For every sequence $s_0, \dots, s_{n-1}$ of $n$ elements of $S$, 
	there exists
	an index $i \in [0,n-1]$ and an idempotent $e \in E(S)$ such that 
	$s_0 \cdots s_i \, e \;=\; s_0 \cdots s_i$
	and dually, there exists an index $j \in[0,n-1]$ 
	and an idempotent $f \in E(S)$ such that
	$f \, s_j s_{j+1} \cdots s_{n-1} \;=\; s_j s_{j+1} \cdots s_{n-1}$.
\end{lemma}
\end{samepage}

The following lemma is a reformulation of the previous one tailored towards stamps.

\begin{samepage}
	\begin{lemma}\label{L pq idem}
    Let $\varphi : \Sigma^+ \to S$ be a stamp with stability index 
	$s\in\bbN_{>0}$ and let $N=|\stab(\varphi)|$ be the size of its stable semigroup.
	Then for all words $u\in\Sigma^{s\N_{>0}}$ with $|u|\geq 2sN$ there exist idempotents
	$p,q\in\Sigma^s$ and a factorization $u=xu'y$
	such that
	$x=_\varphi xp$, $qy=_\varphi y$, $|x|=(i+1)\cdot s$, 
	and $|y|=(j+1)\cdot s$ for some $i,j\in[0,N-1]$.
\end{lemma}
\end{samepage}
\begin{proof}
	Let $u\in\Sigma^{s\N_{>0}}$ be such that $|u|\geq 2sN$.
	Write $u=\overline{u_0}\dots \overline{u_{N-1}}z
\overline{v_0}\dots \overline{v_{N-1}}$, where
	$\overline{u_0},\dots,\overline{u_{N-1}},
	\overline{v_0},\dots,\overline{v_{N-1}}\in\Sigma^s$
	and where $z\in\Sigma^{s\N}$.
	The lemma follows by applying Lemma~\ref{lem:prefidem}
	to the sequences 
	$\sigma_\varphi(\overline{u_0}),\dots,
	\sigma_\varphi(\overline{u_{N-1}})$
	and 
	$\sigma_\varphi(\overline{v_0}),\dots,
	\sigma_\varphi(\overline{v_{N-1}})$, respectively.
\end{proof}

The following lemma treats a special case 
of Proposition~\ref{P const tech},
namely when restricted to $\approx_k$-equivalent
words both of whose length is a multiple of $k$.
\begin{lemma}
	Let $\varphi:\Sigma^+\to S$ be a stamp in $\QESL$ with stability index $s\in\bbN_{>0}$. 
	Let $N=|\stab(\varphi)|$ and $k=sN$. 
	For each $u,v \in \Sigma^{+}$ we have that 
	if $u \approx_{k} v$ and $|u|\equiv_k |v| \equiv_k 0 $, 
	then $u=_\varphi v$. 
\label{lemma:alg-to-cong}
\end{lemma}
\begin{proof}
	Let $u,v\in \Sigma^+$ be such that
	$u \approx_k v$ and $|u|\equiv_k |v|\equiv_k 0 $.
	If $\min(|u|,|v|)\leq 2k$ then $u=v$ and therefore 
	$u=_\varphi v$ and the lemma is proved. 
	Hence, let us assume that $|u|,|v|> 2k$.

	Since $\alpha_k(u)=\alpha_k(v)$ and $\omega_k(u)=\omega_k(v)$
	by Lemma~\ref{L pq idem} there exist idempotents $p,q\in \Sigma^s$ and 
	factorizations $u=xu'y$, $v=xv'y$ 
	such that $|x| = (i+1)\cdot s$, $|y| = (j+1)\cdot s$ 
	for some $i,j\in[0,N-1]$,
$u=_\varphi xpu'qy$ and $v=_\varphi xpv'qy$. 
	Since $|p|,|u'|, |v'|, |q|$ are all multiples of $s$, 
we have
	$\varphi(p),\varphi(u'),\varphi(v'),\varphi(q) \in \stab(\varphi)$. 
	Hence $pu'q=_\varphi pu'pq$, and $pv'q=_\varphi pv'pq$
	by Point (1) of \Cref{E exyf}. 
	Consequently, $u=_\varphi xpu'qy=_\varphi xpu'pqy$ 
	and $v=_\varphi xpv'qy=_\varphi xpv'pqy$.
Thus, to prove the lemma, it suffices to show that 
	$pu'p=_\varphi pv'p$. 

Note that $u'=x^{-1}uy^{-1}=x'\kappa_k(u)y'$ and $v'=x^{-1}vy^{-1}=x'\kappa_k(v)y'$, where 
	$x'=x^{-1}\alpha_{k}(u)=x^{-1}\alpha_{k}(v)$ and
$y'=\omega_{k}(u)y^{-1}=\omega_{k}(v)y^{-1}$.
Since $u\approx_k v$ implies $\kappa_k(u)\congC_k^1 \kappa_k(v)$ and since 
	$\congC_k^1$ is a congruence 
	by Point (1) of \Cref{L cont} we have $u'=x'\kappa_k(u)y' \congC_k^1 x'\kappa_k(v)y' = v'$. 
	Since $s|k$ we also have $u'\congC_s^1 v'$ by Point (2) of Lemma~\ref{L cont}.
Recalling that the lengths of $u'$ and $v'$ are both divisible by
$s$ we therefore have
	$pu'p=_\varphi pv'p$ by \Cref{lemma:pup=pvp}.
\end{proof}

We are now ready to prove Proposition~\ref{P const tech}.

\subsubsection{Proof of Proposition~\ref{P const tech}}
\begin{proof}
Let $k=sN$.
    By \Cref{lemma:alg-to-cong} 
	for all $u,v\in\Sigma^+$ we have that
	if $u\approx_{k} v$ 
	and $|u|\equiv_{k} |v| \equiv_{k} 0$, then $u=_\varphi v$. 
	We claim that for all $u,v \in \Sigma^+$, if $u \approx_{2k} v$,
	then $u=_\varphi v$.

Assume that $u,v\in \Sigma^+$ such that $u\approx_{2k} v$. If $|u|$ or $|v|$ is $\leq 4k$,
	then $u=v$, so $u=_\varphi v$ and we are done. 
Next assume that $|u|,|v|>4k$. Since $u\approx_{2k}v$, 
	clearly $|u|\equiv_{2k} |v|$. 
	Let $u',v'\in \Sigma^+$ and $u'',v''\in \Sigma^*$ be such that
	$u=u'u''$, $v=v'v''$, $|u'|\equiv_{k} |v'|\equiv_{k} 0 $ 
	and $|u''|=|v''|<k$. Since $u$ and $v$ agree on their $2k$-length suffix, clearly $u''=v''$. 
	Hence to show that $u=_\varphi v$ it suffices to 
	show that $u'=_\varphi v'$. 
	To show $u'=_\varphi v'$ in turn, it is sufficient to prove 
	$u'\approx_{k} v'$ by \Cref{lemma:alg-to-cong}.

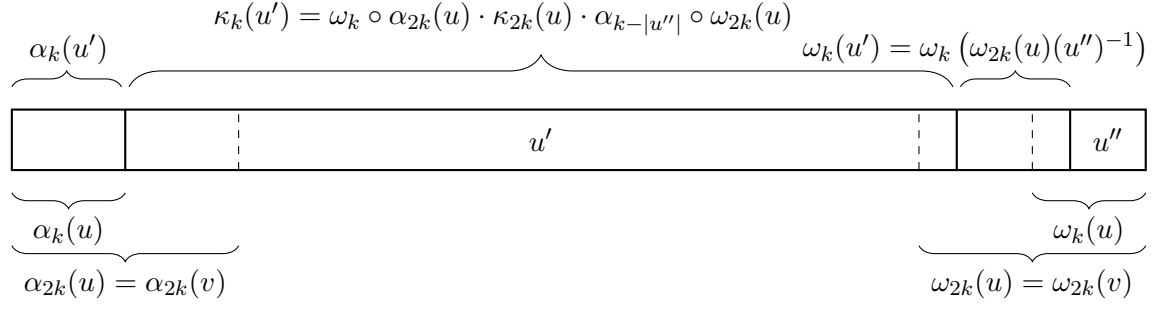
\begin{figure}
    \centering
    \begin{tikzpicture}
    \draw[thick] (0,0) rectangle (15,0.8);

    \draw[thick] (14,0) -- (14,0.8);
    \draw[thick] (12.5,0) -- (12.5,0.8);
   
    \draw[dashed] (12,0) -- (12,0.8);
    \draw[dashed] (3,0) -- (3,0.8);
    
  \draw[dashed] (13.5,0) -- (13.5,0.8);
    \draw[thick] (1.5,0) -- (1.5,0.8);
  
    \node at (7,0.4) {$u'$};
    \node at (14.5,0.4) {$u''$};

\draw[decorate,decoration={brace,amplitude=6pt,mirror}]
        (0,-0.3) -- (1.5,-0.3);
    \node at (.75,-0.8) {$\alpha_{k}(u)$};

\draw[decorate,decoration={brace,amplitude=6pt}]
        (0,1.1) -- (1.45,1.1);
    \node at (.75,1.6) {$\alpha_{k}(u')$};
 
    \draw[decorate,decoration={brace,amplitude=6pt,mirror}]
        (0,-1) -- (3,-1);
    \node at (1.5,-1.5) {$\alpha_{2k}(u)=\alpha_{2k}(v)$};

    \draw[decorate,decoration={brace,amplitude=6pt,mirror}]
        (13.5,-0.3) -- (15,-0.3);
    \node at (14.25,-0.8) {$\omega_{k}(u)$};
    \draw[decorate,decoration={brace,amplitude=6pt,mirror}]
        (12,-1) -- (15,-1);
    \node at (13.5,-1.5) {$\omega_{2k}(u)=\omega_{2k}(v)$};

    \draw[decorate,decoration={brace,amplitude=6pt}]
        (12.55,1.1) -- (14,1.1);
    \node at (12.75,1.6) {$\omega_{k}(u')=\omega_k\left(\omega_{2k}(u)(u'')^{-1}\right)$};

    \draw[decorate,decoration={brace,amplitude=12pt}]
        (1.55,1.1) -- (12.45,1.1);
    \node at (6.5,2) {$\kappa_{k}(u')=\omega_k\circ \alpha_{2k}(u)\cdot \kappa_{2k}(u) \cdot \alpha_{k-|u''|}\circ \omega_{2k}(u)$};

\end{tikzpicture}
    \caption{The word $u=u'u''$}
    \label{fig:u'u''}
\end{figure}
 
	So let us prove $u'\approx_{k} v'$. 
	Since $u\approx_{2k} v$, we have $\alpha_{2k}(u)=\alpha_{2k}(v)$ and $\omega_{2k}(u)=\omega_{2k}(v)$. 
	Therefore, $\alpha_{k}(u')=\alpha_{k}(u)=\alpha_{k}(v)=
	\alpha_{k}(v')$.  Also,
	$\omega_{k}(u')=\omega_{k}(\omega_{2k}(u)(u'')^{-1})=\omega_{k}(\omega_{2k}(v)(v'')^{-1})
	=\omega_{k}(v')$ (see \Cref{fig:u'u''}).
	It remains to show that $\kappa_{k}(u')\congC_{k}^1 \kappa_{k}(v')$.
We observe that 
\begin{align}
	\kappa_{k}(u')&=(\omega_k\circ \alpha_{2k}(u))\cdot \kappa_{2k}(u) \cdot 
	(\alpha_{k-|u''|}\circ \omega_{2k}(u))\label{E ao1}\\
	\kappa_{k}(v')&=(\omega_k\circ \alpha_{2k}(v))\cdot \kappa_{2k}(v) \cdot 
	(\alpha_{k-|v''|}\circ \omega_{2k}(v))\label{E ao2}\\
	&=(\omega_k\circ \alpha_{2k}(u))\cdot \kappa_{2k}(v) \cdot 
	(\alpha_{k-|u''|}\circ \omega_{2k}(u))\label{E ao3}
\end{align}
    
	Since $u\approx_{2k}v$ we have $\kappa_{2k}(u)\congC_{2k}^1 \kappa_{2k}(v)$. 
	By Point (2) of \Cref{lemma:prop-cont}
	we have $\kappa_{2k}(u)\congC_{k}^1 \kappa_{2k}(v)$. 
	Since $\congC_{k}^1$ is a congruence
	by Point (1) of \Cref{lemma:prop-cont} we have
	$\kappa_{k}(u')\congC_{k}^1 \kappa_{k}(v')$ by (\ref{E ao1}) and (\ref{E ao3}), 
	as required.
\end{proof}

\end{samepage}

\section{Decidability of Membership in $O(1)$ Circuit Complexity}\label{S PSPACE}

Recall that a {\em nondeterministic finite automaton (NFA)}
is a tuple $\calA=(Q,\Sigma,\Delta,q_0,F)$, where $Q$ is a finite set of
{\em states}, $\Sigma$ is a finite {\em alphabet}, $\Delta\subseteq Q\times\Sigma\times Q$
is a {\em transition relation}, $q_0\in Q$ is the {\em initial state} and $F\subseteq Q$
is a set of {\em final states}.
By $|\calA|=|Q|+|\Sigma|+|\Delta|$ we denote the {\em size} of $\calA$ and by
$L(\calA)$
we denote the {\em language} of $\calA$.

The {\em $O(1)$-membership problem} takes as input a regular 
language $L$ and asks whether $\opc(L)\in O(1)$.

\begin{samepage}
	\begin{theorem}\label{T constant PSPACE}
		The $O(1)$-membership problem for NFAs
		is $\mathbf{PSPACE}$-complete.
\end{theorem}
\end{samepage}

The theorem follows from 
Proposition~\ref{P constant PSPACE upper}
and 
Proposition~\ref{P constant PSPACE lower}.

\subsection{Upper Bound}

\begin{samepage}
\begin{proposition}\label{P constant PSPACE upper}
		The $O(1)$-membership problem for NFAs 
		is in $\mathbf{PSPACE}$.
\end{proposition}
\end{samepage}
\begin{proof}
	For the $\mathbf{PSPACE}$ upper bound let $L=L(\calA)$ for some given 
	NFA $\calA=(Q,\Sigma,\Delta,q_0,F)$.
	Let us show that in space polynomial in $|\calA|$ one can decide if 
	$\opc(L)=O(1)$. 

	Let $\mathbb{B}^{Q\times Q}$ denote the multiplicative 
	semigroup of Boolean matrices of dimension $Q\times Q$.
	Let $\mu:\Sigma^+\to\mathbb{B}^{Q\times Q}$
	be the unique morphism satisfying $\mu(a)=t_a$,
	where $t_a(q,q')=1\Leftrightarrow (q,a,q')\in\Delta$ for each $a\in\Sigma$.
	Let $T$ denote the subsemigroup $\mu(\Sigma^+)$ of 
	$\mathbb{B}^{Q\times Q}$.
	
	Clearly, each element of $\mathbb{B}^{Q\times Q}$ can be stored in space polynomial in $|\calA|$
	and so can (matrix) multiplication in $\mathbb{B}^{Q\times Q}$ 
	be performed in polynomial space in $|\calA|$.

	Let $\beta:T\to S(L)$ denote the unique surjective morphism that 
	factors $\eta_L$ as $\eta_L=\beta\circ\mu$.
	In the following let $B=|T|$, let $s$ denote the stability index
	of $\eta_L$ and let $s_\mu$ denote the stability index of $\mu$.
	We observe that 
	\begin{align}
		\beta(\stab(\mu))=\beta(\mu(\Sigma^{s_\mu}))=
		\beta(\mu(\Sigma^{s\cdot s_\mu}))=
		\eta_L(\Sigma^{s\cdot s_\mu})=
		\eta_L(\Sigma^s)=\stab(\eta_L).\label{E stab eta}
	\end{align}

\medskip

\noindent
	{\em Claim 1.}
	Let $t\in\mathbb{B}^{Q\times Q}$. Then $t\in\stab(\mu)$ if, and only if, 
	there exist $d\in[1,B]$ and words $u,v\in\Sigma^+$ and $w\in\Sigma^*$ satisfying 
	the following:
	\begin{enumerate}
		\item $\mu(uvw)=t$,
		\item $\mu(u)=\mu(uv)$,
		\item $|v|=d$,
		\item $|uw|\leq 2dB$, and
		\item $|uw|\equiv_d 0$.
	\end{enumerate}
	\begin{proof}
		Let $s_\mu$ denote the stability index of $\mu$.

		\medskip
		\noindent
		$(\Longleftarrow)$ \quad Assume such $d\in[1,B]$
		and words $u,v\in\Sigma^+$ and $w\in\Sigma^*$ satisfying Points {(1)~--~(5)}.
		By Point (5) we have $|uw|\in d\N$, so there exists $r\in\bbN$
		such that $|uw|=rd$.
		Let $i=s_\mu(r+1)-r\in\bbN_{>0}$.
		Hence 
		\begin{align}
			|uv^iw|=i\cdot|v|+|uw|=(s_\mu(r+1)-r)d+rd=s_\mu(r+1)d\label{E uviw}.
		\end{align}
		Thus, we have
		\begin{align}
			t&=\mu(uvw)&\text{(By Point (1))}\\
			&=\mu(uv^{i}w)&\text{(By Point (2))}\\
			&\in\mu(\Sigma^{s_\mu(r+1)d})&\text{(By (\ref{E uviw}))}\\
			&=\mu(\Sigma^{s_\mu})&(\text{Since $s_\mu$
			is the stability index of $\mu$})\\
			&=\stab(\mu)
		\end{align}
		\noindent
		$(\Longrightarrow)$ \quad Assume $t\in\stab(\mu)$.
		Thus, for all $i\in s_\mu\N_{>0}$ there exists some word
		$x\in\Sigma^i$ such that $\mu(x)=t$.
		Let $i=(B+1)!\cdot s_\mu\in s_\mu\bbN_{>0}$ and let $x\in\Sigma^{i}$ be some word
		such that $\mu(x)=t$.
		By the pigeonhole principle there exists a factorization
		$x=uvw$ such that $\mu(u)=\mu(uv)$ 
		and where $|v|=d$ for some $d\in[1,B]$.
		Thus Points (1),(2) and (3) are satisfied.
		Since $|x|$ is a multiple of $d$ and $|v|=d$
		we have $|uw|\equiv_d 0$, hence Point (5) is
		also satisfied.
		It remains to prove Point (4).

		We claim that every word $z\in\Sigma^+$ of length $>dB$
		can be factorized as $z=z_1z_2z_3$ such that
		$\mu(z_1)=\mu(z_1z_2)$ and
		$z_2\in\Sigma^{d\bbN_{>0}}$.
		To prove the claim consider 
		the $B+1$ prefixes
		$\alpha_1(z),\alpha_{d+1}(z),\dots, \alpha_{dB+1}(z)$
		of $z$.
		The claim follows from the fact that, 
		by the pigeonhole principle,
		at least two of the above prefixes of $z$ have
		the same image under $\mu$.

		By the above claim if $u$ or $w$ has length
		$>dB$ they can be replaced by a
		smaller word with the same image under $\mu$
		and the same length modulo $d$.
		By possibly repeatedly applying the above shortening
		of words we may assume without loss of generality
		that $|u|,|w|\leq dB$, so $|uw|\leq 2dB$. 
		Thus Point (4) follows.
	\end{proof}

\medskip

\noindent
	{\em Claim 2.}
	Given $t\in\mathbb{B}^{Q\times Q}$ one can decide in 
	space polynomial in $|\calA|$
	if $t\in\stab(\mu)$.
	\begin{proof}
		Let $t\in \mathbb{B}^{Q\times Q}$.
		By Claim 1 it suffices to guess in polynomial space 
		$d\in[1,B]$ and three matrices 
		$t_u,t_v,t_w\in\mathbb{B}^{Q\times Q}$,
		and numbers $a,b\in[0,2dB]$ in binary and verify in polynomial space
		the following conditions. Below we assume
		$b\not=0$; the case when $b=0$ is analogous and hence omitted.
		\begin{samepage}
		\begin{enumerate}
			\item $a>0$, $a+b\equiv_d 0$, and $a+b\leq 2dB$,
			\item $t_u\in\mu(\Sigma^a)$
			\item $t_w\in\mu(\Sigma^b)$ 
			\item $t_v\in\mu(\Sigma^d)$ 
			\item $t_u=t_ut_v$
			\item $t=t_ut_vt_w$
		\end{enumerate}
		\end{samepage}
		Notably Points (2), (3), and (4) are verified
		on the fly since we cannot store exponentially
		long words.
	\end{proof}

\medskip

\noindent
	{\em Claim 3.}
	Given $t,t'\in\stab(\mu)$ one can decide in space polynomial in 
	$|\calA|$ if $\beta(t)=\beta(t')$.
	\begin{proof}
	For all $t\in\mathbb{B}^{Q\times Q}$ and $q\in Q$ we introduce
	the notation $t(q,F)=\bigvee_{f\in F}t(q,f)$.
		For a semigroup $S$ let $S^1$ denote the monoid obtained
		by adjoining an identity if necessary. 
		By definition of syntactic morphism, observe that
		\begin{align} \beta(t)=\beta(t')&\quad\Longleftrightarrow\quad
			\forall p,q\in S(L)^1: p\beta(t)q\in\eta_L(L)\Leftrightarrow
			p\beta(t')q\in\eta_L(L)\\
			&\quad\Longleftrightarrow\quad
			\forall t_p,t_q\in T^1: (t_ptt_q)(q_0,F)=(t_pt't_q)(q_0,F)\label{E pqcheck}
		\end{align}
		To verify (\ref{E pqcheck}) we enumerate in polynomial
		space all pairs $(t_p,t_q)\in(\mathbb{B}^{Q\times Q})^2$ and 
		first verify if $t_p,t_q\in T^1$. The latter in turn can be tested
		by verifying if $t_p$ and/or $t_q$ are the identity respectively,
		or by guessing on the fly (letter by letter) words $w_p$ and/or
		$w_q$ of length at most $B$ and checking if $\mu(w_p)=t_p$
		and $\mu(w_q)=t_q$ indeed holds.
		In case the former test fails, we continue with the next pair.
		In case the former test succeeds, we verify if
		$(t_ptt_q)(q_0,F)=(t_pt't_q)(q_0,F)$ in polynomial space.
	\end{proof}
	By Point (3) of Theorem~\ref{T const main} it suffices to verify if $\eta_L$ is
	in $\QESL$, i.e. if $\stab(\eta_L)$ satisfies both the identities 
	$epxqf= epxxqf$ and $epxyqf= epyxqf$ of Definition~\ref{D ESL},
	where $e$ and $f$ denote idempotents.
	We only discuss the verification of the former identity.

	In polynomial space we deterministically enumerate all possible assignments
	of representatives $\rho=(t_e,t_f,t_p,t_x,t_q)\in
	(\mathbb{B}^{Q\times Q})^5$
	for $e,f,p,x,q$ respectively. 
	Let $\rho=(t_e,t_f,t_p,t_x,t_q)$. By Claim 2 one can verify in
	polynomial space if $t$ is indeed in $\stab(\mu)$ 
	for all components $t$ of $\rho$. We skip the current assignment if 
	at least one component of $\rho$ is not in $\stab(\mu)$;
	in the latter case we continue with the next assignment.
	
	So let us assume $t\in\stab(\mu)$ for all components $t$ of $\rho$.
	For all of these we have $\beta(t)\in\stab(\eta_L)$ by
	(\ref{E stab eta}).

	Next we verify in polynomial space if $\beta(t_e)$ and $\beta(t_f)$
	are indeed idempotents. 
	Note that it is not necessary for $t_e$ to be idempotent
	for $\beta(t_e)$ to be idempotent.
	However, note that $t_e^\omega$, denoting the unique
	idempotent power of $t_e$, satisfies $\beta(t_e^\omega)=\beta(t_e)$
	if $\beta(t_e)$ is idempotent.
	Hence, it suffices to verify in polynomial
	space if $t_e$ and $t_f$ are actually idempotents in $T$
	themselves, which is easy.
	Again, if $t_e$ or $t_f$ is not idempotent, we skip the current assignment
	$\rho$.
	Next, we compute $t=t_et_pt_xt_qt_f$ and $t'=t_et_pt_xt_xt_qt_f$
	in polynomial space.
	Finally we verify if $\beta(t)=\beta(t')$ which is possible in 
	polynomial space by Claim 3. 
	We only accept if all assignments that are not skipped indeed
	satisfy the identity.
\end{proof}

\subsection{Lower Bound}
\begin{samepage}
\begin{proposition}\label{P constant PSPACE lower}
		The $O(1)$-membership problem for NFAs 
		is $\mathbf{PSPACE}$-hard under
		logarithmic space reductions.
\end{proposition}
\end{samepage}
\begin{proof}
	It is folklore that universality of NFAs is $\textbf{PSPACE}$-hard
	under logarithmic space reductions: we refer
	to Theorem~3.13 in~\cite{EsparzaBlondin2023}, 
	where for a fixed linear bounded automaton $\calB$ and some input $x$
	one computes in logarithmic space an NFA $\calA$ over
	some alphabet $\Sigma$ such that there is a word $w\in\Sigma^+$ 
	for which:
	\begin{align}
		L(\calA)=\begin{cases}
\Sigma^+\setminus\{w\} &\text{if $\calB$ accepts $x$}\\
			\Sigma^+&\text{otherwise}
			\label{E univ}
		\end{cases}
	\end{align}
	Let $\Gamma=\Sigma\cup\{\$\}$ for some fresh symbol $\$\not\in\Sigma$
	and let $\varphi:\Gamma^*\to\Sigma^*$ denote the monoid morphism
	that erases all occurrences of $\$$, i.e.
	$\varphi(a)=a$ for all $a\in\Sigma$ and $\varphi(\$)=\varepsilon$.
	We claim that one can compute from $\calA$ in logarithmic space
	an NFA $\calA'$ computing
	\begin{align}
		L(\calA')=\varphi^{-1}(L(\calA))
	\end{align}
	Such an NFA $\calA'$ can indeed easily be computed 
	from $\calA$ in logarithmic space
	by adding a $\$$-loop on every state.
	To finish the proof it suffices to show 
	that $\opc(L(\calA'))\in O(1)$ if, and only if, $L(\calA)=\Sigma^+$.

	\medskip

	\noindent
	$(\Longleftarrow)$\quad
	Assume $L(\calA)=\Sigma^+$. 
	Then $L(\calA')=\varphi^{-1}(L(\calA))=\Gamma^+\setminus \$^+$. 
	Clearly, $\opc(L(\calA'))\in O(1)$.

\medskip

	\noindent
	$(\Longrightarrow)$\quad
	Assume $L(\calA)\not=\Sigma^+$.
	By (\ref{E univ}) we have $L(\calA)=\Sigma^+\setminus\{w\}$
	for some word $w\in\Sigma^+$.
	Note that by definition 
	$L(\calA')=\varphi^{-1}(\Sigma^+\setminus\{w\})=\varphi^{-1}(\Sigma^+)
	\setminus\varphi^{-1}(w)$ has
	the neutral letter $\$$. 
	Note that for an idempotent and commutative language $K$
	a word $x\in K$ if, and only if, $xx\in K$.
	Since $\varphi^{-1}(ww)\subseteq L(\calA')$ and
	$L(\calA')\cap\varphi^{-1}(w)=\emptyset$ we conclude that
	$L(\calA')$ cannot be both idempotent and commutative.
	By Point (2) of Corollary~\ref{C Neutral Letter} we have 
	$\opc(L(\calA'))\not\in O(1)$, as required.
\end{proof}

\section{Discussion}\label{S Discussion}

\subsection{The Pseudovariety $\bfQ\EACom$ and $\approx_k^t$-languages}

The pseudovariety $\ACom$ of aperiodic and commutative semigroups is the 
generalization of $\bfJone$. Therefore it is natural to consider the lm-variety $\QEACom$ and 
its circuit complexity. 
Moreover the class of languages corresponding to $\QEACom$ contains 
the well-known threshold languages (of the form $\{ w \in \Sigma^+ \mid |w|_a > t\}$ for some $t\in \bbN$) from the literature~\cite{Friedman86,HASTAD1994200}.

\begin{definition}
	Let $\EACom$  denote the semigroup pseudovariety given by the identities 
\begin{align}
epx^{\omega+1} qf&=epx^{\omega}qf \label{eq:EAComeq1}\\
epxyqf&=epyxqf\label{eq:EAComeq2}
\end{align}
	for all $e,f\in E(S)$ and $p,q,x,y\in S$.
By $\QEACom$ we denote the lm-pseudovariety of
	stamps whose stable semigroup 
	is in $\EACom$.
\end{definition}

	\begin{lemmarep}\label{L cont EACom}
		Let $S$ be a semigroup in $\EACom$ with exponent $n$. Let $e,f\in E(S)$ 
		be idempotents. 
		Then for all words $u,v\in S^+$,
		if $\cont^{=t}(u)\setminus \{e,f\}=\cont^{= t}(v)\setminus \{e,f\}$ 
		for all $t \in [0,n-1]$ and 
        $\cont^{\geq n}(u)\setminus \{e,f\}=\cont^{\geq n}(v)\setminus \{e,f\}$, then $\pi(euf)=\pi(evf)$,
		where $\pi$ is the evaluation morphism.
\end{lemmarep}

\begin{appendixproof}
Let $u,v\in S^+$ be such that $\cont^{=t}(u)\setminus \{e,f\}=\cont^{= t}(v)\setminus \{e,f\}$ for all $t \in [0,n-1]$ and $\cont^{\geq n}(u)\setminus \{e,f\}=\cont^{\geq n}(v)\setminus \{e,f\}$.
Hence $\cont(u)\setminus \{e,f\}=\cont(v)\setminus \{e,f\}$ as well.

 Using  (\ref{eq:EAComeq2}), as in the proof of \Cref{L cont EJ1}, we obtain that $\pi(euf)=\pi(eu'f)$ for any permutation $u'$ of $u$. 
Assume that $e\not= f$ (when $e = f$ the analysis is similar and hence omitted).
Also assume that $\cont(u)\setminus\{e,f\}=\cont(v)\setminus\{e,f\}=
	\{s_1,\ldots, s_k\}\subseteq S$ for some $k\in\bbN$.
	Since $e^{|u|_e}s_1^{|u|_{s_1}}\ldots s_k^{|u|_{s_k}}f^{|u|_f}$
	is a permutation of $u$,
	we deduce
\begin{align}
    \pi(euf)&=\pi(ee^{|u|_e}s_1^{|u|_{s_1}}\ldots s_k^{|u|_{s_k}}f^{|u|_f}f)\\
    &=\pi(es_1^{i_1}\ldots s_k^{i_k}f) &&(\text{By (\ref{eq:EAComeq1})})
\end{align}
	where $i_j=\min(n,|u|_{s_j})$, thus 
	$i_j \equiv^n |u|_{s_j}$ for all $j\in [1,k]$.
By the assumption $|u|_{s_i}\equiv^n |v|_{s_i}$ for all $i\in [1,k]$. 
Hence we have $\pi(euf)=\pi(es_1^{i_1}\ldots s_k^{i_k}f)=\pi(evf)$.
\end{appendixproof}

\begin{corollary}
    The following are identities of $\EACom$, where $e$ and $f$ range over idempotents.
    \begin{enumerate}
\item $exf=exef=efxf$
\item$exyf=exeyf=exfyf$
\end{enumerate}
\label{E Acom-exyf}
\end{corollary}

\ificalp
\else

\begin{samepage}
\begin{definition}\label{D Acom Count}
	Let $k,t\in \bbN_{>0}$. 
	For each vector $\nu = (\nu_{i,a})_{i \in [0,k-1], a \in \Sigma}$, 
	where $\nu_{i,a} \in [0,t]$ we define the language
	\begin{align}
    \mathrm{Count}_k^t(\nu)&=\{ w \in \Sigma^{k\N_{>0}} : 
		\left|\pi_{i,k}(w)\right|_a \equiv^t \nu_{i,a}\text{ for all
		$i\in[0,k-1]$ and all $a\in\Sigma$}\}.
	\end{align}
\end{definition}
\end{samepage}

\begin{definition}\label{D Acom Thr}
	For each $k,t\in \bbN_{>0}$, $i\in [0,k-1]$ and $a \in \Sigma$ 
	the language $\mathrm{Thr}_{i,k}(a,t)$ is defined to be
    \begin{align}
        \mathrm{Thr}_{i,k}(a, t) &= 
	    \{ w \in \Sigma^{k\bbN_{>0}} : \left|\pi_{i,k}(w)\right|_a \geq t \}.
    \end{align}
\end{definition}
\fi

\begin{remark}
For 
$\nu=(\nu_{i,a})_{i\in[0,k-1],\,a\in\Sigma}
\in \{0,\ldots,t\}^{[0,k-1]\times\Sigma}$ 
we have
	\begin{align}
\mathrm{Count}_k^t(\nu)
=
\bigcap_{\substack{i\in[0,k-1]\\ a\in\Sigma}}
\begin{cases} \Sigma^{k\bbN_{>0}}\cap 
\overline{\mathrm{Thr}_{i,k}(a,1)},
	& \text{if }\nu_{i,a}=0,\\[1mm]
\mathrm{Thr}_{i,k}(a,\nu_{i,a})
\cap
\overline{\mathrm{Thr}_{i,k}(a,\nu_{i,a}+1)},
    & \text{if }0<\nu_{i,a}<t,\\[1mm]
\mathrm{Thr}_{i,k}(a,t),
    & \text{otherwise, i.e. }\nu_{i,a}=t.
\end{cases}
	\end{align}
\label{R Acom Count}
\end{remark}

\begin{samepage}
	\begin{theorem}\label{T QEAComm}
	Let $L\subseteq\Sigma^+$ be a regular language. 
		Then the following statements are equivalent:
	\begin{enumerate}
		\item $\eta_L\in\QEACom$.
		\item $L$ is a union of $\approx_k^t$-classes
			for some $k,t\in\bbN_{>0}$.

        \item $L$ is a disjoint finite union of finite languages and languages of the 
		form $u\mathrm{Count}_k^t(\nu)v$, where $k,t\in\bbN_{>0}$, 
			$\nu \in [0,t]^{[0,k-1]\times \Sigma}$ and $u ,v\in \Sigma^+$.
		\item $L$ is a finite Boolean combination of finite languages and languages
			of the form
			$u\mathrm{Thr}_{i,k}(a,t)v$,
			where $k,t\in\bbN_{>0}$, $i \in [0,k-1]$,
            $a \in\Sigma$, and $u,v\in\Sigma^+$.
        
		\item $L$ is definable in 
			$\rmC^1[\Sigma,  +\omega, \min, \max, \rmmod]$.
\item $L$ is definable in 
			$\rmC^1[\Sigma, \mathrm{reg}]$.
  \item $L$ is definable in 
			$\rmC^1[\Sigma, \mathrm{arb}]$.
    \end{enumerate}
\end{theorem}
\end{samepage}
\begin{proof}
We show 
$(1) \Longrightarrow (2) \Longrightarrow (1)$, $(2) \Longrightarrow (3) \Longrightarrow
(4) \Longrightarrow (5) \Longrightarrow (2)$ and $(5) \Longrightarrow (6)
\Longrightarrow (7)\Longrightarrow (5)$.
\medskip

\begin{samepage}
\noindent $(1) \Longrightarrow  (2)$\quad 
By Proposition~\ref{P log tech} proved in 
		\S\,\ref{S Congruence Acomm}.
	\end{samepage}
\medskip

	\begin{samepage}
\noindent 
	$(2) \Longrightarrow  (1)$\quad It suffices to show that the stamp $\varphi:\Sigma^+ \to \Sigma^+/\approx_k^t$ is in $\QEACom$. 
    Let $s$ be the stability index of $\varphi$ and let $N=|\stab(\varphi)|$. Let $s'=2k+kt|\Sigma|^k$. 
   If $w\in \Sigma^{s'}$ then $\block_k(\kappa_k(w))$ contains $t|\Sigma|^k$ blocks of words of length $k$. By pigeonhole principle there is a block $x\in \Sigma^k$ that appears at least $t$ times in $\block_k(\kappa_k(w))$. From this observation it is easy to derive that $w \approx_k^t w'$ for some $w' \in \Sigma^{2s'}$ and vice versa. Therefore $\varphi(\Sigma^{s'}) = \varphi(\Sigma^{2s'})$. Moreover,
    \begin{align}
\varphi(\Sigma^{s'})=\varphi(\Sigma^{s's})=\varphi(\Sigma^{s}) = \stab(\varphi).     
    \end{align}
    
	Take $e,f\in E(\stab(\varphi)), p,q,x,y\in\stab(\varphi)$
		and choose $t_\ell$ with $\varphi(t_\ell)=\ell$.
	It is easy to verify the following:
    \begin{align}
        t_et_pt_x^{N!}t_qt_f &\approx_k^t t_et_pt_x^{N!+1}t_qt_f\\
        t_et_pt_xt_yt_qt_f &\approx_k^t t_et_pt_yt_xt_qt_f
    \end{align}
    Therefore $\stab(\varphi)$ satisfies the identities of $\EACom$.
\end{samepage}
\medskip

\begin{samepage}
\noindent 
	$(2) \Longrightarrow  (3)$\quad Each equivalence class of $\approx_k^t$ is either a singleton 
		set, or a set of the form  $uK_wv$ where $u,v \in \Sigma^{k}$ and $K_w=\{w'\mid w' \calC_k^t w\}$ for some $w\in \Sigma^{\geq k}$. To prove the direction it suffices to consider 
		those $w\in\Sigma^{\geq k}$ such that $|w|\not\equiv_k 0$.
        
For a word $u\in \Sigma^+$, let $\nu_u \in [0,t]^{[0,k-1]\times \Sigma}$ denote the vector 
\begin{align}
\nu_u = \left (\min \left (\left | \pi_{i,k}(u)\right|_a,t \right) \right)_{i\in [0,k-1], a\in \Sigma}.
\end{align}

For vectors $\nu,\nu',\nu''  \in [0,t]^{[0,k-1]\times \Sigma}$, we write $\nu+\nu'=\nu''$ if $\nu_{i,a}+\nu_{i,a}' \equiv^t \nu_{i,a}''$ for all $i\in [0,k-1]$ and $a\in \Sigma$. 

Let $w\in \Sigma^{\geq k}$ be such that $|w| \not \equiv_k 0$. Let $r\in [1,k-1]$ be such that $|w| \equiv_k r$. We write $K_w$ as a disjoint union $K_w'$ of languages by performing a case analysis
on the last $r$ letters, and the number of times they appear in their respective residue classes previously:
\begin{align}
    K_w'\quad =\quad
    \biguplus_{x \in \Sigma^r}\ \biguplus_{\nu + \nu_x = \nu_w}
	 \mathrm{Count}_{k}^t(\nu) \cdot x.
\end{align}
		We claim that $K_w=K_w'$.
To show the right-to-left inclusion assume some $w'\in K_w'$.
		Then $|w'|\equiv_k r$  and $\nu_{w'}=\nu_w$. 
		Hence $w'\in K_w$.
Conversely assume that $w'\in K_w$. 
		Then $w'=vx$
		for words $v\in\Sigma^{k\N_{>0}}$ and $x\in\Sigma^r$.
		Since $\nu_{w'}=\nu_{v}+\nu_{x}$, in fact we have  $v\in\mathrm{Count}_k^t(\nu_v)$
		and we obtain that $w'=vx\in K_w'$, as required.

Therefore we can write $uK_wv$ as a disjoint union of languages
		of the form $uK''xv$, where $K''$ is a language
		of the form $\mathrm{Count}_k^t(\nu)$ for some  $\nu \in [0,t]^{[0,k-1]\times \Sigma}$.
	\end{samepage}
    \medskip

\begin{samepage}
\noindent 
	$(3) \Longrightarrow  (4)$\quad Follows from Remark~\ref{R Acom Count}. Note that for $k>0$
\begin{align}
    \Sigma^{k\bbN_{>0}} &= \bigcup_{(i,a) \in [0,k-1]\times \Sigma } \mathrm{Thr}_{i,k}(a,1).
\end{align}
    \end{samepage}

\begin{samepage}
\noindent 
	$(4) \Longrightarrow  (5)$\quad We repeat the proof of the corresponding direction in the proof of Theorem~\ref{T const main}. It suffices to show that there is a formula defining sets of the form 
$L=a_0\ldots a_{m-1} \mathrm{Thr}_{i,k}(a,t)b_0\ldots b_{n-1}$ where $m,n,k,t\in \bbN_{>0}$ and $i\in [0,k-1]$.

Let $i'=(m+i)\bmod k\in[0,k-1]$. We define
\begin{align}
	\varphi_L&\defeq \bigwedge_{i\in[0,m-1]} \exists x\, (\min+i=x \wedge P_{a_i}(x)) \wedge \bigwedge_{i\in[0,n-1]} \exists x\, (x+i = \max \wedge P_{b_{n-1-i}}(x))\label{E Acom formula 1}\\
	&~~\wedge  \exists^{\geq t} x\,(x \in [\min + m, \max - n] \wedge x \equiv i' \bmod k \wedge  P_{a}(x))\label{E Acom formula 2}\\
	&~~\wedge\max\equiv m+n-1\bmod k\ \wedge
	\exists x\, (x\in[\min + m, \max - n])\label{E Acom formula 3}
\end{align}
As before (\ref{E Acom formula 1}) expresses that the word begins
with $a_0\dots a_{m-1}$ and ends with $b_0\dots b_{n-1}$ and 
the conjunct (\ref{E Acom formula 2}) expresses that there are at least $t$ positions labelled $a$ that are neither among the first $m$ nor the last $n$ and are $i'\pmod{k}$.
The part (\ref{E Acom formula 3}) expresses that the length of word is 
in residue class $m+n\bmod k$ and that the word has length at least $m+n+1$.
\medskip
	\end{samepage}
    \medskip

\begin{samepage}
\noindent 
	$(5) \Longrightarrow  (2)$\quad By Proposition~\ref{P Logic-to-cong Acomm} proved in \S\,\ref{S Logic Acomm}.
	\end{samepage}
    \medskip

\begin{samepage}
\noindent 
	$(5) \Longrightarrow  (6)$\quad Trivial.
	\end{samepage}
    \medskip

    \begin{samepage}
\noindent 
	$(6) \Longrightarrow  (7)$\quad Trivial.
	\end{samepage}
\medskip

 \begin{samepage}
\noindent $(7) \Longrightarrow  (5)$\quad By Proposition~\ref{P Logic Acomm} proved in \S\,\ref{S Logic Acomm}.
\end{samepage}
\medskip

\end{proof}

\ificalp
\else
\begin{samepage}
	\begin{proposition}\label{P QEACom log}
    If $\eta_L \in \QEACom$ then $\opc(L) \in O(\log n)$.
\end{proposition}
\begin{proof}
If $\eta_L \in \QEACom$ then $L$ is a finite Boolean combination of finite languages and languages
of the form $u\mathrm{Thr}_{i,k}(a,t)v$,
where $k,t\in\bbN_{>0}$, $i \in [0,k-1]$,
$a \in\Sigma$, and $u,v\in\Sigma^+$. Since each language of the form $u\mathrm{Thr}_{i,k}(a,t)v$ is computed by a circuit of size $O(\log n)$ \cite{HASTAD1994200,Friedman86}, the proposition follows.
\end{proof}
\end{samepage}

The following example shows that there are languages in $\SIZE(O(\log n))$ that are not in $\QEACom$. 
\begin{example}
    The language $L=c^*ac^*bc^*$ can be computed by a circuit family of size $O(\log n)$. It is easy to see that $\eta_L\not \in \QEACom$.
\end{example}
\fi

\subsubsection{From Logic to Congruence}
\label{S Logic Acomm}

Let $\mathrm{arb}_1$ denote the class of all arbitrary \emph{unary} numerical predicates. Similarly, let $\mathrm{reg}_1$  denote the class of all regular unary numerical predicates, i.e., the class of unary predicates expressible as a Boolean combination of the set of unary predicates $C_{\min+i}$ and $C_{\max-i}$ for each $i\in \bbN$ that denote the positions $\min+i$ and $\max-i$, $L_{r,k}$ that is true on all positions if and only if the last position is in residue class $r \bmod k$, and  
the unary predicates $\{x \equiv r \bmod k \mid k \in \bbN_{>0}, r \in \bbN\}$ \cite{Straubingbook}. 
We use the following theorem.

\begin{theorem}[Fijalkow-Paperman \cite{FijalkowPaperman}]
  Assume that $\varphi \in \mathrm{MSO}[\Sigma,<,\mathrm{arb}_1]$ is a formula
  using the predicates $R_1,\ldots, R_\ell \in \mathrm{arb}_1$ for $\ell \in \bbN$. If $\varphi$ defines a regular language then there exist regular unary numerical predicates $P_1,\ldots, P_\ell \in \mathrm{reg}_1$ such that the formula $\varphi' \in \mathrm{MSO}[\Sigma,<,\mathrm{reg}_1]$ obtained by substituting $R_i$ by $P_i$ for each $i\in [1,\ell]$ defines the same language.
  \label{T FijalkowPaperman}

\end{theorem}

\begin{proposition}
Every regular language defined by a formula in $\rmC^1[\Sigma, \mathrm{arb}]$ is definable in $\rmC^1[\Sigma, +\omega, \min, \max, \mod]$.
\label{P Logic Acomm}
\end{proposition}
\begin{proof}
Assume that $\varphi \in \rmC^1[\Sigma, \mathrm{arb}]$ defines a regular language. Since the logic uses only one variable, each occurrence of an arbitrary numerical predicate with arity $>1$ in $\varphi$ can be replaced by an arbitrary unary predicate. For instance every occurrence of $R(\min, x, \max, x)$ can be replaced by a predicate $R'(x)$ such that $R'(i) \Longleftrightarrow R(0,i,n-1,i)$ on a word of length $n$. Likewise, the occurrence $R(\max, x, \min, x)$ can be replaced by another predicate $R''(x)$ such that $R''(i) \Longleftrightarrow R(n-1,i,0,i)$. If the occurrence of a predicate does not contain $x$, for example $R(\min, \min, \max, \max)$, then replace that occurrence by the formula $\exists x\, R'''(x)$, where $R'''(x)$ is the unary predicate satisfying $R'''(i) \Longleftrightarrow R(0, 0, n-1, n-1)$ on a word of length $n$. All occurrences of the letter predicates with constants can be rewritten as $P_a(c) \equiv \exists x \,( Q_c(x) \wedge P_a(x))$ where $Q_c$ is the unary predicate that is true precisely at the positions $c$.

Let $\varphi'$ be the formula obtained by applying the above translation. By rewriting each quantification $\exists^{\geq t} x\, \xi(x)$ in $\varphi'$ by the FO formula 
\begin{align}
\exists x_0\ldots \exists x_{t-1} \left (\bigwedge_{i,j\in [0,t-1]\atop i\neq j} x_i \neq x_j \wedge \bigwedge_{i\in [0,t-1]} T(\xi)(x_i)\right)
\end{align}
where \(T\) is defined recursively and $T(\xi)(x_i)$ denotes the translation of $\xi(x)$ with its unique free variable $x$ replaced by $x_i$,
we obtain an equivalent formula $\varphi'' \in \mathrm{FO}[\Sigma,<,\mathrm{arb}_1]$. 
Since $\varphi''$ defines a regular language,
we apply Theorem~\ref{T FijalkowPaperman} and replace each unary numerical predicate in $\varphi''$ by a regular numerical predicate. Since the form of the formula remains unchanged, we reverse the translation  to obtain an equivalent formula in $\rmC^1[\Sigma, +\omega, \min, \max, \rmmod]$. 
\end{proof}

\begin{proposition}
Every language defined by a formula in $\rmC^1[\Sigma, +\omega, \min, \max, \rmmod]$ is a union of $\approx_k^t$-classes for some $k,t\in \bbN_{>0}$
\label{P Logic-to-cong Acomm}
\end{proposition}
\begin{proof}
Let $\varphi$ be a formula in $\rmC^1[\Sigma,+\omega,\min,\max,\rmmod]$. Without loss of generality, assume that $\varphi$ is a Boolean combination of formulas of the form
$\exists^{\geq n}\xi$, where $\xi$ is quantifier-free. Let $k$ be a multiple of all moduli occurring in $\varphi$ and larger than all parameters occurring in counting quantifiers and in atomic formulas such as $\min+i=x$. Set $t=k$. It suffices to show that, for every subformula
$\exists^{\geq n}\xi$ of $\varphi$, its language is a union of $\approx_k^t$-classes.

Replacing negative letter and modular predicates by disjunctions of positive ones, we may assume that the only negative literals in $\xi$ are of the form $c+i\neq x$, $c-i\neq x$ or $c_1+i\neq c_2$. Put $\xi$ in disjunctive normal form. Every satisfiable conjunct can be written as
\begin{align}
o(x)\wedge\sigma(x)\wedge o'\wedge\sigma',
\end{align}
where $o(x)$ contains the conditions on the position $x$, $\sigma(x)$ is a letter predicate on $x$, $o'$ contains conditions involving only $\min,\max$, and $\sigma'$ contains letter predicates on constants.

The conditions in $o(x)$ involving fixed positions concern only the first and last $k$ positions. Hence, for positions $i$ with $k\leq i<|w|-k$
the truth of $o(i)\wedge\sigma(i)$ depends only on the residue of $i$ modulo $k$ and the letter occurring at $i$. Thus, for each residue $r<k$ and $a\in\Sigma$, all interior positions with residue $r$ carrying the letter $a$ have the same truth value for $\xi$.

We claim that $ u\models\exists^{\geq n}x\,\xi(x)
\Longleftrightarrow
v\models\exists^{\geq n}x\,\xi(x)$ for all words $u,v \in \Sigma^+$ such that $u\approx_k^t v$. If $|u|,|v|\leq 2k$, then $u=v$ and the claim holds trivially. Hence, assume that $|u|,|v| > 2k$. The prefix and suffix of $u$ and $v$ of length $k$ are identical, and therefore contribute equally to the number of positions satisfying $\xi$. 
Let $b$ be the number of positions satisfying 
$\xi$ in the prefix and suffix of length $k$ of $u$, and equivalently in the prefix and suffix of length $k$ of $v$.
Let $n_u$ and $n_v$ denote the numbers of interior positions
satisfying $\xi$ in $\kappa_k(u)$ and $\kappa_k(v)$, respectively.
For each $r \in [0,k-1]$ and $a\in\Sigma$, either all interior positions of residue $r$ carrying $a$ satisfy $\xi$, or none of them do. Moreover, $\left|\pi_{r,k}(\kappa_k(u))\right|_a \equiv^t
\left|\pi_{r,k}(\kappa_k(v))\right|_a$
for every $r\in [0,k-1]$ and $a\in\Sigma$, since
$\kappa_k(u)\mathcal C_k^t\kappa_k(v)$.
Summing over precisely those pairs $(r,a)$ for which the
corresponding interior positions satisfy $\xi$, we obtain
$n_u\equiv^t n_v$.
The total numbers of satisfying positions in $u$ and $v$ are,
respectively, $b+n_u$ and $b+n_v$. Since $n\leq t$ by the choice of $t$, we have $n-b\leq t$. If $n_u<t$, then $n_u\equiv^t n_v$ implies $n_u=n_v$, and hence $b+n_u\geq n \Longleftrightarrow
b+n_v\geq n$. If $n_u\geq t$, then $n_v\geq t$ as well, and therefore
$b+n_u\geq t\geq n$, so both $b+n_u$ and $b+n_v$ are at least $n$.
Thus, in all cases, $b+n_u\geq n \Longleftrightarrow 
b+n_v\geq n$. Hence $u\models\exists^{\geq n}x,\xi(x)
\Longleftrightarrow v\models\exists^{\geq n}x,\xi(x)$.
Therefore the language defined by every subformula
$\exists^{\geq n}x,\xi(x)$ is a union of $\approx_k^t$-classes.
Since Boolean combinations of unions of $\approx_k^t$-classes are
again unions of $\approx_k^t$-classes, the language defined by
$\varphi$ has the required form.
\end{proof}

\subsubsection{From Algebra to Congruence}
\label{S Congruence Acomm}

We repeat the proof of Proposition~\ref{P const tech} with necessary changes.
Note that Lemma~\ref{lemma:idem-insert}
and Lemma~\ref{L pq idem} do not require any modification.

\begin{lemma}
	Let $\varphi : \Sigma^+ \to S$ be a stamp in $\QEACom$ with stability index $s\in\bbN_{>0}$. Let $N=|\stab(\varphi)|$.
	Let $u,v \in \Sigma^{s\bbN_{>0}}$ be  words such that 
	$u \congC_s^{N!} v$. 
Then $pup=_\varphi pvp$
for all idempotents $p\in \Sigma^s$.
\label{lemma:acom-pup=pvp}
\end{lemma}
\begin{proof}
	Let $u=u_0\ldots u_{m-1}$ 
	and $v=v_0\ldots v_{n-1}$ be words in $\Sigma^{s\bbN_{>0}}$ such that 
	$u \congC_s^{N!} v$.
    
Let $p\in \Sigma^s$ be idempotent, and for each $i\in\bbN$ and $a\in \Sigma$, let $e_i$, $r_i$, $\ell_i$, $\lambda_i(a)$, and $\lambda'_i(a)$ be as in the proof of Lemma~\ref{lemma:pup=pvp}.
We have $pup=_\varphi pu'p$ and $pvp=_\varphi pv'p$,
	where $u'=u_0e_0u_1e_1\ldots u_{m-1}e_{m-1}$ and  
    $v'=v_0e_0v_1e_1\ldots v_{n-1}e_{n-1}$
by \Cref{lemma:idem-insert}. As in the proof of Lemma~\ref{lemma:pup=pvp}, we observe that
\begin{align}
	pu'&=  \Big( \prod_{i=0}^{m-1} \lambda'_i(u_i) p \Big) p\\
     pv'&=  \Big( \prod_{i=0}^{n-1} \lambda'_i(v_i) p \Big) p
\end{align}

Since $u\congC_s^{N!} v$ we have $|\pi_{i,s}(u)|_{a} \equiv^{N!} |\pi_{i,s}(v)|_{a}$ for all $i\in [0,s-1]$ and $a\in \Sigma$. Therefore for each $x=x_0\ldots x_{s-1}\in\Sigma^s$ such that $x \neq p$,
\begin{align}
    |\block_s(pu'p)|_{x} 
    &=\phantom{{N!}} \sum_{i \in [0,s-1]} |\{ j\equiv_s i\mid j \in [0,m-1], \lambda_i(u_j) = x\}|\\
	&=\phantom{{N!}} \sum_{i \in [0,s-1], a \in \Sigma \atop \lambda_i(a) = x} |\pi_{i,s}(u)|_{a}\\     
	&\equiv^{N!} \sum_{i \in [0,s-1], a \in \Sigma \atop \lambda_i(a) = x} |\pi_{i,s}(v)|_{a}\\ 
    &=\phantom{{N!}} \sum_{i \in [0,s-1]} |\{ j\equiv_s i\mid j \in [0,n-1], \lambda_i(v_j) = x\}|\\
    &=\phantom{{N!}} |\block_s(pv'p)|_{x} \label{E acom-cont-pvp}
\end{align}

Finally we have
	\begin{align}
		\varphi(pup)&=\varphi(pu'p)\\
        &=\pi\circ \sigma_\varphi(\block_s(pu'p)) &&(\text{By (\ref{E pi-sigma})})\\
        &=\pi\circ \sigma_\varphi(\block_s(pv'p)) &&(\text{By (\ref{E acom-cont-pvp}) \& Lemma~\ref{L cont EACom}})\\
        &=\varphi(pv'p)\\
        &=\varphi(pvp).
	\end{align}
\end{proof}

\begin{samepage}
    
\begin{lemma}
	Let $\varphi:\Sigma^+\to S$ be a stamp in $\QEACom$ with stability index $s\in\bbN_{>0}$. 
	Let $N=|\stab(\varphi)|$ and $k=sN$. 
	For each $u,v \in \Sigma^{+}$ we have that 
	if $u \approx_{k}^{N!} v$ and $|u|\equiv_k |v| \equiv_k 0 $, 
	then $u=_\varphi v$. 
\label{lemma:acom-alg-to-cong}
\end{lemma}
\end{samepage}

\begin{proof}
	As in the proof of Lemma~\ref{lemma:alg-to-cong} it suffices to consider words $u,v\in \Sigma^+$ such that
	$u \approx_k^{N!} v$ and $|u|\equiv_k |v|\equiv_k 0 $ and $|u|,|v|> 2k$.
Analogously as in the proof of Lemma~\ref{lemma:alg-to-cong} there exist  factorizations $u=xu'y$, $v=xv'y$ and idempotents $p,q\in \Sigma^s$ such that $u=_\varphi xpu'qy$ and
$v=_\varphi xpv'qy$, where $|p|,|u'|, |v'|, |q|$ are all multiples of $s$
and where $|x|,|y|\leq k$.
Thus, $\varphi(p),\varphi(u'),\varphi(v'),\varphi(q) \in \stab(\varphi)$. 
	Hence $pu'q=_\varphi pu'pq$, and $pv'q=_\varphi pv'pq$
	by Point (1) of \Cref{E Acom-exyf}. 
	Consequently, $u=_\varphi xpu'qy=_\varphi xpu'pqy$ 
	and $v=_\varphi xpv'qy=_\varphi xpv'pqy$.
Thus, to prove the lemma, it suffices to show that 
	$pu'p=_\varphi pv'p$. 
Since $u\approx_k^{N!} v$ implies $\kappa_k(u)\congC_k^{N!} \kappa_k(v)$ and since $\congC_k^{N!}$ is a congruence 
	by Point (1) of \Cref{L cont} we have 
	$u'=x'\kappa_k(u)y' \congC_k^{N!} x'\kappa_k(v)y' = v'$, where 
    $x'=x^{-1}\alpha_{k}(u)=x^{-1}\alpha_{k}(v)$ and
$y'=\omega_{k}(u)y^{-1}=\omega_{k}(v)y^{-1}$. 
	Since $s|k$ we have $u'\congC_s^{N!} v'$ by Point (2) of Lemma~\ref{L cont}, and since the lengths of $u'$ and $v'$ are both divisible by
$s$ we have
	$pu'p=_\varphi pv'p$ by \Cref{lemma:acom-pup=pvp}.
\end{proof}

\begin{proposition}\label{P log tech}
	Let $\varphi:\Sigma^+\to S$ be a stamp in $\QEACom$ with stability index $s\in\bbN_{>0}$. 
	Let $N=|\stab(\varphi)|$.
	For each $u,v \in \Sigma^{+}$ we have that 
	if $u \approx_{2sN}^{N!} v$, then $\varphi(u)=\varphi(v)$. 
\end{proposition}
\begin{proof}
Let $k=sN$.
    By \Cref{lemma:acom-alg-to-cong} 
	for all $u,v\in\Sigma^+$ we have that
	if $u\approx_{k}^{N!} v$ 
	and $|u|\equiv_{k} |v| \equiv_{k} 0$, then $u=_\varphi v$. 
	We claim that for all $u,v \in \Sigma^+$, if $u \approx_{2k}^{N!} v$,
	then $u=_\varphi v$.
Without loss of generality assume that $|u|,|v|>4k$. Since $u\approx_{2k}^{N!}v$, 
	clearly $|u|\equiv_{2k} |v|$. 
	Let $u',v'\in \Sigma^+$, $x\in \Sigma^*$ be such that
	$u=u'x$, $v=v'x$, $|u'|\equiv_{k} |v'|\equiv_{k} 0 $, $|x|<k$. 
It suffices to 
	show that $u'=_\varphi v'$ for which by \Cref{lemma:acom-alg-to-cong} it suffices to prove $u'\approx_{k}^{N!} v'$ .
We argue $\alpha_{k}(u')=
	\alpha_{k}(v')$ and 
	$\omega_{k}(u')
	=\omega_{k}(v')$ as in the proof of Proposition~\ref{P const tech}.
To prove $\kappa_{k}(u')\congC_{k}^{N!} \kappa_{k}(v')$, observe that 
\begin{align}
	\kappa_{k}(u')&=(\omega_k\circ \alpha_{2k}(u))\cdot \kappa_{2k}(u) \cdot 
	(\alpha_{k-|x|}\circ \omega_{2k}(u))\label{E acom ao1}\\
	\kappa_{k}(v')
	&=(\omega_k\circ \alpha_{2k}(u))\cdot \kappa_{2k}(v) \cdot 
	(\alpha_{k-|x|}\circ \omega_{2k}(u))\label{E acom ao3}
\end{align}
    
	Since $u\approx_{2k}^{N!}v$ we have $\kappa_{2k}(u)\congC_{2k}^{N!} \kappa_{2k}(v)$. 
	By \Cref{lemma:prop-cont}
	we have $\kappa_{2k}(u)\congC_{k}^{N!} \kappa_{2k}(v)$ and 
	since $\congC_{k}^{N!}$ is a congruence
	we have
	$\kappa_{k}(u')\congC_{k}^{N!} \kappa_{k}(v')$ by (\ref{E acom ao1}) and (\ref{E acom ao3}), 
	as required.
\end{proof}

\subsection{Comparison with Other Input Encodings}
\ificalp
\else
Let $\Sigma=\{a_0,\dots,a_{k-1}\}$ be an alphabet.
The {\em one-hot encoding}
is given by the morphism $\varphi:\Sigma^+\to\{0,1\}^+$,
where $\varphi(a_i)=0^i10^{k-i-1}$ for each $i\in[0,k-1]$.

A \emph{one-hot circuit family computing a language $L\subseteq\Sigma^+$}
is a circuit family $\calC$ over the input alphabet $\{0,1\}$ such that 
$L(\calC)=\varphi(L)$.
In contrast, a \emph{one-hot circuit family with promise computing
a language $L\subseteq\Sigma^+$}
is a circuit family $\calC$ over the input alphabet $\{0,1\}$ such that 
$\varphi^{-1}(L(\calC))=L$.
\fi

\begin{propositionrep}\label{P with promise}
Our subset encoding model of Definition~\ref{D Circuit} is equivalent to the
	one-hot encoding with promise in the following sense.
For all languages $L\subseteq\Sigma^+$ and all functions $f:\bbN_{>0}\to\bbN_{>0}$
the following two statements are equivalent:
	\begin{enumerate}
		\item There is a one-hot circuit family with promise $\calC$ computing $L$,
			where $|\calC|(kn)=f(n)$ for all $n\in\bbN_{>0}$.
		\item There is a circuit family $\calC$ computing $L$,
			where $|\calC|(n)=f(n)$ for all $n\in\bbN_{>0}$.
	\end{enumerate}
\end{propositionrep}
\begin{appendixproof}

\noindent 
$(1) \Longrightarrow  (2)$\quad
	Let $\calC=(C_n)_{n\in\bbN_{>0}}$ 
	be a one-hot circuit family with promise computing a language
	$L\subseteq\Sigma^+$, where $|\calC|(kn)=f(n)$ for all $n\in\bbN_{>0}$.
	We construct
	a circuit family $\calC'=(C_n')_{n\in\bbN_{>0}}$ over the input 
	alphabet $\Sigma$ such that $L(\calC')=\varphi^{-1}(L(\calC))$
	and $|C_{n}'|=f(n)$ for all $n\in\bbN_{>0}$.

	For this it suffices to show that for all $n\in\bbN_{>0}$ 
	the $kn$-circuit $C_{kn}=(G,\prec,\lambda,g_{out})$ over 
	the alphabet $\{0,1\}$ can be turned into an $n$-circuit $C_n'$ over 
	the alphabet $\Sigma$
	such that $|C_n'|=|C_{kn}|$ and moreover
	$w\in L(C_n')$ if, and only if, $\varphi(w)\in L(C_{kn})$
	for all $w\in\Sigma^n$.

	Without loss of generality we may assume that every
	input gate $h$ with $\lambda(h)=(i,\Gamma)$ and $h\prec g$ for
	some $g\in G_1$ satisfies $|\Gamma|=1$.
	Indeed, the following process establishes this normal form
	without increasing the number of non-input gates.
	Assume there were to exist a wire $(h,g)\in\prec$, where 
$\lambda(h)=(i,\Gamma)$, $|\Gamma|\not=1$ and 
	$\lambda(g)=\wedge$, then one could proceed as follows.
	Let us first treat the case when $|\Gamma|=2$.
	In this case we have $\Gamma=\{0,1\}$ and one could tacitly
	remove the wire $(h,g)$ from $\prec$ since $g$ is an $\wedge$-gate;
	in case this leaves $g$ without predecessor, we 
	make $g$ a gate that evaluates to constant $1$ by
	relabeling it as $\lambda(g)=\vee$ and connecting it to two 
	input gates $h_0,h_1$ with $\lambda(h_0)=(0,\{0\})$ and $\lambda(h_1)=(0,\{1\})$.
	Let us treat the case when $|\Gamma|=0$. 
	In this case we remove the wire $(h,g)$ and then
	make $g$ a gate that evaluates to constant $0$
	by connecting it to two input gates $h_0,h_1$ with 
	$\lambda(h_0)=(0,\{0\})$ and $\lambda(h_1)=(0,\{1\})$.
	
	Analogous remarks apply to wires $(h,g)$, where $h$ is an input gate
	and $\lambda(g)=\vee$.

	The circuit $C_n'=(G',\prec',\lambda',g_{out})$ is obtained from 
	$C_{kn}$ by keeping 
	the set of non-input gates $G_1$ (so $G_1'=G_1$) and by introducing
	wires that emerge from the following replacements:
	\begin{enumerate}
		\item Replace every wire $(h,g)\in\prec$, where $\lambda(h)=(i,\{0\})$ and
			$\lambda(g)=\vee$, by $(h',g)\in\prec'$, such that 
			$h'$ is an input gate satisfying
			$\lambda'(h')=(\lfloor i/ k\rfloor,\Gamma)$, where
			$\Gamma=\Sigma\setminus\{a_{i\bmod{ k}}\}$.
		\item Replace every wire $(h,g)\in\prec$, where $\lambda(h)=(i,\{1\})$ and
			$\lambda(g)=\vee$, by $(h',g)\in\prec'$, such that
			$h'$ is an input gate satisfying
			$\lambda'(h')=(\lfloor i/ k\rfloor,\{a_{i\bmod{ k}}\})$.
		\item Replace every wire $(h,g)\in\prec$, 
			where $\lambda(h)=(i,\{1\})$ and
			$\lambda(g)=\wedge$, by $(h',g)\in\prec'$, such that
			$h'$ is an input gate satisfying
			$\lambda'(h')=(\lfloor i/ k\rfloor,\{a_{i\bmod{ k}}\})$.
		\item Replace every wire $(h,g)\in\prec$, 
			where $\lambda(h)=(i,\{0\})$ and
			$\lambda(g)=\wedge$, by $(h',g)\in\prec'$, such that
			$h'$ is an input gate satisfying
			$\lambda'(h')=(\lfloor i/ k\rfloor,\Gamma)$, where
			$\Gamma=\Sigma\setminus\{a_{i\bmod{ k}}\}$.
	\end{enumerate}
	Obviously $w\in L(C_n')$ if, and only if, $\varphi(w)\in L(C_{kn})$ for all
	$w\in\Sigma^n$. Moreover $|C_n'|=|C_{kn}|=f(n)$.
\medskip

\noindent 
$(2) \Longrightarrow  (1)$\quad
	Let $\calC=(C_n)_{n\in\bbN_{>0}}$ 
	be a circuit family computing a language
	$L\subseteq\Sigma^+$, where $|\calC|(n)=f(n)$ for all $n\in\bbN_{>0}$.
	We construct
	a one-hot circuit family with promise $\calC'=(C_n')_{n\in\bbN_{>0}}$ 
	computing $L$ such that 
	$|C_{kn}'|=f(n)$ for all $n\in\bbN_{>0}$.

For this it suffices to show that for all $n\in\bbN_{>0}$
the $n$-circuit $C_n=(G,\prec,\lambda,g_{out})$ over $\Sigma$
can be turned into a $kn$-circuit $C_{kn}'$ over $\{0,1\}$
such that $|C_{kn}'|=|C_n|$ and
		$\varphi(w)\in L(C_{kn}')$ 
		if, and only if, $w\in L(C_n)$ for all
	$w\in\Sigma^n$. 
	The $kn$-circuit $C_{kn}'=(G',\prec',\lambda',g_{out})$ is obtained
	from $C_n$ by keeping 
	the set of non-input gates $G_1$ (so $G_1'=G_1$) and by introducing
	wires that emerge from the following replacements:
	\begin{enumerate}
		\item Replace each wire $(h,g)\in\prec$, where 
			$\lambda(h)=(i,\Gamma)$ and $\lambda(g)=\vee$,
			by the set of wires $\{(h_{j},g)\mid a_j\in\Gamma\}$
			in $\prec'$ such that 
			$\lambda'(h_j)=(ik+j,\{1\})$ for each $a_j\in\Gamma$ in
			case $\Gamma\not=\emptyset$; and
			by the wire $(h_0,g)$ in $\prec'$ such that
			$\lambda'(h_0)=(0,\emptyset)$ in case $\Gamma=\emptyset$.
		\item Replace each wire $(h,g)\in\prec$, where 
			$\lambda(h)=(i,\Gamma)$ and $\lambda(g)=\wedge$,
			by the set of wires $\{(h_{j},g)\mid a_j\in\Sigma\setminus\Gamma\}$
			in $\prec'$ such that 
			$\lambda'(h_j)=(ik+j,\{0\})$ for each 
			$a_j\in\Sigma\setminus\Gamma$
in case $\Gamma\not=\Sigma$;
 and by the wire $(h_0,g)$ in $\prec'$ such that
			$\lambda'(h_0)=(0,\{0,1\})$ in case $\Gamma=\Sigma$.
	\end{enumerate}
	Clearly, $|C_{kn}'|=|C_n|$ and moreover $\varphi(w)\in L(C_{kn}')$ 
	if, and only if, 
	$w\in L(C_n)$.
\end{appendixproof}
The following corollary states that
the circuit complexity of a regular language
under the one-hot encoding is an upper bound for
the circuit complexity under the subset encoding
up to scaling of input lengths.

\begin{samepage}
\begin{corollary}
	Let $f:\N_{>0}\to\N_{>0}$ be a function and $L\subseteq\Sigma^+$ be a regular
	language with $|\Sigma|=k$.
	Then $\opc(\varphi(L))(kn)\leq f(n)$ implies 
	$\opc(L)(n)\leq f(n)$ for all $n\in\bbN_{>0}$.
\end{corollary}
	\begin{proof}
Let $\calC$ be a circuit family computing $\varphi(L)$, where
		$|\calC|(kn)\leq f(n)$.
Then $L(\calC)=\varphi(L)$, so 
$\varphi^{-1}(L(\calC))=\varphi^{-1}(\varphi(L))=L$.
Thus, in particular $\calC$ is a one-hot circuit family
with promise computing $L$. Hence
there is a circuit family $\calC'$ computing $L$
		with $|\calC'|(n)\leq f(n)$ 
		for all $n\in\bbN_{>0}$ by 
		Proposition~\ref{P with promise}.
	\end{proof}
\end{samepage}

Using our lower bound results from Section~\ref{S Lower Bound const} 
the following proposition tells us that in fact the one-hot encoding
(without promise) is not a robust circuit model.
\begin{samepage}
	\begin{proposition}\label{P encoding}
	Let $\Sigma=\{a,b\}$ and let $L=\Sigma^+$. Then the following holds:
	\begin{enumerate}
		\item $\opc(\varphi(L))\in\Omega(\log n)$, whereas 
			$\opc(\varphi(\Sigma^+\setminus L))\in O(1)$.
\item Let $\alpha:\Sigma^+\to\{b\}^+$ denote the length-multiplying
	morphism, where $\alpha(a)=\alpha(b)=b$.
			Then $\opc(\varphi(\{b\}^+))\in O(1)$,
			whereas 
			$\opc(\varphi(\alpha^{-1}(\{b\}^+)))\in \Omega(\log n)$.
	\end{enumerate}
	Thus, the class of regular languages of constant circuit
	complexity under the one-hot encoding (without promise) is neither closed 
	under complement
	nor under inverse length-multiplying morphisms, and 
	hence is in particular
	not an lm-variety of languages (see Section~\ref{S Circuit LM}
	for more details).
\end{proposition}
\end{samepage}
\begin{proof}
	Let $K=\varphi(L)=\{01,10\}^+$.
	Let us first prove Point (1).
	Clearly 
	$\opc(\varphi(\Sigma^+\setminus L))=
	\opc(\emptyset)\in O(1)$.
	We refer to \S\,\ref{S Main Result const} for 
	further details on the languages 
	$\Cont_k(\Gamma_0,\dots,\Gamma_{k-1})$ (Definition~\ref{D Cont}).
We claim that $K$ is not a disjoint finite union of finite languages
	over $\{0,1\}$
	and languages of the form $u\Cont_k(\Gamma_0,\dots,\Gamma_{k-1})v$,
	where $u,v\in\{0,1\}^+$
	and $\emptyset\not=\Gamma_0,\dots,\Gamma_{k-1}\subseteq\{0,1\}$.
	Towards a contradiction assume there
	exists some $k\in\bbN_{>0}$ such that $K$ is such a disjoint union
	\begin{align}
		K\quad=\quad\biguplus_{i\in I} K_i\ \uplus\ 
		\biguplus_{j\in J}u_j
		\Cont_k(\Gamma_{j,0},\dots,\Gamma_{j,k-1})v_j,
	\end{align}
	where $I$ and $J$ are finite index sets,
	$K_i$ is a finite language for each $i\in I$,
	$u_j,v_j\in\{0,1\}^+$ for each $j\in J$,
	and $\Gamma_{j,0},\dots,\Gamma_{j,k-1}\subseteq\{0,1\}$
	for each $j\in J$.
	We perform case distinction on whether $|\Gamma_{j,r}|=1$
	for all $j\in J$ and all $r\in[0,k-1]$ or not.

	Let us first consider the case when
	$|\Gamma_{j,r}|=1$ for all $j\in J$ and all $r\in[0,k-1]$.
	Then all but finitely many words of $K$ are 
	contained in a finite union of languages all of the form
	$u(x)^+v$ for words $u,x,v\in\{0,1\}^+$.
	But then the function $n\mapsto |K\cap\{0,1\}^n|$ is bounded, 
	contradicting the fact that this function
	grows exponentially: indeed, note that $|K\cap\{0,1\}^{2n}|=2^n$
	for all $n\in\bbN_{>0}$.

	Let us now consider the case 
	$|\Gamma_{j,r}|\not=1$ for some $j\in J$ and some $r\in[0,k-1]$.
	Since $\Gamma_{j,r}\not=\emptyset$ we must have
	$|\Gamma_{j,r}|=2$.
	But then the language
	$u_j\Cont_k(\Gamma_{j,0},\dots,\Gamma_{j,k-1})v_j$
	contains a word $w$ such that another word
	with Hamming distance $1$ is in the language
	as well.
	The same applies to $K$, which is a contradiction
	since any two distinct words in $K$ have Hamming distance
	at least two.
	By Point (5) and Point (2) of Theorem~\ref{T const main}
	we have $\opc(K)\in\Omega(\log n)$.
	This concludes the proof of Point (1).

	Let us now prove Point (2).
	Since $\varphi(\{b\}^+)=(01)^+$ can be computed 
	by a circuit family of constant size that checks on inputs
	of even length if even positions have letter $0$ and that
	odd positions have letter $1$,
	we obtain $\opc(\varphi(\{b\}^+))\in O(1)$.
	On the other hand, we have $\varphi(\alpha^{-1}(\{b\}^+))=\varphi(L)$,
	so $\opc(\varphi(L))\in\Omega(\log n)$, as shown in Point (1).
\end{proof}

\section{Conclusion}\label{S Conclusion}

We have characterized the regular languages of
constant circuit complexity in various ways:
logically in terms of 
$\rmFO^1[\Sigma,\mathrm{arb}]$, 
algebraically
in terms of the languages whose syntactic
morphism is in $\QESL$, in terms of 
unions of classes of the congruence $\approx_k$ for
some $k\in\bbN_{>0}$,
and as finite disjoint unions of finite languages
and languages of the form 
$u\Cont_k(\Gamma_0,\dots,\Gamma_{k-1})v$.
We have characterized the neutral letter regular
languages of constant circuit complexity 
as the ones that are both idempotent and commutative,
logically in terms of $\rmFO^1[\Sigma]$, 
and as finite unions of $\congC_1^1$-classes.

We have shown that it is 
$\textbf{PSPACE}$-complete to decide if the
language of a given nondeterministic finite
automaton has constant circuit complexity.

As a technical tool we have
introduced rational truth-table reductions
and showed that, for a class of functions
we call mild, reducibility implies
the preservation of circuit complexity upper
and lower bounds.
We also showed that the class of regular languages
whose circuit complexity is upper bounded by
mild functions is a length-multiplying variety
of languages.

We have addressed the subtlety of
the choice of the underlying circuit model.
While $\ACO$ circuit
families under the one-hot encoding with or without promise
recognize equivalent languages we showed
that for the regular languages of constant circuit
complexity these models indeed differ:
without promise the respective class is not
a length-multiplying variety of languages.

In subsequent work we plan to shed new
light into the structure of the regular languages
whose circuit complexity is in $\Omega(\log n)$.

\bibliographystyle{plain}
\bibliography{references}

@article{FijalkowPaperman,
author = {Fijalkow, Nathana{\"e}l and Paperman, Charles},
title = {{Monadic Second-Order Logic with Arbitrary Monadic Predicates}},
year = {2017},
issue_date = {July 2017},
publisher = {Association for Computing Machinery},
address = {New York, NY, USA},
volume = {18},
number = {3},
issn = {1529-3785},
url = {https://doi.org/10.1145/3091124},
doi = {10.1145/3091124},
journal = {ACM Trans. Comput. Logic},
month = aug,
articleno = {20},
numpages = {17}
}

@incollection{Kouckyhandbook,
  author       = {Michal Kouck{\'{y}}},
  editor       = {Jean{-}{\'{E}}ric Pin},
  title        = {Circuit complexity of regular languages},
  booktitle    = {Handbook of Automata Theory},
  pages        = {493--523},
  publisher    = {European Mathematical Society Publishing House, Z{\"{u}}rich,
                  Switzerland},
  year         = {2021},
  url          = {https://doi.org/10.4171/Automata-1/14},
  doi          = {10.4171/AUTOMATA-1/14}
}

@book{hanbookauto,
  TITLE = {{Handbook of Automata Theory}},
  AUTHOR = {Pin, Jean-{\'E}ric},
  URL = {https://hal.science/hal-03579131},
  PUBLISHER = {{European Mathematical Society Publishing House}},
  YEAR = {2021},
  MONTH = Sep,
  DOI = {10.4171/Automata},
  HAL_ID = {hal-03579131},
  HAL_VERSION = {v1},
}

@article{GrosshansMS22,
  author       = {Nathan Grosshans and
                  Pierre McKenzie and
                  Luc Segoufin},
  title        = {{Tameness and the power of programs over monoids in {DA}}},
  journal      = {Log. Methods Comput. Sci.},
  volume       = {18},
  number       = {3},
  year         = {2022},
  url          = {https://doi.org/10.46298/lmcs-18(3:14)2022},
  doi          = {10.46298/LMCS-18(3:14)2022}}

@inproceedings{Grosshans21,
  author       = {Nathan Grosshans},
  title        = {{A Note on the Join of Varieties of Monoids with {LI}}},
  booktitle    = {{MFCS} 2021},
  series       = {LIPIcs},
  pages        = {51:1--51:16},
  publisher    = {Schloss Dagstuhl - Leibniz-Zentrum f{\"{u}}r Informatik},
  year         = {2021},
  url          = {https://doi.org/10.4230/LIPIcs.MFCS.2021.51},
  doi          = {10.4230/LIPICS.MFCS.2021.51}}

@article{ConstantDepth1984,
  author  = {Chandra, Ashok K. and Stockmeyer, Larry J. and Vishkin, Uzi},
  title   = {Constant Depth Reducibility},
  journal = {SIAM Journal on Computing},
  volume  = {13},
  number  = {2},
  pages   = {423--439},
  year    = {1984},
  doi     = {10.1137/0213028}
}

@article{CadilhacP22,
  author       = {Micha{\"{e}}l Cadilhac and
                  Charles Paperman},
  title        = {{The regular languages of wire linear AC\({}^{\mbox{0}}\)}},
  journal      = {Acta Informatica},
  volume       = {59},
  number       = {4},
  pages        = {321--336},
  year         = {2022},
  url          = {https://doi.org/10.1007/s00236-022-00432-2},
  doi          = {10.1007/S00236-022-00432-2}
}

@INPROCEEDINGS{Chaubard,
  author={Chaubard, L. and Pin, J. and Straubing, H.},
  booktitle={21st Annual IEEE Symposium on Logic in Computer Science (LICS'06)}, 
  title={{First Order Formulas with Modular Predicates}},
  year={2006},
  volume={},
  number={},
  pages={211-220},
  doi={10.1109/LICS.2006.24}}

@article{VitanyiM84,
  author       = {Paul M. B. Vit{\'{a}}nyi and
                  Lambert G. L. T. Meertens},
  title        = {Big omega versus the wild functions},
  journal      = {Bull. {EATCS}},
  volume       = {22},
  pages        = {14--19},
  year         = {1984},
  bibsource    = {dblp computer science bibliography, https://dblp.org}
}

@InProceedings{Papermanfo2,
  author =	{Dartois, Luc and Paperman, Charles},
  title =	{{Two-variable first order logic with modular predicates over words}},
  booktitle =	{STACS 2013},
  pages =	{329--340},
  series =	{LIPIcs},
  ISBN =	{978-3-939897-50-7},
  ISSN =	{1868-8969},
  year =	{2013},
  volume =	{20},
  address =	{Dagstuhl, Germany},
  URL =		{https://drops.dagstuhl.de/entities/document/10.4230/LIPIcs.STACS.2013.329}}

@article{SurveySmall,
author = {Diekert, Volker and Gastin, Paul and Kufleitner, Manfred},
year = {2008},
month = {06},
pages = {513-548},
title = {{A Survey on Small Fragments of First-Order Logic over Finite Words.}},
volume = {19},
journal = {Int. J. Found. Comput. Sci.},
doi = {10.1142/S0129054108005802}
}

@book{Eilenberg1976,
  author    = {Samuel Eilenberg},
  title     = {{Automata, Languages, and Machines, Volume B}},
  series    = {Pure and Applied Mathematics},
  volume    = {59},
  publisher = {Academic Press},
  year      = {1976},
  isbn      = {978-0-12-234002-4},
  note      = {Includes chapters by Bret Tilson},
}

@InProceedings{Straubing02,
author="Straubing, Howard",
editor="Rajsbaum, Sergio",
title="{On Logical Descriptions of Regular Languages}",
booktitle="LATIN 2002: Theoretical Informatics",
year="2002",
publisher="Springer Berlin Heidelberg",
address="Berlin, Heidelberg",
pages="528--538",
isbn="978-3-540-45995-8"
}

@article{StraubingPin10,
author = {Pin, Jean-\'{E}ric and Straubing, Howard},
journal = {RAIRO - Theoretical Informatics and Applications},
language = {eng},
month = {3},
number = {1},
pages = {239-262},
publisher = {EDP Sciences},
title = {{Some results on {C}-varieties}},
url = {http://eudml.org/doc/92759},
volume = {39},
year = {2005},
}

@article{Kunc10,
author = {Kunc, Michal},
journal = {RAIRO - Theoretical Informatics and Applications},
language = {eng},
month = {3},
number = {3},
pages = {243-254},
publisher = {EDP Sciences},
title = {{Equational description of pseudovarieties of homomorphisms}},
url = {http://eudml.org/doc/92722},
volume = {37},
year = {2003},
}

@article{Friedman86,
  author       = {Joel Friedman},
  title        = {{Constructing O(n log n) Size Monotone Formulae for the k-th Threshold
                  Function of n Boolean Variables}},
  journal      = {{SIAM} J. Comput.},
  volume       = {15},
  number       = {3},
  pages        = {641--654},
  year         = {1986},
  url          = {https://doi.org/10.1137/0215047},
  doi          = {10.1137/0215047}
}

@article{KW15,
  author       = {Manfred Kufleitner and
                  Tobias Walter},
  title        = {{One quantifier alternation in first-order logic with modular predicates}},
  journal      = {{RAIRO} Theor. Informatics Appl.},
  volume       = {49},
  number       = {1},
  pages        = {1--22},
  year         = {2015},
  url          = {https://doi.org/10.1051/ita/2014024},
  doi          = {10.1051/ITA/2014024}
}

@manual{Pin2025Mathematical,
  title        = {{Mathematical Foundations of Automata Theory}},
  author       = {Pin, Jean--{\'E}ric},
  year         = {2025},
  note         = {Lecture notes, MPRI, Universit{\'e} Paris and IRIF},
  url          = {https://www.irif.fr/~jep/PDF/MPRI/MPRI.pdf}
}

@article{Schutz65,
  author  = {Sch{\"u}tzenberger, Marcel-Paul},
  title   = {{On Finite Monoids Having Only Trivial Subgroups}},
  journal = {Information and Control},
  volume  = {8},
  year    = {1965},
  pages   = {190--194}
}

@article{LogicMeetsAlgebra,
    title      = {{Logic Meets Algebra: the Case of Regular Languages}},
    author     = {Pascal Tesson and Denis Therien},
    url        = {https://lmcs.episciences.org/2226},
    doi        = {10.2168/LMCS-3(1:4)2007},
    journal    = {Logical Methods in Computer Science},
    issn       = {1860-5974},
    volume     = {Volume 3, Issue 1},
    eid        = 4,
    year       = {2007},
    month      = {Feb},
}

@article{BridgesTessonTherien,
author = {Tesson, Pascal and Therien, Denis},
year = {2006},
month = {01},
pages = {},
title = {{Bridges between algebraic automata theory and complexity theory}},
journal = {Bulletin of the European Association for Theoretical Computer Science EATCS}
}

@book{McNPap71,
  author    = {McNaughton, Robert and Papert, Seymour},
  title     = {{Counter-Free Automata}},
  publisher = {The MIT Press},
  address   = {Cambridge, MA},
  year      = {1971}
}

@article{McKenzieNC1,
	author = {McKenzie, Pierre and P{\'e}ladeau, Pierre and Therien, Denis},
	date = {1991/12/01},
	doi = {10.1007/BF01212963},
	id = {McKenzie1991},
	isbn = {1420-8954},
	journal = {computational complexity},
	number = {4},
	pages = {330--359},
	title = {{$\mathrm{NC}^1$: The automata-theoretic viewpoint}},
	url = {https://doi.org/10.1007/BF01212963},
	volume = {1},
	year = {1991}}

@article{ActionCvariety,
title = {{Actions, wreath products of {C}-varieties and concatenation product}},
journal = {Theoretical Computer Science},
volume = {356},
number = {1},
pages = {73-89},
year = {2006},
note = {In honour of Professor Christian Choffrut on the occasion of his 60th birthday},
author = {Laura Chaubard and Jean-Éric Pin and Howard Straubing},
issn = {0304-3975}
}

@article{PeladeauTCS,
title = {{Finite semigroup varieties defined by programs}},
journal = {Theoretical Computer Science},
volume = {180},
number = {1},
pages = {325-339},
year = {1997},
issn = {0304-3975},
doi = {https://doi.org/10.1016/S0304-3975(96)00297-6},
author = {Pierre Péladeau and Howard Straubing and Denis Therien},
}

@InProceedings{NathanJ,
author="Grosshans, Nathan",
title="{The Power of Programs over Monoids in {$\bf{J}$}}",
booktitle="Language and Automata Theory and Applications",
year="2020",
publisher="Springer International Publishing",
address="Cham",
pages="315--327"
}

@book{Straubingbook,
  author    = {Straubing, Howard},
  title     = {{Finite Automata, Formal Logic, and Circuit Complexity}},
  publisher = {Birkh{\"a}user},
  address   = {Boston},
  year      = {1994}
}

@article{BCST92,
  author       = {David A. Mix Barrington and
                  Kevin J. Compton and
                  Howard Straubing and
                  Denis Th{\'{e}}rien},
  title        = {{Regular Languages in NC{\({^1}\)}}},
  journal      = {J. Comput. Syst. Sci.},
  volume       = {44},
  number       = {3},
  pages        = {478--499},
  year         = {1992},
  url          = {https://doi.org/10.1016/0022-0000(92)90014-A},
  doi          = {10.1016/0022-0000(92)90014-A}
}

@article{Barrington1989,
  author  = {David A. Barrington},
  title   = {{Bounded-Width Polynomial-Size Branching Programs Recognize Exactly Those Languages in {NC}$^1$}},
  journal = {Journal of Computer and System Sciences},
  volume  = {38},
  number  = {1},
  pages   = {150--164},
  year    = {1989}
}

@book{AroraBarak,
  author = {Arora, Sanjeev and Barak, Boaz},
  isbn = {0521424267},
  publisher = {Cambridge University Press},
  title = {{Computational Complexity: A Modern Approach}},
  url = {http://www.cs.princeton.edu/theory/complexity/},
  year = 2009
}

@article{FurstSaxeSipser1984,
  author  = {Merrick Furst and James B. Saxe and Michael Sipser},
  title   = {{Parity, Circuits, and the Polynomial-Time Hierarchy}},
  journal = {Mathematical Systems Theory},
  volume  = {17},
  number  = {1},
  pages   = {13--27},
  year    = {1984},
  doi     = {10.1007/BF01744431}
}

@article{Koucky2009,
  author  = {Michal Kouck{\'y}},
  title   = {{Circuit Complexity of Regular Languages}},
  journal = {Theory of Computing Systems},
  volume   = {45},
  number   = {4},
  pages    = {865--879},
  year     = {2009},
  doi      = {10.1007/s00224-009-9172-8}
}

@article{Hromkovic85a,
  author       = {Juraj Hromkovic},
  title        = {{Linear Lower Bounds on Unbounded Fan-In Boolean Circuits}},
  journal      = {Inf. Process. Lett.},
  volume       = {21},
  number       = {2},
  pages        = {71--74},
  year         = {1985},
  url          = {https://doi.org/10.1016/0020-0190(85)90035-3},
  doi          = {10.1016/0020-0190(85)90035-3}
}

@article{BarringtonTherien1988,
  author  = {Barrington, David A. Mix and Th{\'e}rien, Denis},
  title   = {{Finite Monoids and the Fine Structure of {$\mathrm{NC}^1$}}},
  journal = {Journal of the ACM},
  volume  = {35},
  number  = {4},
  pages   = {941--952},
  year    = {1988},
  month   = oct,
  doi     = {10.1145/48014.63138}
}

@book{EsparzaBlondin2023,
  author    = {Javier Esparza and Michael Blondin},
  title     = {{Automata Theory: An Algorithmic Approach}},
  publisher = {MIT Press},
  year      = {2023},
  isbn      = {978-0-262-04863-7}
}

@article{HASTAD1994200,
title = {{Optimal Depth, Very Small Size Circuits for Symmetrical Functions in {$\mathrm{AC}^0$}}},
journal = {Information and Computation},
volume = {108},
number = {2},
pages = {200-211},
year = {1994},
issn = {0890-5401},
doi = {https://doi.org/10.1006/inco.1994.1008},
url = {https://www.sciencedirect.com/science/article/pii/S089054018471008X},
author = {J. Hastad and I. Wegener and N. Wurm and S.Z. Yi}
}

\end{document}